\documentclass[11pt]{article}
\usepackage[utf8]{inputenc}

\usepackage{amssymb, amsmath, amsthm, graphicx, authblk, bm, bbm, fullpage, soul}
\usepackage[round]{natbib}
\usepackage{mathtools}
\usepackage{mathrsfs} 
\usepackage{stmaryrd}
\usepackage{multirow}
\usepackage{enumitem}
\usepackage[dvipsnames]{xcolor}

\usepackage{algorithm}
\usepackage{algpseudocode}
\algnewcommand\algorithmicinput{\textbf{INPUT:}}
\algnewcommand\INPUT{\item[\algorithmicinput]}
\algnewcommand\algorithmicoutput{\textbf{OUTPUT:}}
\algnewcommand\OUTPUT{\item[\algorithmicoutput]}

\usepackage{hyperref}[]
\hypersetup{
    colorlinks=true,
    linkcolor=blue,
    filecolor=magenta,      
    urlcolor=cyan,
      citecolor=blue,
    }
\usepackage{cleveref}
\usepackage{booktabs} 

\newtheorem{theorem}{Theorem}
\newtheorem{lemma}[theorem]{Lemma}
\newtheorem{proposition}[theorem]{Proposition}
\newtheorem{corollary}[theorem]{Corollary}

\newtheorem{definition}{Definition}
\newtheorem{remark}{Remark}
\newtheorem{assumption}{Assumption}
\newtheorem{example}{Example}

\allowdisplaybreaks

\DeclareMathOperator*{\argmin}{arg\,min}

\DeclareMathOperator*{\op}{op}

\DeclareMathOperator*{\var}{Var}

\DeclareMathOperator*{\rank}{rank}

\DeclareMathOperator*{\diag}{diag}
\DeclareMathOperator*{\tr}{tr}
\renewcommand{\d}[1]{\ensuremath{\operatorname{d}\!{#1}}}

\title{Multivariate Poisson intensity estimation via  low-rank tensor decomposition}
\author[1]{Haotian Xu}
\author[2]{Carlos Misael Madrid Padilla}
\author[3]{Oscar Hernan Madrid Padilla}
\author[4]{Daren Wang}
\affil[1]{Department of Mathematics and Statistics, Auburn University}
\affil[2]{Department of Statistics and Data Science, Washington University in St.~Louis}
\affil[3]{Department of Statistics and Data Science, University of California,~Los Angeles}
\affil[4]{Department of Mathematics, University of California,~San Diego}
\date{\today}

\begin{document}

\maketitle

\begin{abstract}
In this work, we propose new matrix- and tensor-based methodologies for estimating multivariate intensity functions of inhomogeneous point processes. By viewing multivariate intensity functions as infinite-dimensional matrices or tensors within function spaces, our algorithms attain the optimal bias–variance trade-off, yielding rate-optimal estimation error, with model complexity governed by matrix or tensor ranks. They substantially improve estimation accuracy, while simultaneously reducing computational cost.   To illustrate the adaptivity of the proposed framework, we show that many fundamental classes of multivariate functions, including additive and mean-field models, admit finite-rank tensor representations. We apply our method to a four-dimensional U.S. Geological Survey earthquake dataset, comprising features such as latitude, longitude, depth, and magnitude. Our tensor estimator recovers localized seismicity patterns (California, Oklahoma, Pacific Northwest, north-central US), whereas the kernel baseline oversmooths them.

\end{abstract}


\section{Introduction}\label{sec:intro}

Inhomogeneous point processes model random collections of events in a domain $\mathbb{X}\subset\mathbb{R}^D$ and play a central role in astronomy, biology, neuroscience, epidemiology, seismology, economics, and finance. Classical examples include forest fires \citep{stoyan2000recent,waagepetersen2008estimating,moller2010structured}, earthquakes \citep{bray2013assessment}, crime incidents \citep{baddeley2021analysing}, and financial transactions \citep{bauwens2009modelling,jevtic2014class}. In many contemporary applications, the events are inherently multivariate: seismic catalogs on (longitude, latitude, depth, magnitude), optionally augmented by focal-mechanism or stress-drop attributes; X-ray astronomy photon events on (position, time, energy) \citep{cash1979parameter,park2006bayesian}; and ecological or neuroimaging point patterns in three-dimensional physical space enriched with continuous attributes such as diameter, orientation, or activation level.

When some coordinates are singled out as responses, these processes fall under the \emph{marked point process} framework \citep[see e.g.][]{baccelli2010stochastic,cheng2025deep}, with the responses serving as marks and the remaining coordinates as locations. Rather than the standard regression-style decomposition into a ground intensity and conditional mark distribution \citep{diggle2013statistical,illian2008statistical,baddeley2016spatial}, we directly target the \emph{joint intensity} $\lambda^*:\mathbb{X}\to\mathbb{R}_+$ on the full product space. This avoids the response--covariate asymmetry, and the ground intensity and conditional mark density can be recovered computationally efficiently from the joint estimate by marginalization and normalization (see \Cref{remark:beyond_poisson}).

Accurate estimation of $\lambda^*$ is essential for scientific interpretation and for downstream tasks such as higher-order summaries (e.g.~self-excitation or inhibition) and dependence assessment \citep{baddeley2007spatial,baddeley2000non}. Classical nonparametric estimators such as kernel intensity estimation (KIE) \citep[see e.g.][]{gonzalez2016spatio} suffer the \emph{curse of dimensionality}: both the number of coefficients required for a faithful approximation and the resulting variance grow exponentially in $D$, so KIE's gap to structured estimators is already visible at $D=2$ and is widening sharply for $D\ge 3$. 

We develop matrix- and tensor-based estimators that exploit (approximately) low-rank structure in $\lambda^*$, viewed as an infinite-dimensional tensor on function space. Such structure arises naturally in additive, mean-field, and related models (see \Cref{sec:example_low-rank}); by estimating only the leading spectral components, our methods reduce both estimation error and computational cost relative to classical smoothing. Given $n$ point processes $\{N^{(i)}\}_{i=1}^n$ with $\alpha$-smooth common intensity $\lambda^*$, KIE attains squared $\mathbb{L}_2(\mathbb{X})$ error of order $O(n^{-2\alpha/(2\alpha+D)})$, which degrades rapidly with $D$, whereas our rank-controlled estimators circumvent this exponential dependence, as sketched below.

\noindent\textbf{Bivariate case: low-rank matrix approximation.} For $\lambda^*(x,y)$ on $\mathbb{X}\subset\mathbb{R}^2$, the tensor-product basis expansion
\[
\lambda^*(x,y) \approx \sum_{\mu_1 = 1}^m\sum_{\mu_2 = 1}^m b^*_{\mu_1,\mu_2}\phi_{\mu_1}(x)\phi_{\mu_2}(y),
\qquad b^*_{\mu_1,\mu_2}=\int\!\!\int\lambda^*\phi_{\mu_1}\phi_{\mu_2},
\]
encodes $\lambda^*$ in the $m\times m$ matrix $b^*$ with truncation error $O(m^{-2\alpha})$ under $\alpha$-smoothness \citep{hackbusch2012tensor}. Estimating all $m^2$ entries yields the classical rate $O(n^{-2\alpha/(2\alpha+D)})$ at $D=2$, whereas a rank-$R$ truncated SVD of $b^*$ (see \Cref{fig:SVD}) reduces the parameter count to $2mR+R$, so for bounded $R$ the variance becomes $O(m/n)$. The price is the rank-$R$ bias
\begin{equation}\label{eq:approx_error_rankR}
\xi_{(R)} \;=\; \inf_{\rank(g)\le R}\,\|g-\lambda^*\|_{\mathbb{L}_2(\mathbb{X})},
\end{equation}
which vanishes when $\lambda^*$ is exactly rank-$R$, as in additive or mean-field models. The resulting squared $\mathbb{L}_2(\mathbb{X})$ error is $O(n^{-2\alpha/(2\alpha+1)})+\xi_{(R)}^2$, recovering the one-dimensional minimax rate up to the approximation error.
\begin{figure}[!htbp]
    \centering
    \includegraphics[width=0.6\textwidth]{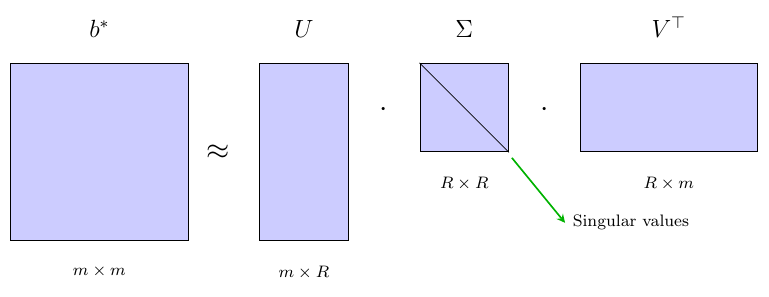}
    \caption{Truncated SVD approximation of the coefficient matrix $b^*$, where $U$ and $V$ are the left and right singular matrices and $\Sigma$ contains the leading $R$ singular values.}
    \label{fig:SVD}
\end{figure}

\noindent\textbf{Multivariate case: low-rank tensor approximation.} The same idea extends to $\lambda^*(x_1,\dots,x_D)$ on $\mathbb{X}\subset\mathbb{R}^D$ via
\[
\lambda^*(x_1,\dots,x_D) \approx \sum_{\mu_1 = 1}^m\cdots\sum_{\mu_D = 1}^m b^*_{\mu_1,\dots,\mu_D}\phi_{\mu_1}(x_1)\cdots\phi_{\mu_D}(x_D),
\]
again with approximation error $O(m^{-2\alpha})$, where $\{b^*_{\mu_1,\dots,\mu_D}\}$ forms a $D$th-order tensor of ambient size $m^D$. Estimating all $m^D$ entries quickly becomes infeasible, but a Tucker rank-$(R_1,\dots,R_D)$ approximation (see \Cref{fig:tensor}) reduces the parameter count to $\prod_{j=1}^D R_j + m\sum_{j=1}^D R_j$. With bounded Tucker ranks, the variance is $O(m/n)$ and the squared $\mathbb{L}_2(\mathbb{X})$ error becomes $O(n^{-2\alpha/(2\alpha+1)})+\xi_{(R_1,\dots,R_D)}^2$, with $\xi_{(R_1,\dots,R_D)}$ defined in \eqref{eq:tucker_low-rank_approx}. This again recovers the one-dimensional minimax rate up to the approximation error.
\begin{figure}[!htbp]
    \centering
    \includegraphics[width=0.6\textwidth]{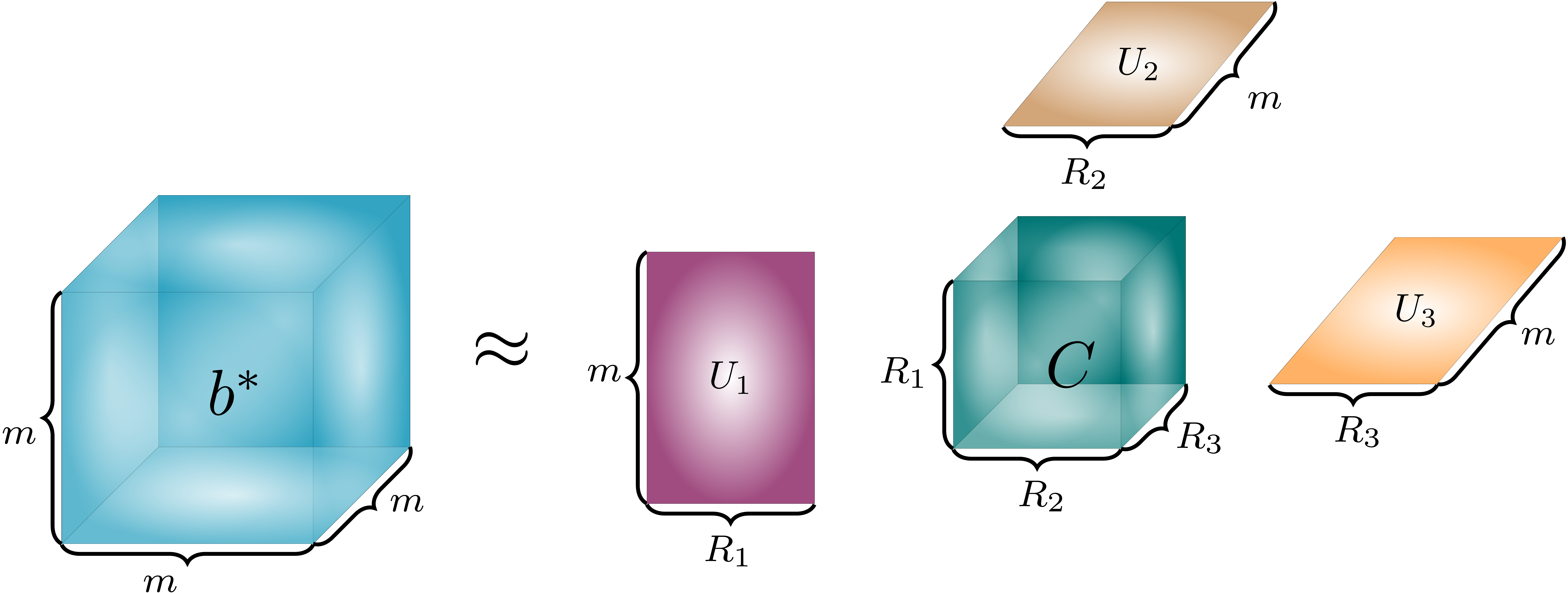}
    \caption{Tucker low-rank approximation of the third-order tensor $b^*$ with rank $(R_1, R_2, R_3)$; $C$ is the core tensor and $U_1, U_2, U_3$ are factor matrices.}
    \label{fig:tensor}
\end{figure}

\subsection{List of contributions}
Our work makes the following key contributions.

\begin{itemize}
    \item \textbf{Matrix- and tensor-based intensity estimators.} We propose nonparametric intensity estimators for multivariate inhomogeneous point processes via low-rank matrix and Tucker tensor decompositions, with explicit, easy-to-implement procedures (Algorithms~\ref{alg0} and~\ref{alg1}).

    \item \textbf{Nonasymptotic theory with rank-adaptive bias--variance trade-off.} We establish finite-sample error bounds (Theorems~\ref{theorem_matrix} and~\ref{thm:main_result}) that cleanly separate the basis-expansion approximation error from the low-rank estimation error, with user-specified ranks controlling complexity. The same guarantees apply to both multiple-process ($n>1$) and single-process ($n=1$) regimes (\Cref{rem:scaling_thinning}), and extend beyond Poisson sampling to Neyman--Scott and temporally dependent data (Appendix~\ref{sec:other_pp}).

    \item \textbf{Minimax optimality.} \Cref{Minimax-Theorem-Tensor} establishes a minimax lower bound for nonparametric intensity estimation that our upper bounds match up to a $\log n$ factor, so the proposed estimators are rate-optimal in the rank-constrained class.

    \item \textbf{Strong empirical performance.} Across seven representative simulation scenarios with $D\in\{2,\ldots,6\}$, the proposed estimators significantly outperform kernel intensity estimation (KIE) in both relative $\mathbb{L}_2$ error (\Cref{summary-num-results} and \Cref{add-num-res}) and computational cost (\Cref{sec:runtime_empirical}), and scale more favorably with $D$. Analogous gains hold on three four-dimensional real-data applications: a daily-aggregated U.S.~earthquake catalog (\Cref{sec:real_data}), a regional California--Nevada seismic catalog at the event level (\Cref{sec:real_data_eq_regional}), and U.S.~tornado trajectory events (\Cref{sec:real_data_tornado}).

    \item \textbf{New theoretical tools for approximately low-rank tensors.} Existing perturbation and approximation results for low-rank tensor estimation assume \emph{exact} low-rank structure \citep{cai2018rate,zhang2018tensor}. We develop analogous tools targeting \emph{approximately} low-rank infinite-dimensional tensors (\Cref{sec:approx_low_rank_tools}), the relevant regime for smooth intensities and other nonparametric objects, which may be of independent interest beyond the present setting.
\end{itemize}

\subsection{Organization}
The rest of the paper is organized as follows.  In \Cref{sec:notation_background}, we introduce notations as well as discuss some background on low-rank tensor approximation for multivariate functions and on inhomogeneous point processes. \Cref{sec:method} introduces our matrix- and tensor-based intensity estimation methods, summarized in Algorithms~\ref{alg0} and~\ref{alg1}, respectively. Theoretical guarantees for both methods are presented in \Cref{sec:approx_low-rank}, focusing on Poisson point processes.   Numerical
studies including a real data application are conducted in \Cref{sec:numeric}. A review of nonparametric intensity function estimation, tensor network approximation, and low-rank tensor estimation is provided in \Cref{sec:lit_rev}.

\section{Notation and background}\label{sec:notation_background}

\subsection{Notation}\label{sec:notation}
For a positive integer $m$,  write $[m]=\{1,\ldots, m\}$. For $a, b \in \mathbb{R}$,  $\lceil a \rceil$ denote the smallest integer
greater than or equal to $a$, $\lfloor a \rfloor$ denote the largest integer less than or equal to $a$, $a \vee b = \max\{a,b\}$ and $a \wedge b = \min\{a,b\}$.

\noindent\textbf{Matrices.}
Let $\mathbb O_{p,r}=\{V\in\mathbb R^{p\times r}:V^\top V=I_r\}$ and $\mathbb O_p=\mathbb O_{p,p}$.
For $M\in\mathbb R^{p\times q}$, its singular value decomposition (SVD) is
$M=U\Sigma V^\top$ with $U\in\mathbb O_p$, $V\in\mathbb O_q$, and singular values
$\sigma_1(M)\ge\cdots\ge\sigma_{\min\{p,q\}}(M)\ge0$ on the diagonal of $\Sigma$.
We write $\|M\|_{\op}=\sigma_1(M)$ and $\|M\|_{\mathrm F}=(\sum_{i=1}^p\sum_{j=1}^q M_{ij}^2)^{1/2}$.
For $R\le \rank(M)$, the rank-$R$ truncated SVD is
$M_{(R)}=U_{(R)}\Sigma_{(R)}V_{(R)}^\top$ with
$U_{(R)}\in\mathbb O_{p,R}$, $V_{(R)}\in\mathbb O_{q,R}$ and
$\Sigma_{(R)}=\diag(\sigma_1(M),\ldots,\sigma_R(M))$.
Throughout, we use the shorthand $\mathrm{SVD}_{(R)}(M) =U_{(R)}$.

\noindent\textbf{Tensors.}
For an $s$th-order tensor $B\in\mathbb R^{p_1\times\cdots\times p_s}$,
$\|B\|_{\mathrm F}=\Big(\sum_{\mu_1=1}^{p_1}\cdots\sum_{\mu_s=1}^{p_s} B_{\mu_1,\ldots,\mu_s}^2\Big)^{1/2}$.
For $j\in[s]$, let $p_{-j} =(\prod_{\ell=1}^s p_\ell)/p_j$ and let $\mathcal M_j(B)$ be the mode-$j$
matricization (the $p_j\times p_{-j}$ unfolding).
For $M\in\mathbb R^{m\times p_j}$, the mode-$j$ product $B\times_j M\in
\mathbb R^{p_1\times\cdots\times p_{j-1}\times m\times p_{j+1}\times\cdots\times p_s}$ is defined by $(B\times_j M)_{\mu_1,\ldots,\mu_{j-1},\, i,\, \mu_{j+1},\ldots,\mu_s}
=\sum_{\mu_j=1}^{p_j} B_{\mu_1,\ldots,\mu_s}\,M_{i,\mu_j}$.
The Tucker rank of $B$ is $(r_1,\ldots,r_s)$ if $\rank(\mathcal M_j(B))=r_j$ for all $j\in[s]$.

\noindent\textbf{Function spaces and tensor products.}
Let $\mathbb X_j\subset\mathbb R^{d_j}$ be measurable and let $\upsilon_j$ be Lebesgue measure restricted
to $\mathbb X_j$. Define the product space
\[
\mathbb X =\mathbb X_1\times\cdots\times\mathbb X_s\subset\mathbb R^{d_1}\times\cdots\times\mathbb R^{d_s}
=\mathbb R^{D},\qquad D =\sum_{j=1}^s d_j,
\]
equipped with the product measure $\upsilon =\upsilon_1\times\cdots\times\upsilon_s$.
We write $\mathbb L_2(\mathbb X)$ for square-integrable functions (w.r.t.\ $\upsilon$), and for
$A:\mathbb X\to\mathbb R$ define
\[
\|A\|_{\mathbb L_2} =\Big(\int_{\mathbb X} A(x)^2\d\upsilon(x)\Big)^{1/2},
\qquad
\|A\|_\infty =\mathrm{ess\,sup}_{x\in\mathbb X}|A(x)|.
\]
For $u_j\in \mathbb L_2(\mathbb X_j)$, define the multilinear functional
\begin{equation}\label{eq:inner_prod}
A[u_1,\ldots,u_s]
=\int_{\mathbb X} A(x_1,\ldots,x_s)\prod_{j=1}^s u_j(x_j)\d\upsilon(x_1,\ldots,x_s).
\end{equation}
For $u_j\in \mathbb L_2(\mathbb X_j)$ and $u_k\in \mathbb L_2(\mathbb X_k)$, the tensor product
$u_j\otimes u_k\in \mathbb L_2(\mathbb X_j\times\mathbb X_k)$ is defined by
$(u_j\otimes u_k)(x_j,x_k)=u_j(x_j)u_k(x_k)$, and $u_1\otimes\cdots\otimes u_s$ is defined analogously.

\noindent\textbf{Basis expansion and coefficient tensor.}
For each $j\in[s]$, let $\{\phi_{j,\mu}\}_{\mu\ge1}$ be an orthonormal basis of $\mathbb L_2(\mathbb X_j)$.
Then $\big\{\phi_{1,\mu_1}\otimes\cdots\otimes\phi_{s,\mu_s}\big\}_{\mu_1,\ldots,\mu_s\ge1}$ is an
orthonormal basis of $\mathbb L_2(\mathbb X)$, and any $\lambda\in \mathbb L_2(\mathbb X)$ admits the expansion
\begin{equation}\label{eq:basis_expansion}
\lambda(x_1,\ldots,x_s)=\sum_{\mu_1=1}^\infty\cdots\sum_{\mu_s=1}^\infty
b_{\mu_1,\ldots,\mu_s}\prod_{j=1}^s \phi_{j,\mu_j}(x_j),
\qquad \text{in }\mathbb L_2(\mathbb X),
\end{equation}
with coefficients
$b_{\mu_1,\ldots,\mu_s}
=\big\langle \lambda,\ \phi_{1,\mu_1}\otimes\cdots\otimes\phi_{s,\mu_s}\big\rangle_{\mathbb L^2(\mathbb X)}
=\lambda[\phi_{1,\mu_1},\ldots,\phi_{s,\mu_s}]$.

\noindent\textbf{Order notation.}
For a sequence of random variables $\{X_n\}$ and positive numbers $\{a_n\}$, we write $X_n = O_p(a_n)$ if $\lim_{K \to \infty}\limsup_{n \to \infty}\mathbb{P}(|X_n| \geq Ka_n) = 0$. For two sequences of positive numbers $\{a_n\}$ and $\{b_n\}$, we write $a_n = O(b_n)$ if there exists some constant $C > 0$ such that $a_n/b_n \leq C$ for all large $n$.

For ease of reference, a quick-reference list of the most frequently used symbols is provided in \Cref{tab:notation_glossary} of the supplementary material.

\subsection{Low-rank approximation for multivariate functions}\label{sec:approximately_low-rank_fct}
Fix a positive integer $s \leq D$, by \eqref{eq:basis_expansion}, we write $\lambda^*\in\mathbb{L}_2(\mathbb{X})$
\begin{equation}\label{eq:expand_A}
\lambda^*(x_1, \dots, x_s) = \sum_{\mu_1 = 1}^{\infty}\cdots\sum_{\mu_s = 1}^{\infty}b^*_{\mu_1, \dots, \mu_s}\prod_{j=1}^s\phi_{j,\mu_j}(x_j),
\end{equation}
where $b^*_{\mu_1, \dots, \mu_s} =  \lambda^*  [\phi_{1,\mu_1},\dots, \phi_{s, \mu_s}]$ are the coefficient tensor entries.

\noindent\textbf{Mode-wise spectra and approximate low rank.}
For $s=2$, we associate $\lambda^*$ with the Hilbert--Schmidt integral operator
$H_{\lambda^*}:\mathbb{L}_2(\mathbb{X}_1)\to \mathbb{L}_2(\mathbb{X}_2)$,
\[
(H_{\lambda^*} g)(x_2)= \int_{\mathbb{X}_1} \lambda^*(x_1,x_2)\,g(x_1) \d \upsilon_1 (x_1),
\]
whose singular values $\{\sigma_\mu(\lambda^*)\}_{\mu\ge1}$ are well-defined and satisfy
$\sum_{\mu\ge1}\sigma_\mu^2(\lambda^*)=\|\lambda^*\|_{\mathbb{L}_2(\mathbb{X})}^2<\infty$.
For $s\ge3$, we consider the mode-$j$ reshaping
as a two-variable function admitting a function SVD:
\begin{equation}\label{eq:unfold_SVD}
\lambda^*(x_1, \dots, x_s) = \lambda^*_j(x_j, x_{-j}) = \sum_{\mu = 1}^{\infty}\sigma_{j,\mu}(\lambda^*)\Phi_{j,\mu}(x_{j})\Psi_{j,\mu}(x_{-j}),
\end{equation}
where $x_{-j}=(x_1,\dots,x_{j-1},x_{j+1},\dots,x_s)$, $\{\sigma_{j,\mu}(\lambda^*)\}_{\mu = 1}^{\infty}$ are singular values of $H_{\lambda^*_j}$, as well as $\{\Phi_{j,\mu}\}_{\mu = 1}^{\infty} \subset \mathbb{L}_2(\mathbb{X}_{1})$ and $\{\Psi_{j,\mu}\}_{\mu = 1}^{\infty} \subset \mathbb{L}_2(\mathbb{X}_{2})$ are singular functions. 
Our methods are guided by the common situation in which these spectra decay sufficiently fast, so that each reshaping is well approximated by a low-rank truncation.

\noindent\textbf{Finite-dimensional projection.} 
To obtain a computationally tractable representation, we restrict each coordinate space to the span of the first $m^{d_j}$ basis functions.
This yields the projected approximation
\begin{equation}\label{eq:projection_tensor}
\lambda^*(x_1,\dots,x_s)
\approx
\sum_{\mu_1=1}^{m^{d_1}}\cdots\sum_{\mu_s=1}^{m^{d_s}}
b^*_{\mu_1,\dots,\mu_s}\,
\phi_{1,\mu_1}(x_1)\cdots \phi_{s,\mu_s}(x_s),
\end{equation}
with approximation properties summarized in \Cref{remark:approx_error}.
When $m$ is large enough, the coefficient tensor $b^*=[b^*_{\mu_1,\dots,\mu_s}]$ inherits the dominant mode-wise spectral structure of $\lambda^*$,
in the sense that the leading singular subspaces of $\mathcal{M}_j(b^*)$ approximate those of the population operators induced by the corresponding reshaping of $\lambda^*$, i.e.~$\sigma_{j,\mu}(b^*) \approx \sigma_{j,\mu}(\lambda^*)$.

\noindent\textbf{Low-rank approximation error.} For $s=2$, we use the rank-$R$ approximation error in \eqref{eq:approx_error_rankR}.
For $s\ge3$ and user-specified Tucker ranks $(R_1,\dots,R_s)$, define
\begin{equation}\label{eq:tucker_low-rank_approx}
\xi_{(R_1,\dots,R_s)}
=
\inf\Big\{\|g-\lambda^*\|_{\mathbb{L}_2(\mathbb{X})}:\ 
\rank\big(g_j(\cdot,\cdot)\big)\le R_j\ \text{for all } j\in[s]\Big\},
\end{equation}
where $g_j(x_j,x_{-j})$ denotes the mode-$j$ reshaping of $g(x_1,\dots,x_s)$.
This quantity captures the population bias incurred by approximating $\lambda^*$ by a function whose mode-wise unfoldings have ranks bounded by
$(R_1,\dots,R_s)$.

\begin{example}[A bivariate intensity]\label{ex:approx_low_rank} Let $\lambda^*(x,y)=2+\exp\!\big(-(x-y)^2/0.1\big)$ on $\mathbb{X}=[0,1]^2$. The diagonal Gaussian kernel is analytic and non-separable, so the singular values of $H_{\lambda^*}$ decay exponentially. Hence $\lambda^*$ is not exactly low rank, but is well approximated by its leading components: the relative truncation error $\xi_{(R)}/\|\lambda^*\|_{\mathbb{L}_2}$ defined in \eqref{eq:approx_error_rankR} equals $0.140$, $0.068$, and $0.022$ for $R=1,2,3$, respectively. Figure~\ref{fig:approx_low_rank} shows $\lambda^*$, its singular values, and the best rank-$R$ truncations: rank-$1$ captures only the constant baseline and the overall trend, rank-$2$ begins to resolve the diagonal ridge, and rank-$3$ is visually indistinguishable from $\lambda^*$. This is the type of fast spectral decay that our matrix- and tensor-based estimators are designed to exploit. \end{example}

\begin{figure}[!htbp] \centering \includegraphics[width=0.85\textwidth]{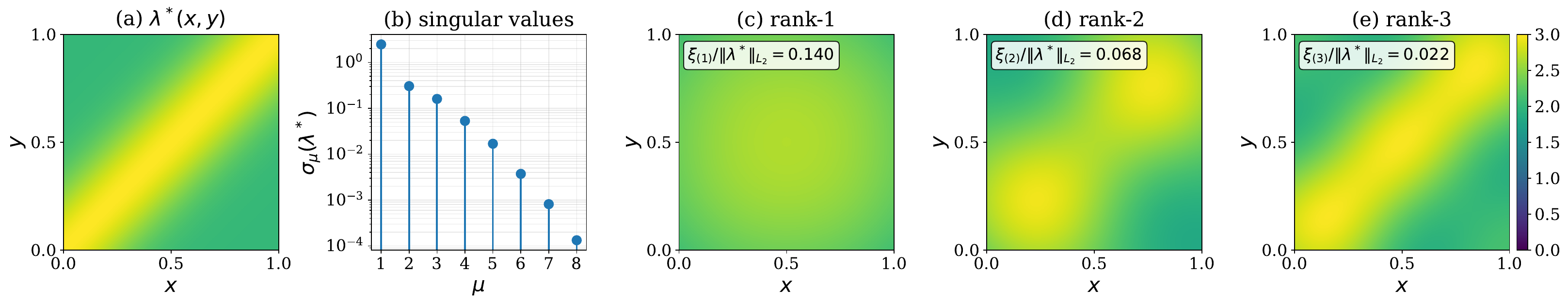} \caption{Bivariate intensity $\lambda^*(x,y)=2+e^{-(x-y)^2/0.1}$ on $[0,1]^2$. (a)~heatmap of $\lambda^*$; (b)~singular values $\sigma_\mu(\lambda^*)$ on a log scale; (c)--(e)~best rank-$R$ approximations for $R=1,2,3$.} \label{fig:approx_low_rank}
\end{figure}

\subsection{Inhomogeneous Poisson point processes}\label{sec:background}
An inhomogeneous point process $N$ is a set of random points $\{X_1, X_2, \dots\} \subseteq \mathbb X \subset \mathbb R^D$. For any compact subset $S \subseteq \mathbb{X}$,  let $N(S) = |S \cap N|$ be the number of points in $S$.
The intensity measure $\Pi(S) = \mathbb{E}[N(S)]$ gives the expected number of points in $S$.
If $\Pi$ is absolutely continuous with respect to the Lebesgue measure $\upsilon$, there exists an  intensity function $\lambda^*$ such that
\[
 \Pi(S) = \int_S \lambda^*(x) \d \upsilon(x), \;\; \text{where} \;\; \lambda^*(x)  =  \frac{\d \Pi}{\d \upsilon}(x).
\] 
Note that $\lambda^*$ is the Radon–Nikodym derivative of $\Pi$ with respect to $\upsilon$, reflecting the first-order properties of $N$.

We call $N$ a \emph{Poisson point process} with intensity function $\lambda^*$ if
\begin{enumerate}
    \item for every compact $S\subseteq\mathbb{X}$, $N(S)$ is Poisson distributed with mean $\Pi(S)$;
    \item for every $w\in\mathbb{N}_+$ and any disjoint compact $S_1,\dots,S_w\subseteq\mathbb{X}$, the counts $N(S_1),\dots,N(S_w)$ are mutually independent.
\end{enumerate}

While our primary focus is on theoretical and numerical analysis of Poisson point processes, marked Poisson processes and other types of point processes are discussed in \Cref{remark:beyond_poisson} and \Cref{sec:other_pp}. Future research will extend theoretical developments to encompass broader classes of point processes under more general sampling conditions.

\section{Methodology}\label{sec:method}
We expand $\lambda^*$ in a truncated tensor-product basis (\Cref{sec:math_setup}), reducing intensity estimation to recovery of the coefficient tensor $b^*$. To exploit its low-rank structure, we propose two estimators: a \emph{matrix-based} method (\Cref{sec:method_matrix}, Algorithm~\ref{alg0}) that applies soft singular-value thresholding to the empirical coefficient matrix when $s=2$, and a \emph{tensor-based} method (\Cref{sec:method_tensor}, Algorithm~\ref{alg1}) that produces a Tucker low-rank approximation of the empirical coefficient tensor via higher-order SVD (HOSVD) initialization with a sketched refinement step when $s\ge 3$.

\subsection{Mathematical setup}\label{sec:math_setup}
Consider \(n\) i.i.d.~inhomogeneous point processes \(\{N^{(i)}\}_{i = 1}^n\) on a compact domain \(\mathbb{X} \subset \mathbb{R}^D\). We assume that the domain \(\mathbb{X}\) factorizes as
\begin{align}\label{eq:domain_partition}
\mathbb{X} = \mathbb{X}_1 \times \cdots \times \mathbb{X}_s \subset \mathbb{R}^{d_1}\times \cdots \times \mathbb{R}^{d_s} = \mathbb{R}^{D}, \quad \text{with} \ \sum_{j=1}^s d_j = D.
\end{align}
Each point process $N^{(i)}$ shares the same unknown intensity function \(\lambda^* : \mathbb{X} \to \mathbb{R}_+\) in \(\mathbb{L}_2(\mathbb{X})\).

For each coordinate space $\mathbb{X}_j$, we select orthonormal basis \(\{\phi_{j,\mu_j}\}_{\mu_j = 1}^{m^{d_j}} \subset \mathbb{L}_2(\mathbb{X}_j)\) 
and consider the tensor-product subspace
\(\mathcal U_1\otimes\cdots\otimes \mathcal U_s\), where \(\mathcal U_j=\mathrm{span}\{\phi_{j,\mu_j}:\mu_j\in[m^{d_j}]\}\). Projecting \(\lambda^*\) onto this subspace yields the coefficient tensor $b^* \in \mathbb{R}^{m^{d_1} \times \cdots \times m^{d_s}}$ with entries
$b^*_{\mu_1, \dots, \mu_s} = \lambda^*[\phi_{1,\mu_1}, \dots, \phi_{s,\mu_s}]$.
Define the empirical measure $\widehat{\lambda} = \frac{1}{n}\sum_{i = 1}^n\sum_{u \in N^{(i)}}\delta_u$,
where $\delta_u$ is a point mass at $u$.
The classical nonparametric method directly estimates $b^*$ by the empirical coefficient tensor $\widehat{b}$ with entries
\begin{align}\label{eq:empirical_coefficient_tensor}
\widehat{b}_{\mu_1,\dots,\mu_s} =  \widehat{\lambda}[\phi_{1,\mu_1}, \dots, \phi_{s,\mu_s}]
= \frac{1}{n}\sum_{i = 1}^n\sum_{X^{(i)} \in N^{(i)}} \phi_{1,\mu_1}(X_{1}^{(i)}) \cdots \phi_{s,\mu_s}(X_{s}^{(i)}),
\end{align}
where \(X^{(i)} = (X^{(i)}_{1}, \dots, X^{(i)}_{s}) \in \mathbb{X}\) represents a point in \(N^{(i)}\). 
We propose methods that seek structured (low-rank) approximations of $b^*$, constructed from $\widehat{b}$.

\subsection{Matrix-based method}\label{sec:method_matrix}
When \(s=2\), the coefficient tensor \(\widehat b\) in \eqref{eq:empirical_coefficient_tensor} is a matrix
\(\widehat b\in\mathbb{R}^{m^{d_1}\times m^{d_2}}\) with \(d_1+d_2=D\). Without loss of generality, assume
\(d_{\min}=d_1\le d_2=d_{\max}\). We estimate the intensity by applying soft singular value thresholding (soft-SVT) to \(\widehat b\).

Compute the SVD $\widehat b = \widehat U \widehat \Sigma \widehat V^{\top}$. For a threshold $\gamma > 0$, define
the thresholded diagonal matrix $T_{\gamma}(\widehat \Sigma)$ with diagonal entries
    $$(T_{\gamma}(\widehat \Sigma))_{j,j}=\max \{0, \widehat \Sigma_{j,j}-\gamma\}, \quad j \in [m^{d_1}],$$
and set the soft-SVT estimator of the coefficient matrix as
    $$T_{\gamma}(\widehat b) = \widehat U T_{\gamma}(\widehat \Sigma) \widehat V^{\top}.$$
Mapping back to the function space yields, for any point $(x_{1}^*, x_{2}^*)\in \mathbb{X}$, $$\widehat{\lambda}_{\text{Matrix}} (x_{1}^*, x_{2}^*) = (\phi^{(1)} (x_{1}^* ))^{\top} \cdot T_{\gamma}(\widehat b) \cdot \phi^{(2)} (x_{2}^*), \quad
\phi^{(j)}(x_j)=(\phi_{j,1}(x_j),\dots,\phi_{j,m^{d_j}}(x_j))^\top.$$
\Cref{alg0} summarizes the procedure.

\begin{algorithm}[htbp]
\caption{ Multivariate intensity estimation via matrix soft-SVT ($s = 2$)
}\label{alg0}
\begin{algorithmic}[1]
\INPUT Point processes $\{N^{(i)}\}_{i = 1}^n$, threshold $\gamma $, basis maps
\(\phi^{(j)}(x_j)\in\mathbb{R}^{m^{d_j}}\) for \(j=1,2\).
\State Compute the empirical coefficient matrix 
 $\widehat{b}$: $\widehat b_{\mu_1,\mu_2} = \widehat \lambda[\phi_{1,\mu_1}, \phi_{2,\mu_2}]$.
\State  Compute the SVD:  
$\widehat b = \widehat U \widehat \Sigma \widehat V^{\top}$.
\State Soft-threshold singular values: $(T_{\gamma}(\widehat \Sigma))_{j,j}=\max \{0, \widehat \Sigma_{j,j}-\gamma\}$, $j \in [m^{d_1}]$.
\State  Form the soft-SVT coefficient matrix: $T_{\gamma}(\widehat b) = \widehat U T_{\gamma}(\widehat \Sigma) \widehat V^{\top}$.
\OUTPUT  
$\widehat{\lambda}_{\text{Matrix}} (x_{1}^*, x_{2}^*) = (\phi^{(1)} (x_{1}^* ))^{\top} \cdot T_{\gamma}(\widehat b) \cdot \phi^{(2)} (x_{2}^*) $ for any point $(x_{1}^*, x_{2}^*)\in \mathbb{X}$. 
\end{algorithmic}
\end{algorithm}

\subsection{Tensor-based method}\label{sec:method_tensor}
When $s \geq 3$, the coefficient tensor $\widehat b \in \mathbb{R}^{m^{d_1} \times \cdots \times m^{d_s}}$ in \eqref{eq:empirical_coefficient_tensor} is of order $s$.  We estimate a Tucker low-rank approximation by recovering the leading mode-wise singular subspaces and then projecting \(\widehat b\) onto these subspaces.
\begin{enumerate}
    \item \textbf{Initialization via HOSVD}: For each \(j\in[s]\), form the mode-\(j\) matricization \(\mathcal M_j(\widehat b)\) and take a truncated SVD to obtain an initial subspace estimator \(\widehat U_j^{(0)}\in\mathbb{R}^{m^{d_j}\times R_j}\).
    \item \textbf{Refinement via tensor sketching}: For each $j$, compute the sketched matrix $$\mathcal{M}_j(\widehat b) \cdot \otimes_{k \neq j} \widehat{U}_{k}^{(0)} \in \mathbb{R}^{m^{d_j} \times \prod_{k \neq j}R_k},$$
    and take the truncated SVD to obtain the refined subspaces $\widehat U_j^{(1)}$.
    \item \textbf{Projection}: Form the low-rank coefficient tensor
    \[\widetilde{b} = \widehat{b} \times_1 \mathcal{P}_{\widehat{U}^{(1)}_{1}} \cdots \times_s \mathcal{P}_{\widehat{U}^{(1)}_{s}}, \]
    where $\mathcal{P}_{\widehat{U}^{(1)}_{j}}$ denotes the projection onto the column space of $\widehat{U}_j^{(1)}$.
    \item \textbf{Intensity Estimation}: Define $\widehat{\lambda}_{\text{Tensor}}$ by mapping $\widetilde{b}$ back to the function space.
\end{enumerate}
The tensor-based method is summarized as \Cref{alg1}.

\begin{algorithm}[htbp]
\caption{ Multivariate intensity estimation via tensor decomposition ($3 \leq s \leq D$)
}\label{alg1}
\begin{algorithmic}[1]
\INPUT Point processes $\{N^{(i)}\}_{i = 1}^n$, target Tucker rank $(R_1, \dots, R_s)$, basis maps
\(\phi^{(j)}(x_j)=(\phi_{j,1}(x_j),\dots,\phi_{j,m^{d_j}}(x_j))^\top\) for \(j\in[s]\).
\State Sample split: Partition $[n]$ into disjoint $H_1 \cup H_2 \cup H_3 = [n]$ of comparable sizes and set $\widehat \lambda^{H_k} = |H_k|^{-1}\sum_{i \in H_k}\sum_{u \in N^{(i)}}\delta_u$.
\For{$k \in [3]$}
    \State Compute the empirical coefficient tensors $\widehat b^{H_k}$: $\widehat b^{H_k}_{\mu_1, ...,\mu_s} = \widehat \lambda^{H_k}[ \phi_{1, \mu_1}, ... ,\phi_{s, \mu_s}]$.
\EndFor
\For{$j \in [s]$}
    \State Initialize the singular vectors: $\widehat{U}_{j}^{(0)} = \text{SVD}_{(R_j)} (\mathcal{M}_j(\widehat b^{H_1}))$.
\EndFor
\For{$j \in [s]$}
    \State Compute the sketched matrix: $S_j = \mathcal{M}_j(\widehat b^{H_2}) \cdot \otimes_{k \neq j} \widehat{U}_{k}^{(0)}$.
    \State Refine the singular vectors: $\widehat{U}_{j}^{(1)} = \text{SVD}_{(R_j)} ( S_j )$.
\EndFor
\State  Compute final low-rank coefficient tensor:
$\widetilde{b} = \widehat{b}^{H_3} \times_1 \mathcal{P}_{\widehat{U}^{(1)}_{1}} \cdots \times_s \mathcal{P}_{\widehat{U}^{(1)}_{s}}$.
\OUTPUT $    \widehat{\lambda}_{\text{Tensor}}(x^*_{1}, \dots, x^*_{s}) = \widetilde{b} \times_1 \phi^{(1)}(x^*_{1}) \cdots \times_s \phi^{(s)}(x^*_{s})$ for any point $(x^*_{1}, \dots, x^*_{s}) \in \mathbb{X}$. 
\end{algorithmic}
\end{algorithm}
\begin{remark}[Sample splitting]
In \Cref{alg1}, sample splitting partitions the $n$ processes into three disjoint subsets $H_1,H_2,H_3$ of size $\approx n/3$, used to initialize the mode-wise singular subspaces $\{\widehat U_j^{(0)}\}_{j=1}^s$, refine them to $\{\widehat U_j^{(1)}\}_{j=1}^s$, and construct the final projection $\widetilde b$, respectively. The resulting independence between coefficient tensors and singular subspaces simplifies the analysis in \Cref{sec:approx_low-rank}. For a single Poisson process ($n=1$), an analogous split is obtained by random thinning: assign each point independently to one of three groups with probability $1/3$, yielding three independent Poisson processes with intensity $\lambda^*/3$ \cite[see e.g.][]{baraud2009estimating}; \Cref{alg1} then estimates $\lambda^*/3$ and the output is rescaled by $3$ (see \Cref{rem:scaling_thinning}). In settings where thinning is infeasible (e.g.~non-Poisson processes), splitting may be omitted in practice with negligible numerical loss, in which case we set $\widehat \lambda^{H_1} = \widehat \lambda^{H_2} = \widehat \lambda^{H_3} = n^{-1}\sum_{i \in [n]}\sum_{u \in N^{(i)}}\delta_u$ in \Cref{alg1}.
\end{remark}

\section{Theory}\label{sec:approx_low-rank}
This section establishes theoretical guarantees for the matrix- and tensor-based estimators of \Cref{sec:method} under Poisson sampling; extensions to Neyman--Scott and temporally dependent settings are deferred to \Cref{sec:other_pp}.

\subsection{Regularity and basis selection}
\label{sec:assumptions}

We follow the mathematical setup introduced in \Cref{sec:math_setup}. Let $W_2^\alpha(\mathbb{X})$ denote the Sobolev space of functions on $\mathbb{X} \subset \mathbb{R}^D$ with smoothness $\alpha$, equipped with the Sobolev norm $\|\cdot\|_{W_2^\alpha(\mathbb{X})}$ (see~\Cref{sec:approximation_theory} for details).
We  impose the following regularity condition.

\begin{assumption}[Regularity  condition on intensity function]\label{ass:approx_err1} 
The (unknown) intensity function \(\lambda^* : \mathbb{X} \to \mathbb{R}_+\) is such that $\tfrac{\|\lambda^*\|_{W_2^\alpha(\mathbb{X})}}{\|\lambda^*\|_{\infty}} \in (0, \infty)$.
\end{assumption}

\begin{remark}[Thinning and superposition]\label{rem:scaling_thinning}
Assumption~\ref{ass:approx_err1} is preserved under the Poisson thinning and superposition operations that bridge the single-process ($n=1$) and multiple-process ($n>1$) regimes: thinning by a factor $p\in(0,1)$ replaces $\lambda^*$ by $p\lambda^*$, and superposing $n$ independent processes adds their intensities, so in either case the underlying intensity is rescaled by a positive constant. For any $\Lambda>0$ and any nonnegative shape $f$ with $\int_{\mathbb X} f = 1$, writing $\lambda^* = \Lambda f$ yields
$\tfrac{\|\lambda^*\|_{W_2^\alpha(\mathbb X)}}{\|\lambda^*\|_{\infty}}
=
\tfrac{\|f\|_{W_2^\alpha(\mathbb X)}}{\|f\|_{\infty}}$,
so the regularity requirement governs only the shape $f$ and is insensitive to the overall magnitude $\Lambda=\int_{\mathbb X}\lambda^*(x)\,\mathrm{d}x$. Consequently, the expected total number of points may grow (or even diverge) under superposition without violating Assumption~\ref{ass:approx_err1}, and the same algorithms and theoretical guarantees apply in both the $n=1$ and $n>1$ settings.
\end{remark}

\noindent\textbf{Building tensor product basis functions:}
For each subdomain $\mathbb{X}_j \subset \mathbb{R}^{d_j}$, we select a suitable kernel function \(\mathcal{K}_j: \mathbb{X}_j \times \mathbb{X}_j \to \mathbb{R}\), such that its reproducing kernel Hilbert space (RKHS) coincides with the univariate Sobolev space 
\(W_2^{\alpha}(\mathbb{X}_j)\). In practice, we pick the first $m^{d_j}$ eigenfunctions of $\mathcal{K}_j$ as an orthonormal basis \(\{\phi_{j,\mu}\}_{\mu=1}^{m^{d_j}}\subset \mathbb L_2(\mathbb{X}_j)\). The resulting tensor products
\(\phi_{1,\mu_1}\otimes\cdots\otimes \phi_{s,\mu_s}\)
form an \(m^D\)-dimensional tensor-product basis for \(\mathbb L_2(\mathbb{X})\).

\begin{assumption}\label{ass:approx_err2}
For each $j \in [s]$, the kernel $\mathcal{K}_j$ generates the RKHS $W_2^{\alpha}(\mathbb{X}_j)$.
\end{assumption}

\begin{remark}[Approximation error]\label{remark:approx_error}
Under Assumption~\ref{ass:approx_err2}, for any functions $A \in W_2^\alpha(\mathbb{X})$,
\begin{equation}\label{eq:approx_error}
\biggl\|A  -     \sum_{\mu_1 = 1}^{m^{d_1}}\cdots\sum_{\mu_s=1}^{m^{d_s}}  
A[\phi_{1,\mu_1}, \dots, \phi_{s,\mu_s}]\;\phi_{1,\mu_1} \cdots \phi_{s,\mu_s}  
\biggr\|_{\mathbb{L}_2(\mathbb{X})}^2 
\;\leq\; s\,m^{-2\alpha} \|A\|_{W_2^\alpha(\mathbb{X})}^2.
\end{equation}
See \Cref{sec:approximation_theory} for more details.
\end{remark}

\subsection{Upper bound for matrix-based method}\label{sec:theory_matrix}

The next theorem establishes the estimation error bound on $\widehat \lambda_{\mathrm{Matrix}}$.

\begin{theorem}[Error bound on the matrix-based estimator]\label{theorem_matrix}
    Let $\{N^{(i)}\}_{i = 1}^n$ be i.i.d.~inhomogeneous Poisson point processes with intensity function $\lambda^*$.  Let $\widehat{\lambda}_{\mathrm{Matrix}}$ be the matrix-based estimator output by \Cref{alg0}, and set 
\begin{equation}\label{eq:choice_matrix}
m = \lceil(\|\lambda^*\|_{W_2^\alpha(\mathbb{X})}n)^{1/(2\alpha+d_{\max})}\rceil \quad \text{and} \quad \gamma = C_{\gamma}\sqrt{ \frac{\|\lambda^*\|_{\infty}^{2\alpha/(2\alpha+d_{\max})}\|\lambda^*\|_{W_2^\alpha(\mathbb{X})}^{2d_{\max}/(2\alpha+d_{\max})}\log(n)}{n^{2\alpha/(2\alpha+d_{\max})}}},
\end{equation}
where $C_{\gamma} > 0$ is an absolute constant, $d_{\max} = \max\{d_1, d_2\}$ and $\alpha \geq 1$ is the smoothness parameter of $\lambda^*$.
Suppose Assumptions~\ref{ass:approx_err1} and \ref{ass:approx_err2} hold,
we have for any integer value $R > 0$
\[
    \|\lambda^* -\widehat{\lambda}_{\mathrm{Matrix}} \|_{\mathbb{L}_2(\mathbb{X})}^2 = O_{p}\left( \frac{\|\lambda^*\|_{\infty}^{2\alpha/(2\alpha+d_{\max})}\|\lambda^*\|_{W_2^\alpha(\mathbb{X})}^{2d_{\max}/(2\alpha+d_{\max})}}{n^{2\alpha/(2\alpha+d_{\max})}}\log(n) +\xi_{(R)}^2 \right),
\]
where $\xi_{(R)}$, as defined in \eqref{eq:tucker_low-rank_approx}, represents error in $\mathbb{L}_2$ norm between $\lambda^*$ and its best rank-$R$ approximation function.
\end{theorem}

\Cref{theorem_matrix} suggests to partition $D$ coordinates into two subgroups with roughly the same size, i.e.~$d_1 \approx d_2 \approx \lceil D/2 \rceil$.  If $\lambda^*$ is an exactly low-rank function, e.g.~the additive or mean-field functions, then the term $\xi_{(R)}$ is zero for all  $R$ no smaller than the true rank. Note that KIE is not adaptive to the low-rank structure of $\lambda^*$: even in exactly low-rank settings,
it still attains the error rate
   \[
    O_p\left(\frac{\|\lambda^*\|_{\infty}^{2\alpha/(2\alpha+D)}\|\lambda^*\|_{W_2^\alpha(\mathbb{X})}^{2D/(2\alpha + D)}}{n^{2\alpha/(2\alpha + D)}}\right),
    \]
In contrast, our matrix-based method replaces $D$ by $\lceil D/2 \rceil$ in the exponent, leading to faster convergence rates.

\subsection{Upper bound for tensor-based method}\label{sec:poisson}
We now analyze the tensor-based estimator produced by \Cref{alg1}. Beyond Assumptions~\ref{ass:approx_err1}--\ref{ass:approx_err2}, the analysis requires (i) a mild compatibility condition between the coordinate partition and the smoothness level, and (ii) a minimum spectral gap at the target Tucker ranks to ensure stable recovery of the mode-wise singular subspaces.

\begin{assumption}[Partition condition and spectral gap]\label{ass:SingularGap_B*}
    Suppose the partition of coordinates satisfies 
    \begin{align}\label{eq:condition_dim}
    D < 2\alpha + d_{\max} + d_{\min},
    \end{align}
      where $d_{\max} = \max\{d_1, \dots, d_s\}$ and $d_{\min} = \min\{d_1, \dots, d_s\}$. In addition, suppose that for each $j \in [s]$, the singular values $\{\sigma_{j,k}(\lambda^*)\}_{k = 1}^{\infty}$, defined in \Cref{sec:approximately_low-rank_fct},  satisfy
      \begin{align}\label{eq:ass2}
      \min_{j = 1}^s\{\sigma_{j,R_j}(\lambda^*)-\sigma_{j,R_j+1}(\lambda^*)\}^2 \geq C_{\mathrm{gap}}\|\lambda^*\|_{W_2^{\alpha}(\mathbb{X})}^{2(D - d_{\min})/(2\alpha+d_{\max})}\|\lambda^*\|_{\infty}^{\beta}n^{-\beta}\log(n),
      \end{align}
      with
      $\beta = (2\alpha + d_{\max} + d_{\min} - D)/(2\alpha +  d_{\max})$,
      where $C_{\mathrm{gap}} > 0$ is a sufficiently large absolute  constant, and $(R_1, \dots, R_s)$ is the user-specified target Tucker rank. 
\end{assumption} 
\Cref{ass:SingularGap_B*}~is a mild assumption in two aspects.   First, Condition~\eqref{eq:condition_dim} on the total dimension $D$ depends on both the smoothness of $\lambda^*$ and the user-specified coordinate partition. For example, if $\alpha = 2$, it allows us to handle inhomogeneous point processes in up to $10$-dimension (see~\Cref{remark:coor_par} for details), accounting for a majority of spatial/spatial-temporal point process data in real world. Second, because  $\beta > 0$, the required spectral gap Condition~\eqref{eq:ass2} vanishes  as $n \to \infty$.

We now present theoretical guarantees for the tensor-based intensity function estimator $\widehat \lambda_{\mathrm{Tensor}}$. See \Cref{sec:proof_main} for a proof.
\begin{theorem}[Error bound on the tensor-based estimator]\label{thm:main_result}
    Let $\{N^{(i)}\}_{i = 1}^n$ be a set of i.i.d.~inhomogeneous Poisson point processes, with intensity function $\lambda^*$.  Let $\widehat{\lambda}_{\mathrm{Tensor}}$ be the tensor-based estimator output by \Cref{alg1} with the target Tucker rank $(R_1, \dots, R_s)$, and set
\begin{equation}\label{eq:choice_m}
m = \lceil(\|\lambda^*\|_{W_2^\alpha(\mathbb{X})}n)^{1/(2\alpha+d_{\max})}\rceil,
\end{equation}
where $d_{\max} = \max\{d_1, \dots, d_s\}$ and $\alpha \geq 1$ is the smoothness parameter of $\lambda^*$.
Suppose Assumptions \ref{ass:approx_err1}, \ref{ass:approx_err2} and \ref{ass:SingularGap_B*} hold,
we have
\begin{align}\label{eq:main_error_bound}
        &\|\lambda^*-\widehat{\lambda}_{\mathrm{Tensor}}\|_{\mathbb{L}_2(\mathbb{X})}^2\nonumber\\
        =& O_p\left(\left\{\frac{\|\lambda^*\|_{\infty}^{2\alpha/(2\alpha + d_{\max})}{\|\lambda^*\|_{W_2^{\alpha}}^{2d_{\max}/(2\alpha + d_{\max})}}\sum_{j = 1}^sR_j}{n^{2\alpha/(2\alpha+d_{\max})}} + \frac{\|\lambda^*\|_{\infty}\prod_{j=1}^sR_j}{n}\right\}\log(n) + \xi_{(R_1, \dots, R_s)}^2\right),
\end{align}
where $\xi_{(R_1, \dots, R_s)}$ is the best rank-$(R_1, \dots, R_s)$ approximation error defined in \eqref{eq:tucker_low-rank_approx}.
\end{theorem}

    Condition~\eqref{eq:condition_dim} and the error rate in \eqref{eq:main_error_bound} together show how the partition (through $d_{\max}$) governs both the allowable ambient dimension $D$ and the estimation error rate. We will explore this trade-off carefully in \Cref{remark:coor_par} through an example.

 \Cref{thm:main_result} shows that our tensor-based method outperforms the matrix-based approach by allowing more partitions (and thus potentially lower $d_{\max}$). If the target Tucker ranks are all bounded constants, \eqref{eq:main_error_bound} reduces to
    \[
    \|\lambda^*-\widehat{\lambda}_{\mathrm{Tensor}}\|_{\mathbb{L}_2(\mathbb{X})}^2
        = O_p\left(\frac{\|\lambda^*\|_{\infty}^{2\alpha/(2\alpha + d_{\max})}{\|\lambda^*\|_{W_2^{\alpha}}^{2d_{\max}/(2\alpha + d_{\max})}}}{n^{2\alpha/(2\alpha+d_{\max})}}\log(n) + \xi_{(R_1, \dots, R_s)}^2\right).
    \]
    Moreover, if $\lambda^*$ is an exactly low-rank function, i.e.~the additive or mean-field functions, then $\xi_{(R_1, \dots, R_s)} = 0$.  In contrast,
   the KIE can only achieve an error rate of
   \[
    O_p\left(\frac{\|\lambda^*\|_{\infty}^{2\alpha/(2\alpha+D)}\|\lambda^*\|_{W_2^\alpha(\mathbb{X})}^{2D/(2\alpha + D)}}{n^{2\alpha/(2\alpha + D)}}\right),
    \]
    which is driven by the ambient dimensionality $D$ rather than the largest block dimension $d_{\max}$ induced by the chosen partition.

\begin{remark}[Beyond Poisson processes]\label{remark:beyond_poisson}
\textbf{Marked Poisson processes.} By the Marking Theorem \citep[Section~5.2]{last2017lectures}, a marked Poisson process on $\mathbb{X}_0\times\mathbb{M}$ is a Poisson point process on the product space whose joint intensity factorizes as $\lambda^*(x,m)=\lambda_0^*(x)\,\mu^*(m\mid x)$, where $\lambda_0^*$ is the ground intensity and $\mu^*(\cdot\mid x)$ is the conditional mark density. Given a joint estimate $\widehat\lambda(x,m)$, the two factors are recovered by marginalization and normalization, $\widehat\lambda_0(x)=\int_{\mathbb{M}}\widehat\lambda(x,m)\,\d \upsilon(m)$ and $\widehat\mu(m\mid x)=\widehat\lambda(x,m)/\widehat\lambda_0(x)$; under the tensor-product basis, the marginal is a finite linear combination of precomputable one-dimensional basis integrals.

\noindent\textbf{Beyond Poisson.} Although \Cref{thm:main_result} is stated for Poisson point processes, the same proof strategy applies whenever, for any deterministic orthonormal $W_j$ and $M_j$ of compatible sizes, the operator norm $\|W_j^\top \mathcal M_j(\widehat b - b^*) M_j\|_{\op}$ admits a deviation bound of the similar form as in the Poisson case (\Cref{lemma:deviation_op_ortho}). Once such a bound is available, the remaining subspace-perturbation and projection-error arguments carry over verbatim. For Neyman–Scott (Appendix~\ref{sec:Neymann-Scott}), the rate matches \Cref{thm:main_result} with an additional factor depending on cluster variance; for $D$-dependent time series (Appendix~\ref{sec:time_series}), an analogous rate holds with $n$ replaced by $n/(2D{+}1)$.
\end{remark}

\begin{remark}[An example on coordinate partition]\label{remark:coor_par}
    We illustrate Condition~\eqref{eq:condition_dim} for the tensor-based method with $\alpha=2$. KIE attains
    \[
    O_p\!\left(\|\lambda^*\|_{\infty}^{4/(4+D)}\|\lambda^*\|_{W_2^2(\mathbb{X})}^{2D/(4+D)}n^{-4/(4+D)}\right),
    \]
    while the matrix-based estimator (the $s=2$ partition with $d_{\max}=\lceil D/2\rceil$) attains, up to a $\log$ factor, the same expression with $D$ replaced by $\lceil D/2\rceil$. For the tensor-based estimator, the choice of partition determines $d_{\max}$ and the rate becomes, up to a $\log$ factor,
    \[
    O_p\!\left(\|\lambda^*\|_{\infty}^{4/(4+d_{\max})}\|\lambda^*\|_{W_2^2(\mathbb{X})}^{2d_{\max}/(4+d_{\max})} n^{-4/(4+d_{\max})}\right).
    \]
    The choices of $(s, d_{\min}, d_{\max})$ listed below satisfy \eqref{eq:condition_dim}.
    \begin{center}
    \small
    \setlength{\tabcolsep}{4pt}
    \renewcommand{\arraystretch}{1}
    \begin{tabular}{lcccc}
    \toprule
     & $3\le D\le 5$ & $D=6$ & $7\le D\le 10$ & $D\ge 11$ \\
    \midrule
    $s$         & $D$ & $4$ & $3$                  & $2$ \\
    $d_{\min}$  & $1$ & $1$ & $\lfloor D/3\rfloor$ & $\lfloor D/2\rfloor$ \\
    $d_{\max}$  & $1$ & $2$ & $\lceil D/3\rceil$   & $\lceil D/2\rceil$ \\
    \bottomrule
    \end{tabular}
    \end{center}
    The tensor-based estimator benefits from finer partitions when $D$ is moderate: $d_{\max}=1$ for $D\le 5$ gives the fastest rate, and $d_{\max}\approx D/3$ for $6\le D\le 10$ outperforms both the matrix-based method and KIE. For $D\ge 11$, only $s=2$ satisfies \eqref{eq:condition_dim}, in which case the tensor-based estimator reduces to the matrix-based one.
\end{remark}

\subsection{Lower bound for intensity estimation}\label{sec:low_bound}

We establishes the minimax lower bound on the estimation error in the context of nonparametric intensity estimation for inhomogeneous point processes. The bound characterizes the fundamental difficulty of the problem by demonstrating the best achievable rate of any estimator restricted to a rank-constrained function class. 

Let $\mathbb{X} = \mathbb{X}_1 \times \cdots \times \mathbb{X}_s$ and for $ \xi_{(R_1, \dots, R_s)} >0 $, define the intensity function class
\begin{align}
\label{Lambda-def}
    \Lambda_{(R_1, \dots, R_s)}^{\alpha, s} = \left\{ \lambda^* : \mathbb{X} \to \mathbb{R}_+ \;\middle|\; \|\lambda^*\|_{W_2^\alpha(\mathbb{X})} < \infty, \; \|\lambda^*\|_{\infty} < \infty, \right. \nonumber \\ 
    \left.  \; \text{and} \ \inf_{\lambda \in \mathcal{T}_{(R_1, \dots, R_s)}} \|\lambda - \lambda^*\|_{\mathbb{L}_2(\mathbb{X})} \leq \xi_{(R_1, \dots, R_s)} \right\},
\end{align}
where \[
\mathcal{T}_{(R_1, \dots, R_s)} = \left\{\lambda \in \mathbb{L}_2(\mathbb{X}): \rank(\lambda_j(x_j,x_{-j})) \leq R_j, \forall j \in [s]\right\}
\]
is the set of functions on $\mathbb{X}$, whose Tucker ranks are bounded by $(R_1, \dots, R_s)$.

The function class \( \Lambda_{(R_1, \dots, R_s)}^{\alpha, s} \) is constructed to encompass intensity functions that exhibit both a prescribed degree of smoothness, as characterized by the Sobolev space \(W_2^\alpha(\mathbb{X})\), and an upper bound on the low-rank approximation error. 

\begin{theorem}[Minimax lower bound]
\label{Minimax-Theorem-Tensor}
    Consider the function class \( \Lambda_{(R_1, \dots, R_s)}^{\alpha, s} \) defined in  \eqref{Lambda-def}. Suppose that $ \{ R_j\}_{j=1}^s$ are all bounded constants. For any estimator $ \widehat{\lambda} \in \mathcal{T}_{(R_1, \dots, R_s)}$ based on the observations \(\{N^{(i)}\}_{i = 1}^n\), we have that 
    $$
    \sup_{\lambda^* \in \Lambda_{(R_1, \dots, R_s)}^{\alpha, s}} \mathbb{E} \left[ \|\lambda^* - \widehat{\lambda}\|_{\mathbb{L}_2(\mathbb{X})}^2 \right] \geq C_{0} \left( \frac{1}{n^{2\alpha/(2\alpha + d_{\max})}}  + \xi_{(R_1, \dots, R_s)}^2 \right),
    $$
    where \(C_{0} > 0\) is a positive constant, \(d_{\max} = \max\{d_1, \dots, d_s\}\), and \(\xi_{(R_1, \dots, R_s)}\) represents an upper bound in the approximation error in the \(\mathbb{L}_2\)-norm between \(\lambda^*\) and its best rank-\((R_1, \dots, R_s)\) approximation, as defined in \eqref{eq:tucker_low-rank_approx}.
\end{theorem}

Recall from \Cref{thm:main_result} upper bound of our tensor-based estimator $\widehat{\lambda}_{\text{Tensor}}$.
If the ranks \((R_1, \dots, R_s)\) are bounded, this upper bound matches the lower bound in \Cref{Minimax-Theorem-Tensor} up to a $\log(n)$ factor. Thus, the proposed tensor-based estimator achieves the best possible convergence rate among estimators in $\mathcal{T}_{(R_1, \dots, R_s)}$ for estimating intensity functions in the class  $ \Lambda_{(R_1, \dots, R_s)}^{\alpha, s}$.


This lower bound applies to estimators targeting functions that are well approximated by Tucker low-rank tensors. In this regime, the low-rank representation captures reduced effective complexity and, together with smoothness, characterizes the intrinsic difficulty of the problem. The bound therefore shows that our estimator is rate-optimal on this class and that its performance is driven by the intended structural assumptions on \(\lambda^*\).

\section{Numerical results}\label{sec:numeric}
This section provides numerical evidence to support our theoretical results for the proposed matrix- and tensor-based estimators.  For comparison, we include the multivariate kernel intensity estimator (KIE) using a Gaussian kernel with the bandwidth automatically selected using Scott's rule. We provide Python code implementing the proposed methods in the supplementary materials. For full code and reproducibility details, see our \href{https://github.com/HaotianXu/Multivariate-Poisson-intensity-estimation-via-low-rank-tensor-decomposition}{GitHub repository}.

\subsection{Data simulation and setup}\label{sec:simu_setup}
We simulate point processes from various intensity functions $\lambda^*$ on $\mathbb{X}=[0,1]^D$, with the dimension $D$ varying from $2$ to $6$. The four scenarios considered in this section are motivated by practical spatial point process applications and span three structural regimes: a clustered point process with random kernel centers (\textbf{S1}), a random-intensity model (\textbf{S2}), and deterministic smooth intensities with explicit finite Tucker rank (\textbf{S3}, \textbf{S4}).
\begin{enumerate}
    \item[\textbf{S1.}] (Neyman--Scott clustered point process.) The intensity is a sum of $K$ isotropic Gaussian kernels centered at randomly drawn parents:
$\lambda^*(x) = (2\pi\sigma^2)^{-D/2} \sum_{c \in \mathcal{C}} \exp\left(-\tfrac{\|x - c\|_2^2}{2\sigma^2}\right)$,
where $\mathcal{C} = \{c_1, \dots, c_K\}$ are i.i.d.~$\mathrm{Uniform}([0,1]^D)$ parents with $K \sim \mathrm{Poisson}(\kappa)$. We use $\kappa = 30$, $\sigma = 0.35$. A new realization of the parents is drawn for each Monte Carlo repetition.
    \item[\textbf{S2.}] (Log-Gaussian Cox process.) The intensity is $\lambda^*(x) = \exp\left(Y(x) - \sigma_Y^2/2\right)$,
where $Y$ is a centered Gaussian random field on $[0,1]^D$ with covariance kernel $k(x,x') = \sigma_Y^2 \exp\bigl(-12.5 \|x - x'\|_2^2\bigr)$, $\sigma_Y^2 = 0.5$. A new realization of $Y$ is drawn for each Monte Carlo repetition (using a random-Fourier-feature approximation with $1000$ frequencies).
    \item[\textbf{S3.}] (Gaussian mixture.) Three isotropic Gaussian bumps centered on the main diagonal:
$\lambda^*(x) = 5 \sum_{c \in \{0.2, 0.5, 0.8\}} \exp\left(-\frac{\|x - c\,\mathbf{1}_D \|_2^2}{0.32}\right)$,
where $\mathbf{1}_D$ is the all-ones vector in $\mathbb{R}^D$.

    \item[\textbf{S4.}] (Ginzburg--Landau-type smooth intensity.) A smooth deterministic intensity that is \emph{not} exactly low-rank:
$\lambda^*(x) = \exp\Bigl(-\tfrac{(D+1)^2}{800}\sum_{i=1}^{D-1}(x_i - x_{i+1})^2 \;-\; \tfrac{5}{32}\sum_{i=1}^{D}(x_i^2 - 1)^2\Bigr)$.
\end{enumerate}
For each scenario and each $D \in \{2,3,4,5,6\}$, we simulate $n = 10^5$ i.i.d.~point processes with intensity $\lambda^*$ on $\mathbb{X} = [0,1]^D$. Performance is measured by the relative $\mathbb{L}_2$ error
\[
\text{Relative Error} = \frac{\| \widehat{\lambda}(\text{test set}) - \lambda^*(\text{test set}) \|_{\mathbb{L}_2(\mathbb{X})}}{\| \lambda^*(\text{test set}) \|_{\mathbb{L}_2(\mathbb{X})}},
\]
evaluated on a regular test grid with $6^D$ points. Each reported result is the average over $100$ Monte Carlo repetitions. Three additional scenarios, i.e.~a piecewise-constant step intensity and two cosine-series families with polynomial and exponential coefficient decay, are reported in \Cref{add-num-res}.

\subsection{Coordinate partition and rank selection}\label{sec:coor_part_cluster}

We partition the $D$-dimensional input space $\mathbb{X}$ into $s$ clusters using a simple clustering procedure, based on the empirical covariance matrix, that groups coordinates with higher pairwise correlations into the same cluster.  Each cluster $\mathbb{X}_j$ has dimension $d_j$, such that $\sum_{j = 1}^s d_j = D$. 
{For each cluster, we construct a tensor-product basis built from $m$ orthonormal local hat functions (continuous piecewise-linear B-splines on a uniform mesh of $[0,1]$) in each coordinate, yielding $m^{d_j}$ basis functions per cluster.
We adopt this basis because its compact support localizes each coefficient to a small neighborhood, which makes the estimator robust to spatially heterogeneous and non-smooth intensities and avoids the boundary oscillations that high-degree global polynomial bases exhibit at large~$m$.}

For the matrix case (\(s = 2\)), we perform SVD on the empirical coefficient matrix \(\widehat{b}\).  Soft-thresholding the singular values yields a low-rank matrix approximation $T_{\gamma}(\widehat{b})$.  The threshold parameter \(\gamma\) is selected through cross-validation: We partition $\{N^{(i)}\}_{i = 1}^n$ into $k$ folds. For each round, one fold is designated as the testing set, and the remaining $k-1$ folds are as the training set. We compute $\widehat{b}$ on each training set. Applying different $\gamma$ values on a grid of $50$ equally spaced points in $[0,5]$ and choosing $\gamma$ minimizing the average relative error on the testing fold.

For the tensor case (\(s \geq 3\)), we compute the empirical coefficient tensor \(\widehat{b}\).  To adaptively select the target Tucker rank $(R_1, \dots, R_s)$, for each mode-$j$ matricization $\mathcal{M}_j(\widehat{b})$, we perform SVD and monitor the consecutive singular value ratios $\rho_k^{(j)} = \sigma_k^{(j)}/\sigma_{k+1}^{(j)}$. We choose the largest index $k$ such that $\rho_k^{(j)} > \tau$ and set the rank $R_j = k+1$. This data-driven approach ensures that only significant singular values are retained.

\subsection{Summary of the results}
\label{summary-num-results}
We vary $m$ over $\{4, 6, 8\}$ and report all three values jointly in Figure~\ref{fig:simu_main}, which compares the proposed Matrix/Tensor estimators with the KIE across the four scenarios for all $D \in \{2,3,4,5,6\}$. For each pair $(m, D)$, the partition size $s \in \{2,\ldots,D\}$ is chosen to minimize the average relative error over the $100$ Monte Carlo replicates; the resulting $s$ is reported in Table~\ref{tab:argmin_s}. Markers are means over the $100$ replicates and vertical bars span Monte Carlo standard deviation.

\begin{figure}[!htbp]
\centering
\includegraphics[width=0.85\textwidth]{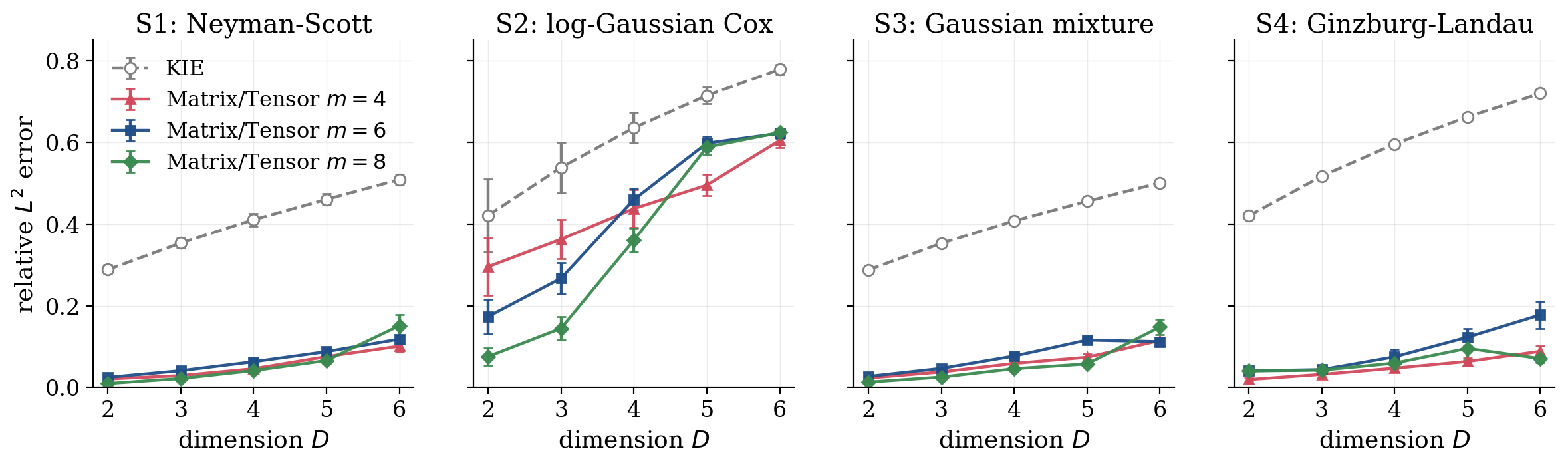}
\caption{Relative $\mathbb{L}_2$ error vs.\ dimension $D$ for the scenarios \textbf{S1}-\textbf{S4}. KIE is shown as a dashed grey curve; the Matrix/Tensor estimator is shown for $m \in \{4, 6, 8\}$, in each case with the partition size $s$ that minimizes the averaged error (see Table~\ref{tab:argmin_s}). Markers are averaged errors over $100$ Monte Carlo replicates; vertical bars show standard deviation. In S1 the parent locations are re-drawn each replicate and in S2 the log-Gaussian field is re-drawn each replicate, so the truth is itself random and the error bars also reflect this between-replicate variation; in S3 and S4 the truth is deterministic.}
\label{fig:simu_main}
\end{figure}

\begin{table}[ht]
\centering
\small
\setlength{\tabcolsep}{6pt}
\renewcommand{\arraystretch}{1}
\begin{tabular}{c|ccc}
\toprule
 & $m=4$ & $m=6$ & $m=8$ \\
\midrule
S1 & $2,\,2,\,2,\,5,\,6$ & $2,\,3,\,4,\,5,\,6$ & $2,\,3,\,4,\,5,\,6$ \\
S2 & $2,\,2,\,2,\,2,\,2$ & $2,\,2,\,2,\,5,\,6$ & $2,\,2,\,2,\,3,\,5$ \\
S3 & $2,\,2,\,4,\,3,\,5$ & $2,\,3,\,4,\,4,\,6$ & $2,\,3,\,4,\,5,\,6$ \\
S4 & $2,\,3,\,4,\,4,\,5$ & $2,\,3,\,4,\,4,\,5$ & $2,\,3,\,4,\,5,\,6$ \\
\bottomrule
\end{tabular}
\caption{Partition size $s$ that minimizes the average (over $100$ Monte Carlo replicates) relative $\mathbb{L}_2$ error of the Matrix/Tensor estimator at each $(\text{scenario}, m, D)$. Each cell lists the selected $s$ at $D = 2, 3, 4, 5, 6$ in that order.}
\label{tab:argmin_s}
\end{table}

\noindent\textbf{Outperformance.} Across all four scenarios and all dimensions, the Matrix/Tensor estimator outperforms KIE. For the structured scenarios (S1, S3, S4) the relative-error ratio of KIE to the best Matrix/Tensor configuration is at least $4$--$5\times$ across all dimensions.

\noindent\textbf{Robustness to the choice of $m$.} The three Matrix/Tensor curves ($m \in \{4, 6, 8\}$) are nearly indistinguishable in S1, S3, and S4, indicating that the estimator is robust to the per-coordinate basis size in regimes where genuine multilinear structure is present. In S2 the curves are more spread out, but all three remain comfortably below KIE.

\noindent\textbf{Selected $s$ tracks $D$ for structured truths.} Table~\ref{tab:argmin_s} shows that for S1, S3, and S4 the selected partition size grows with the dimension and reaches $s = D$ for $D \geq 4$ when $m \in \{6, 8\}$, indicating that the full multilinear factorization is genuinely beneficial. For $m = 4$ the selected $s$ is occasionally smaller, reflecting that a coarser per-coordinate basis cannot resolve the additional structure exposed by larger $s$. In S2, by contrast, the rule prefers $s = 2$ at low $D$ and grows only modestly with $D$: when the underlying truth admits no genuine multilinear structure, the spectral-gap rule honestly returns a low-rank decomposition rather than overfitting.

\noindent\textbf{The log-Gaussian Cox case.} S2 is the only scenario in which the intensity itself is random (resampled per replicate), so the standard deviation of the KIE error is visibly larger and reflects this irreducible variation. Because the truth has no genuine multilinear factorization, the gain over KIE is more modest, i.e.~approximately $2.4\times$ at $D=2$ shrinking to $1.2$--$1.4\times$ at $D \geq 4$, but the proposed method never loses to KIE in this scenario, demonstrating that the Matrix/Tensor decomposition degrades gracefully when its low-rank assumption is violated.

In supplementary material, we report (i)~the additional three scenarios in \Cref{add-num-res} and (ii)~companion timing analyses \Cref{sec:computation_cost}. In particular, an empirical runtime study is given in \Cref{sec:runtime_empirical}.

\subsection{Real data application: Earthquakes in the U.S.}\label{sec:real_data}

We further apply the methods to a real dataset obtained from the U.S.~Geological Survey Earthquake Catalog\footnote{The data set is publicly available at \url{https://earthquake.usgs.gov/earthquakes/search/}.}. The filtered dataset contains records of 101,929 earthquakes in the conterminous United States over the period from 1990-01-01 to 2025-01-01, with $D=4$ attributes (latitude, longitude, depth, and magnitude). For the empirical analysis, we aggregate the catalog at the daily level and treat each day as one realization of a point pattern on this four-dimensional attribute space, yielding 12,670 daily point patterns. The original geographical bounds, with latitude range \([24.6,50]\) and longitude range \([-125,-65]\), are normalized to \([0,1]\times[0,1]\). Similarly, depth and magnitude are normalized to \([0,1]\).

Since no ground truth $\lambda^*$ is available, we evaluate the fitted intensities using two complementary forms of evidence. First, we examine longitude--latitude marginal intensity maps, which provide a direct visual comparison with the held-out empirical spatial distribution. Second, we support this qualitative comparison with a distributional metric based on the sliced Wasserstein-2 distance between samples generated from each fitted estimator and the held-out test data.

In this setting, we compare \(\widehat{\lambda}_{\mathrm{Mat/Ten}}(s)\) for \(s \in \{2,3,4\}\) with \(\widehat{\lambda}_{\mathrm{KIE}}\). 
Our analysis targets the first-order intensity function $\lambda^*(u)$, which represents the average event density over time and does not condition on past event history. Accordingly, for the purpose of estimating this first-order intensity, we treat the daily point patterns as approximately exchangeable units, which allows us to use standard sample splitting and held-out evaluation.
In particular, the data is divided into training ($75\%$) and testing ($25\%$) sets using $30$ random splits at the level of daily point processes, where each day is treated as one realization on the attribute space (latitude, longitude, depth, magnitude). For each split, the days are randomly permuted and assigned to the training and test sets, and averaging over splits reduces variability due to data partitioning.

To place this modeling choice in the context of standard seismological approaches, in \Cref{sec:ETAS_earthquake} of the supplementary material we relate our framework to the widely used Epidemic-Type Aftershock Sequence (ETAS) model. The discussion contrasts the two frameworks' inferential goals (history-dependent conditional intensity in ETAS versus the joint history-averaged first-order intensity here) and their modeling structures, and identifies a shared nonparametric ingredient: although recovering the 2D spatial intensity is not the primary target of our method, it can be obtained as a marginal of the joint estimator and, via the marked Poisson decomposition in \Cref{remark:beyond_poisson}, corresponds to the spatial background that ETAS-based approaches (e.g.~\citealp{zhuang2002stochastic}) estimate nonparametrically.

We begin with the spatial marginal comparison. Figure~\ref{fig:intensity_plots_combined} displays the longitude--latitude marginal intensity maps for \(\widehat{\lambda}_{\mathrm{Mat/Ten}}\) with $m = 10$ and \(s\in\{2,3,4\}\), together with the corresponding map for \(\widehat{\lambda}_{\mathrm{KIE}}\) and an empirical reference map \(\lambda_{\mathrm{Test}}\) computed from the held-out test events. For \(\widehat{\lambda}_{\mathrm{Mat/Ten}}\), the displayed map is the marginal intensity over longitude and latitude obtained by integrating out depth and magnitude. The empirical reference \(\lambda_{\mathrm{Test}}\) is constructed by pooling the held-out earthquake events, binning them over the longitude--latitude grid, and normalizing by the number of test days and the grid-cell area. To produce \Cref{fig:intensity_plots_combined}, each coordinate (latitude, longitude, depth, magnitude) is first rescaled to $[0,1]$, the estimators are fit on the normalized $[0,1]^4$ scale, and the fitted intensities are then mapped back to the original longitude--latitude domain for display

\begin{figure}[!htbp]
\centering
\includegraphics[width=0.8\textwidth]{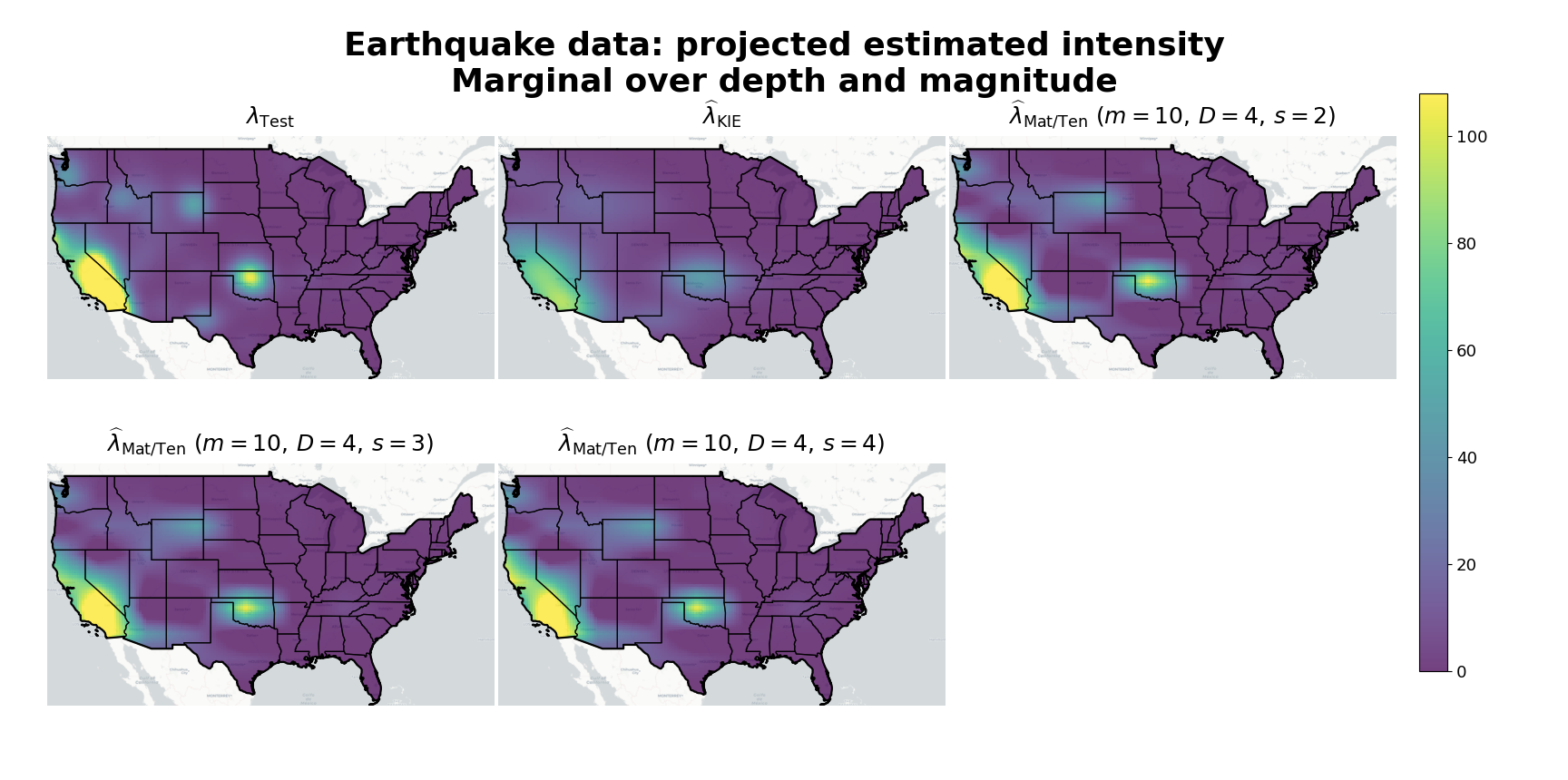}
\caption{Longitude--latitude marginal intensity maps for the earthquake application. The top row shows the empirical reference map \(\lambda_{\mathrm{Test}}\) computed from the held-out test events (left), \(\widehat{\lambda}_{\mathrm{KIE}}\) (center) and \(\widehat{\lambda}_{\mathrm{Mat/Ten}}\) at \(s=2\) (right). The bottom row shows \(\widehat{\lambda}_{\mathrm{Mat/Ten}}\) at \(s=3\) (left) and \(s = 4\) (right).  All displayed maps correspond to marginal intensities over longitude and latitude obtained by integrating out depth and magnitude.}
\label{fig:intensity_plots_combined}
\end{figure}

The empirical reference map shows several concentrated regions of seismic activity, most prominently along California and the broader West Coast, as well as a localized hotspot in Oklahoma. These main features are well captured by \(\widehat{\lambda}_{\mathrm{Mat/Ten}}\) across all values of \(s\). In these maps, the high-intensity region along California exhibits an elongated, oval-shaped structure that closely matches the spatial pattern observed in the empirical test map.

In contrast, \(\widehat{\lambda}_{\mathrm{KIE}}\) substantially oversmooths the intensity and does not clearly recover the dominant localized structures visible in the empirical test map. In particular, it diffuses the high-intensity region along California and the broader West Coast, an area associated with major tectonic boundaries such as the San Andreas fault system \citep{bird2003plate, usgs_hazards}. As a result, the mass is spread over a wider area and the peak intensity observed in the test map is attenuated. A similar effect is observed for the Oklahoma hotspot, where the kernel-based estimate smooths out the localized peak and fails to clearly distinguish it from the surrounding regions. This area is known to exhibit elevated seismicity in recent years, largely attributed to induced seismic activity \citep{keranen2014sharp}.

In addition, the empirical reference map reveals smaller regions of elevated activity in the north-central United States, around Wyoming and extending into South Dakota (and parts of southern Montana), as well as in the Pacific Northwest (notably in Washington state). These localized features are captured by \(\widehat{\lambda}_{\mathrm{Mat/Ten}}\) but are largely missed by \(\widehat{\lambda}_{\mathrm{KIE}}\), further illustrating the ability of \(\widehat{\lambda}_{\mathrm{Mat/Ten}}\) to recover fine-scale spatial structure that is not preserved under global kernel smoothing. These regions are consistent with known seismic zones in the United States, including activity in the Intermountain West and along the Cascadia subduction zone \citep{ goldfinger2012turbidite, usgs_hazards}.

We next support the visual comparison with a distributional metric computed on the full four-dimensional normalized attribute space. For each fitted intensity estimator \(\widehat{\lambda}\), we generate a sample \(\widehat{X}\) of the same size as the held-out test sample \(X_{\mathrm{test}}\), and define the sliced Wasserstein-2 distance
\[
\mathrm{SW}_2(\widehat{\lambda},\,\text{test})
:=
\mathrm{SW}_2\bigl(\widehat{X},\,X_{\mathrm{test}}\bigr),
\]
where \(\mathrm{SW}_2(\cdot,\cdot)\) denotes the sliced Wasserstein-2 distance between empirical distributions. All coordinates are normalized to \([0,1]^4\), and the distance is computed using a fixed number of random projections. Smaller values indicate closer agreement between the distribution of samples generated from the fitted intensity and the empirical test distribution.

\begin{table}[ht]
\centering
\small
\scalebox{0.9}{
\setlength{\tabcolsep}{6pt}
\renewcommand{\arraystretch}{0.6}
\begin{tabular}{lcccc}
\toprule
 & $\widehat{\lambda}_{\mathrm{Mat/Ten}}(s=2)$ & $\bm{\widehat{\lambda}_{\mathrm{Mat/Ten}}(s=3)}$ & $\widehat{\lambda}_{\mathrm{Mat/Ten}}(s=4)$ & $\widehat{\lambda}_{\mathrm{KIE}}$ \\
\midrule
$\mathrm{SW}_2(\widehat{\lambda},\,\text{test})$ & $0.2099$ & $\textbf{0.2095}$ & $0.2394$ & $0.2581$ \\
\bottomrule
\end{tabular}}
\caption{Sliced Wasserstein-2 ($\mathrm{SW}_2$) distance between samples generated from each fitted estimator and the held-out test set, averaged over $30$ random splits. Smaller values indicate closer agreement between the distribution of samples generated from the fitted intensity and the empirical test distribution.}
\label{tab:earthquake_sw2}
\end{table}

Table~\ref{tab:earthquake_sw2} provides quantitative support for the visual comparison in Figure~\ref{fig:intensity_plots_combined}. The Matrix/Tensor estimators achieve smaller sliced Wasserstein-2 distances than KIE, indicating closer distributional agreement with the held-out earthquake data. Among the Matrix/Tensor configurations, \(s=2\) and \(s=3\) perform similarly and yield the smallest distances, while \(s=4\) exhibits a modest increase. This suggests that, although the Matrix/Tensor estimates are broadly stable across partition sizes, overly fine partitioning may introduce additional variability without improving held-out distributional fit.

Overall, the Matrix/Tensor estimators, \(\widehat{\lambda}_{\mathrm{Mat/Ten}}\), provide stable fitted intensities across partition sizes, closer agreement with the held-out empirical distribution under the \(\mathrm{SW}_2\) metric, and clearer recovery of localized spatial features in the longitude--latitude marginal maps. In contrast, the kernel estimator tends to oversmooth the intensity surface, attenuating the dominant California and Oklahoma hotspots as well as smaller localized regions.

Additional real data applications are provided in Appendix~\ref{sec:add-real-data}. 
These examples complement the national earthquake analysis considered here by examining different data structures and levels of aggregation. 
Appendix~\ref{sec:real_data_tornado} considers NOAA tornado events represented through four-dimensional start--end locations, treating each event as a trajectory in spatial coordinates. 
Appendix~\ref{sec:real_data_eq_regional} revisits earthquake data in a localized California--Nevada region and analyzes the catalog directly at the event level rather than through daily aggregation. 
This regional analysis highlights fine-scale localized seismic structure that is partially obscured in the national, temporally aggregated setting. 
Together, these examples show that the proposed Matrix/Tensor estimators remain effective across both aggregated point-pattern data and raw event-level observations.

\section{Conclusion}
In this paper, we introduced novel methods for estimating multivariate intensity functions in inhomogeneous point processes by utilizing low-rank matrix or tensor decompositions. By exploiting the approximately low-rank structures of square-integrable multivariate functions, our approaches effectively mitigate the curse of dimensionality both theoretically and computationally. We developed new theoretical tools to rigorously justify the statistical performance of our estimators, providing, to the best of our knowledge, the first statistical analysis of approximately low-rank tensor estimation. The error bounds on our proposed estimators expose an interesting bias-variance trade-off controlled by the user-specified approximation model's complexity (ranks), paralleling the trade-offs commonly seen in other approximate inference frameworks, e.g.~variational inference. Furthermore, this work represents the first application of matrix and tensor decompositions for intensity function estimation, opening new avenues for research in multidimensional point processes.

\newpage
\bibliographystyle{apalike}
\bibliography{ref}

\clearpage
\newpage
\appendix

\section*{Appendices}

A quick-reference list of the most frequently used symbols is provided in \Cref{tab:notation_glossary}. A concise literature review, additional simulation studies, and all relevant technical details are provided in the supplementary material. Specifically, \Cref{sec:lit_rev} offers an overview of nonparametric intensity function estimation, tensor network approximation, and low-rank tensor estimation. \Cref{add-num-res} presents further numerical results, reinforcing our primary findings and examining the robustness of our algorithms concerning the selected tuning parameters. \Cref{sec:add-real-data} presents additional real data analyses, i.e.~U.S.~tornado and regional earthquake in California–Nevada, to further evaluate the proposed Matrix/Tensor estimators.  Examples demonstrating functions exhibiting low-rank properties are given in \Cref{sec:example_low-rank}. Detailed information about Sobolev spaces and RKHS bases is discussed in \Cref{sec:approximation_theory}, while computational considerations of the proposed methods are covered in \Cref{sec:computation_cost}. Statistical guarantees for our methods under Poisson point processes and other types of point processes are rigorously proven in Appendices~\ref{sec:proofs} and \ref{sec:other_pp}, respectively. Lastly, auxiliary technical tools utilized in the paper are comprehensively documented in Appendices~\ref{sec:approx_low_rank_tensor} and \ref{sec:HilbertSpace}.

\begin{table}[!htbp]
\centering
\small
\setlength{\tabcolsep}{5pt}
\renewcommand{\arraystretch}{1}
\begin{tabular}{ll}
\toprule
Symbol & Meaning \\
\midrule
$D,\,s,\,d_j$ & ambient dimension; number of coordinate blocks; size of the $j$th block ($\sum_{j=1}^s d_j = D$) \\
$\mathbb X = \mathbb X_1\times\cdots\times \mathbb X_s$ & product domain with $\mathbb X_j \subset \mathbb R^{d_j}$ \\
$\lambda^*$ & target intensity function on $\mathbb X$ \\
$\{N^{(i)}\}_{i=1}^n$ & $n$ i.i.d.\ point processes observed on $\mathbb X$ \\
$\{\phi_{j,\mu_j}\}_{\mu_j=1}^{m^{d_j}}$ & orthonormal basis of $\mathbb L_2(\mathbb X_j)$; tensor products span the truncated $\mathbb L_2(\mathbb X)$ subspace \\
$b^*,\,\widehat b$ & population and empirical coefficient tensors in $\mathbb R^{m^{d_1}\times\cdots\times m^{d_s}}$ \\
$\mathcal M_j(\cdot)$ & mode-$j$ matricization (unfolding) of a tensor \\
$R,\,(R_1,\dots,R_s)$ & rank for the matrix-based method; target Tucker rank for the tensor-based method \\
$\alpha$ & Sobolev smoothness of $\lambda^*$ in $W_2^\alpha(\mathbb X)$ \\
$\xi_{(R)},\,\xi_{(R_1,\dots,R_s)}$ & best rank-$R$ / Tucker-rank-$(R_1,\dots,R_s)$ $\mathbb L_2$ approximation error to $\lambda^*$ \\
\bottomrule
\end{tabular}
\caption{Frequently used symbols. Detailed definitions appear in \Cref{sec:notation}.}
\label{tab:notation_glossary}
\end{table}

\section{A short literature review}\label{sec:lit_rev}
\noindent \textbf{Nonparametric intensity function estimation.} 
Classical nonparametric methods for intensity estimation are typically categorized as either kernel-based or projection-based estimators. Existing approaches within these categories focus on different aspects of nonparametric estimation, such as bandwidth selection \cite[e.g.][]{diggle1985kernel,cronie2018non,davies2018fast,van2020infill,van2024non}, choosing the number of basis functions, e.g.~Wavelet, Fourier or spline, in a way that adapts to the unknown smoothness of intensity functions \cite[e.g.][]{reynaud2003adaptive,willett2007multiscale,kroll2016concentration}, penalizing the number of basis functions or number of knots for spline-based estimators \cite[e.g.][]{choiruddin2018convex,schneble2022intensity} and Bayesian nonparametric approaches \cite[e.g.][]{taddy2012mixture,kang2014bayesian}. Recently, \cite{ward2023nonparametric} studied kernel-based estimators for Poisson point processes on a Riemannian manifold, and \cite{cronie2024cross} developed a cross-validation-based theory for point processes and applied it to kernel estimators.
Other methods exist but are often limited to specific point processes \cite[e.g.][]{cunningham2008fast,guan2008consistent,waagepetersen2009two,flaxman2017poisson}.   The approach that is most related to our method is based on Low-rank matrix/tensor approximation. In particular, \cite{miller2014factorized} use the non-negative matrix factorization to analyze $2$D intensity surfaces in basketball shot data. Relatedly, \cite{dunlavy2025poisson} used a low-rank Poisson CP decomposition of histogram count tensors, motivated by a spatial Poisson-process interpretation of the bins, for multivariate density and entropy estimation; this differs from our direct intensity-estimation framework for spatial point processes. \cite{luo2026online} also adopt low-rank matrix representations of multivariate Poisson intensities, but as a building block for online change-point detection rather than for direct intensity estimation.

To our knowledge, methods that effectively address the curse of dimensionality in nonparametric estimation of multidimensional intensity function are lacking. In fact, all the above-mentioned methods struggle in these settings, they not only suffer from rapidly growing estimation errors as the number of dimensions increases, but they are also computationally demanding and do not scale well with multidimensional data.

From a theoretical perspective, intensity estimation is often examined under two main asymptotic regimes. In the increasing-domain regime \cite[e.g.][]{guan2007thinned,baddeley2014logistic}, the domain over which points are observed expands as the sample grows. Conversely, in the infill regime \cite[e.g.][]{waagepetersen2007estimating,choiruddin2018convex}, the domain remains fixed, but the number of points within it increases. This work focuses on the latter regime, and we provide nonasymptotic analysis of our proposed estimators.

\medskip
\noindent \textbf{Tensor network approximation and low-rank tensor estimation.} Our approach intends to address the curse of dimensionality and is closely related to recent advances in tensor network representations for high-dimensional machine learning and statistical modeling, such as tensor train \citep{hur2023generative}, tensor ring \citep{khoo2017efficient} and tree/hierarchical tensor network \citep{tang2022generative,peng2023generative}. In particular, we adopt the Tucker decomposition, a specific type of the tensor network, to approximate an multidimensional intensity function with the model's complexity governed by the user-specified Tucker ranks. 

To perform low-rank estimation, we build on existing methods. In the matrix setting, techniques such as singular value thresholding (SVT) are well established \citep{chatterjee2015matrix,shah2016stochastically}. For tensors, methods including higher-order singular value decomposition \citep{de2000multilinear} and higher-order orthogonal iteration \citep{de2000best} have been extensively studied  only  in finite dimension.

Two key limitations of these tensor network approaches are that they assume the target tensor is finite-dimensional and exactly low-rank. We overcome these by developing new tools to handle infinite-dimensional Hilbert space where the target function is only approximately low-rank (in the Tucker sense), ensuring that our method remains robust and effective even when the ideal low-rank structure is only approximate.

\section{Additional Simulation Results}
\label{add-num-res}

This section reports the simulation results for the three scenarios excluded from the main paper for space, alongside a companion best-$s$ table. The simulation setup, sample size, test grid, and number of Monte Carlo replicates all follow those in \Cref{sec:simu_setup}.

\subsection{Additional scenarios}
\label{sec:supp_scenarios}
The three additional scenarios are:
\begin{enumerate}
    \item[\textbf{S5.}] (Piecewise-constant step intensity.) Three plateaus along the coordinate-mean direction:
\[
\lambda^*(x) = \begin{cases}
0.85, & \bar{x} \leq 1/3,\\
1.00, & 1/3 < \bar{x} \leq 2/3,\\
1.15, & \bar{x} > 2/3,
\end{cases}
\qquad \text{where} \quad \bar{x} = \frac{1}{D}\sum_{i=1}^{D} x_i.
\]
This intensity is non-smooth and tests whether the local hat basis can resolve discontinuities.

    \item[\textbf{S6.}] (Cosine series with polynomial coefficient decay.)
\[
\lambda^*(x) = 1 + \sum_{k=1}^{K} \sigma_k \prod_{i=1}^{D} \sqrt{2}\cos(\pi k\, x_i),
\qquad \sigma_k \propto k^{-2},
\]
with $K = 10$ and $\sum_{k=1}^{K} |\sigma_k|$ normalized to $0.5\,/\,2^{D/2}$ to keep $\lambda^*$ positive.

    \item[\textbf{S7.}] (Cosine series with exponential coefficient decay.) Same as \textbf{S6} but with $\sigma_k \propto e^{-k}$.
\end{enumerate}

\subsection{Numerical results for the additional scenarios}
\label{sec:supp_results}

Figure~\ref{fig:simu_supp} reports the relative $\mathbb{L}_2$ error of the proposed Matrix/Tensor estimator (for $m \in \{4, 6, 8\}$, with the partition size $s$ selected to minimize the rep-mean error at each $(m, D)$) and of the kernel intensity estimator (KIE) across the three additional scenarios. Table~\ref{tab:argmin_s_supp} reports the selected $s$ for every $(\text{scenario}, m, D)$.

\begin{figure}[ht]
\centering
\includegraphics[width=\textwidth]{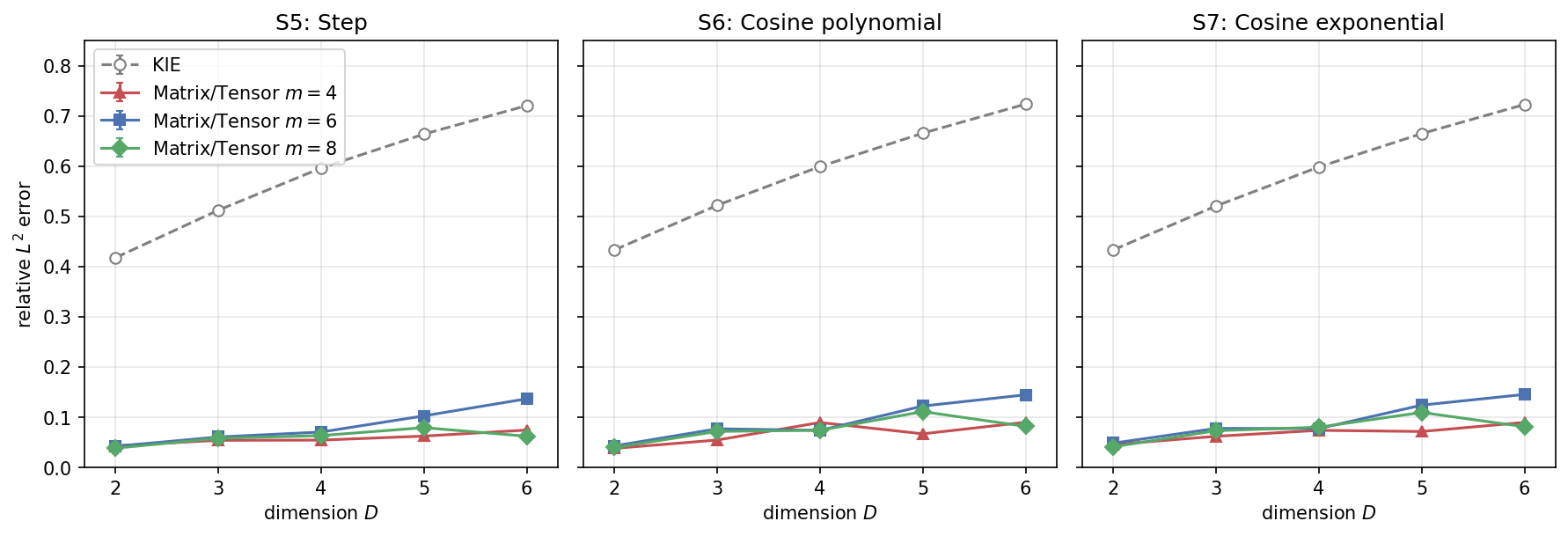}
\caption{Relative $\mathbb{L}_2$ error vs.\ dimension $D$ for the three supplementary scenarios. KIE is shown as a dashed grey curve; the Matrix/Tensor estimator is shown for $m \in \{4, 6, 8\}$, in each case at the partition size $s$ that minimizes the averaged error. Markers are means over $100$ Monte Carlo replicates; vertical bars show the standard deviation.}
\label{fig:simu_supp}
\end{figure}

\begin{table}[ht]
\centering
\small
\setlength{\tabcolsep}{8pt}
\begin{tabular}{c|ccc}
\toprule
 & $m=4$ & $m=6$ & $m=8$ \\
\midrule
S5 & $2,\,3,\,4,\,4,\,5$ & $2,\,3,\,4,\,4,\,5$ & $2,\,3,\,4,\,5,\,6$ \\
S6 & $2,\,2,\,4,\,5,\,5$ & $2,\,3,\,4,\,5,\,5$ & $2,\,3,\,4,\,5,\,6$ \\
S7 & $2,\,2,\,4,\,5,\,5$ & $2,\,3,\,4,\,5,\,5$ & $2,\,3,\,4,\,5,\,6$ \\
\bottomrule
\end{tabular}
\caption{Partition size $s \in \{2, \ldots, D\}$ that minimizes the rep-mean relative $\mathbb{L}_2$ error of the Matrix/Tensor estimator for the supplementary scenarios. Each cell lists the selected $s$ at $D = 2, 3, 4, 5, 6$ in that order.}
\label{tab:argmin_s_supp}
\end{table}

The qualitative findings mirror those in the main paper. The Matrix/Tensor estimator outperforms KIE at every $(\text{scenario}, D)$, with the largest gains at low to moderate $D$ for the smooth structured scenarios (\textbf{S6} and \textbf{S7}) and meaningful gains for the non-smooth step intensity (\textbf{S5}) at every $D$. As in the main paper, the three Matrix/Tensor curves ($m \in \{4, 6, 8\}$) are close in most cells, again indicating robustness to the choice of $m$.

\section{Additional Real Data Applications}
\label{sec:add-real-data}

In this section, we present additional real data analyses to further evaluate the proposed Matrix/Tensor estimators beyond the earthquake application in \Cref{sec:real_data}. These datasets differ in their structure and interpretation, providing complementary evidence of the method’s ability to capture multivariate spatial patterns and localized intensity features.

\subsection{Tornado events in the U.S.}
\label{sec:real_data_tornado}

We apply the proposed methods to tornado events obtained from the NOAA Storm Events Database, available at \url{https://www.ncei.noaa.gov/pub/data/swdi/stormevents/csvfiles/}. We use records from 2020--2024 and retain events with valid starting and ending spatial coordinates. The resulting dataset contains 7,832 tornado events. The number of retained events by year is 1,253 in 2020, 1,536 in 2021, 1,384 in 2022, 1,522 in 2023, and 2,137 in 2024.

Each tornado event is represented as a point in the four-dimensional space
\[
(\text{start longitude},\text{start latitude},\text{end longitude},\text{end latitude}),
\]
which encodes its trajectory. The coordinate ranges before normalization are longitude from $-149.66$ to $-65.88$ and latitude from $18.17$ to $61.13$, and all coordinates are rescaled to $[0,1]^4$ prior to estimation.

In contrast to the earthquake application in \Cref{sec:real_data}, which aggregates events at the daily level and treats each day as a point-pattern realization, the tornado dataset is analyzed at the event level: each tornado contributes one observation in the four-dimensional space. This corresponds to estimating a joint first-order intensity over trajectory endpoints, capturing both the spatial distribution of tornado origins and their displacement patterns.

Since no ground truth intensity is available, we evaluate the fitted estimators using two complementary forms of evidence. First, we examine longitude--latitude marginal intensity maps for both starting and ending locations, which provide a direct visual comparison with the held-out empirical trajectory distribution. Second, we support this qualitative comparison using the sliced Wasserstein-2 distance between samples generated from each fitted estimator and the held-out test data.

Figure~\ref{fig:tornado_maps} displays longitude--latitude marginal intensity maps for both starting and ending locations, obtained by integrating out the complementary coordinates. 
The empirical reference map $\lambda_{\mathrm{Test}}$ is constructed by pooling the held-out tornado events, binning them over the longitude--latitude grid, and normalizing by the total number of test observations and the grid-cell area.

\begin{figure}[]
\centering
\includegraphics[width=1\textwidth]{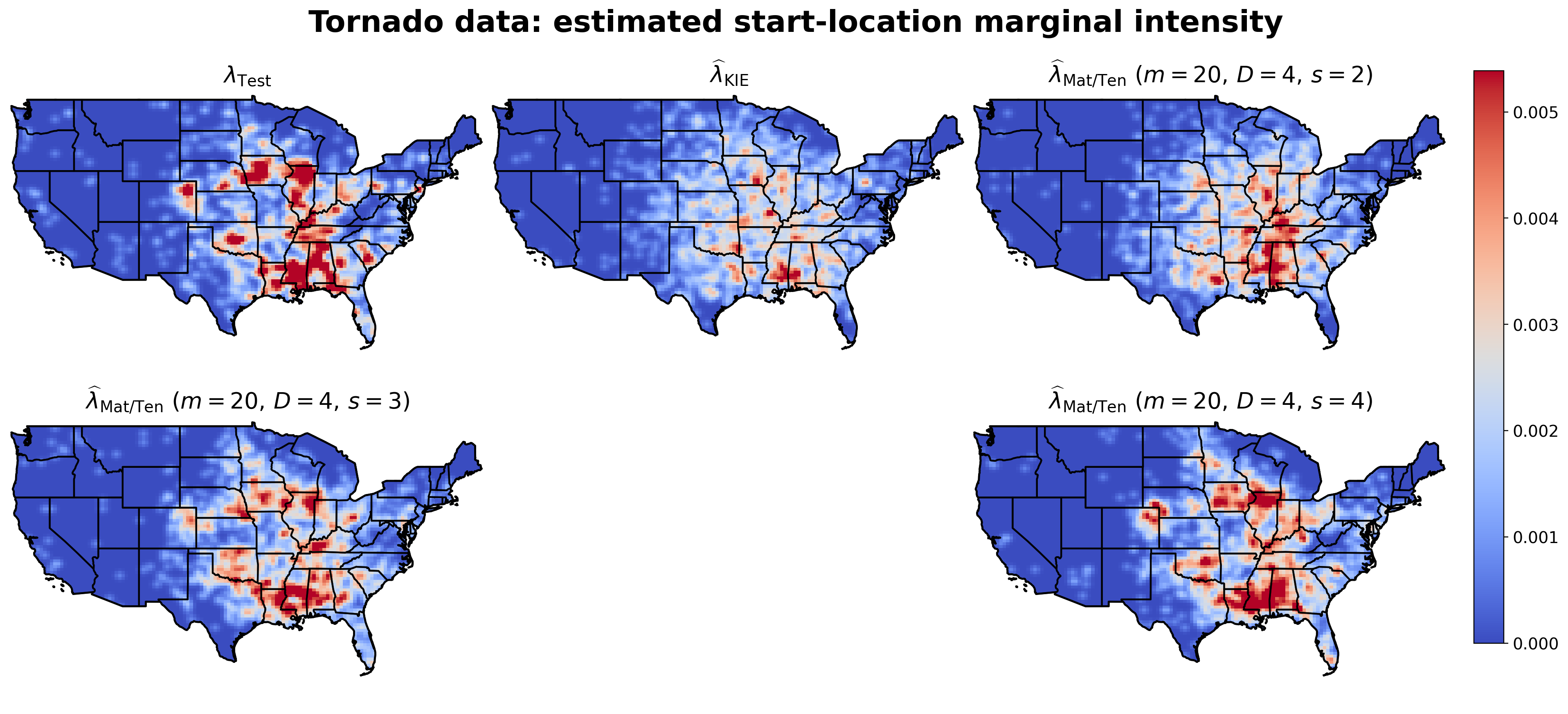}

\includegraphics[width=1\textwidth]{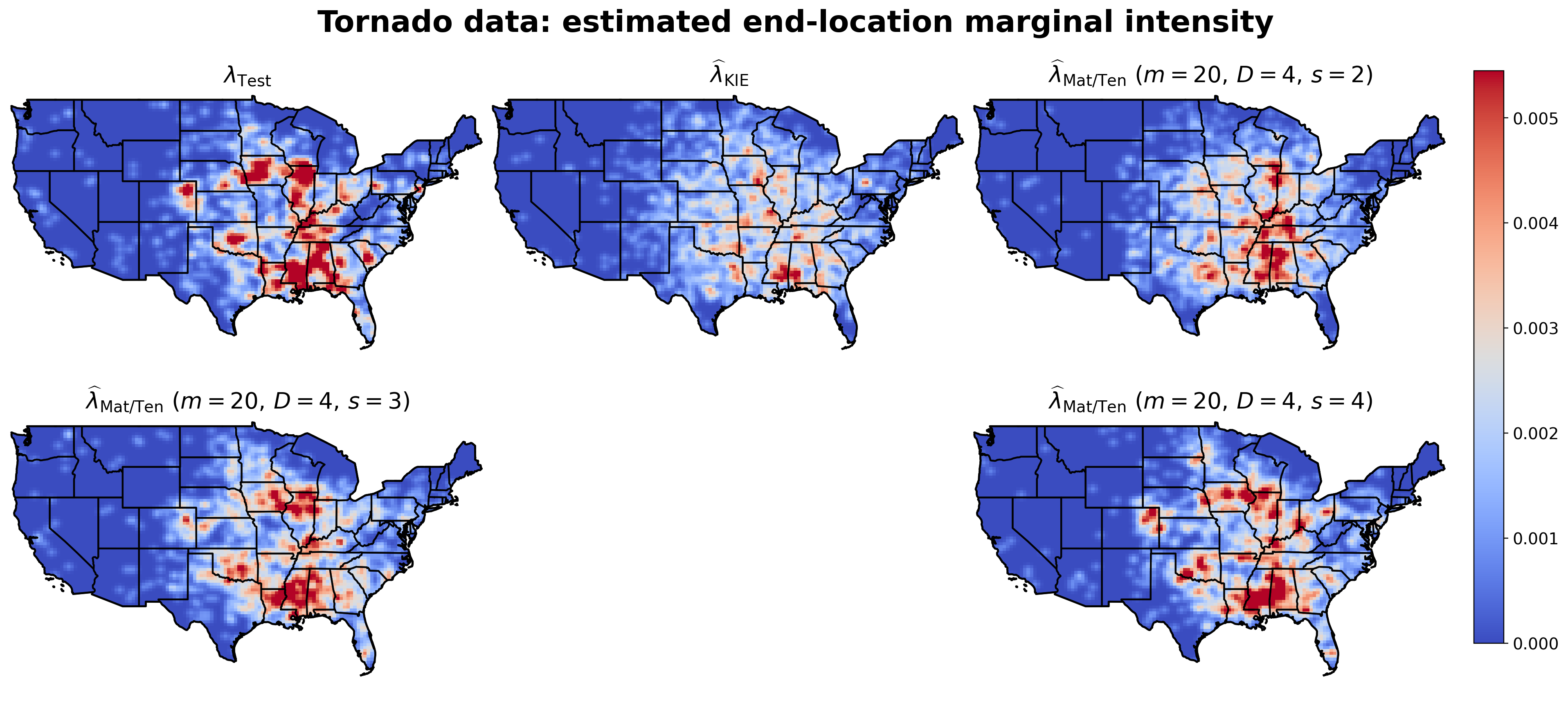}
\caption{Longitude--latitude marginal intensity maps for the tornado dataset. Top: start-location marginal intensity. Bottom: end-location marginal intensity. For each, the first two panels correspond to the empirical reference $\lambda_{\mathrm{Test}}$ and $\widehat{\lambda}_{\mathrm{KIE}}$, followed by $\widehat{\lambda}_{\mathrm{Mat/Ten}}$ for $s=2,3,4$.}
\label{fig:tornado_maps}
\end{figure}

The empirical reference maps reveal a clear concentration of tornado activity in the central and southeastern United States, consistent with well-documented regions such as ``Tornado Alley'' and ``Dixie Alley'' \citep{brooks2003climatological, gagan2010historical}. These regions include parts of Texas, Oklahoma, Kansas, Arkansas, Mississippi, and Alabama.

In addition, the maps exhibit a pronounced extension of activity into the Midwest and lower Ohio Valley, including portions of Illinois, Indiana, Iowa, Tennessee, and Kentucky. The empirical intensity forms a continuous band of elevated activity linking the Southeast and the Midwest. This pattern is consistent with prior climatological studies showing that tornado occurrence follows an irregular semicircular region of elevated risk extending from the Great Plains through the Southeast and into the Midwest \citep{cao2021geographic, brooks2003climatological, dixon2011tornado}. In particular, previous work has documented that elevated risk encompasses large portions of Iowa and Illinois, as well as shifted activity centers including Tennessee and Kentucky, highlighting the continuity of tornado-prone regions across the central United States.

The estimators $\widehat{\lambda}_{\mathrm{Mat/Ten}}$ recover these dominant spatial patterns across all values of $s$, with the sharpest localization for $s=4$, closely followed by $s=3$, while $s=2$ produces a more diffuse representation. In particular, the empirical map shows a prominent band of elevated activity that appears to originate around Oklahoma and extend through eastern Texas into Louisiana, Mississippi, and Alabama, where the intensity is strongest. This southward extension is clearly captured by the Matrix/Tensor estimators, especially for $s=4$, while appearing less sharply defined for $s=3$ and more attenuated for $s=2$. In contrast, $\widehat{\lambda}_{\mathrm{KIE}}$ largely fails to recover this pattern, substantially diffusing the signal and reducing both its spatial coherence and peak intensity.

Beyond these large-scale patterns, finer spatial features further differentiate the methods. In the central Plains and extending into the Midwest and lower Ohio Valley---spanning regions such as Iowa, Illinois, Kentucky, and Tennessee---the empirical reference map exhibits a connected band of elevated activity. This feature is closely matched by $\widehat{\lambda}_{\mathrm{Mat/Ten}}(s=4)$, with $s=3$ capturing most of the pattern and $s=2$ providing a more diffuse approximation. In contrast, $\widehat{\lambda}_{\mathrm{KIE}}$ smooths over this structure, weakening its continuity and attenuating localized peaks.

A similar contrast is observed when examining additional localized features across the domain. In the Ohio Valley, particularly across parts of Kentucky, Indiana, and Ohio, the empirical reference map exhibits a coherent region of elevated intensity. This feature is well recovered by $\widehat{\lambda}_{\mathrm{Mat/Ten}}(s=4)$ and, to a slightly lesser extent, by $s=3$, while $s=2$ yields a more diffuse representation. In contrast, $\widehat{\lambda}_{\mathrm{KIE}}$ substantially attenuates this structure, failing to preserve its spatial coherence.

Another distinct localized feature appears in the Florida peninsula, where the empirical map shows a clear region of elevated intensity. This feature is captured by the Matrix/Tensor estimators---most clearly for $s=4$, followed by $s=3$---but is largely absent under $\widehat{\lambda}_{\mathrm{KIE}}$, which smooths out the signal.

Finally, in the central Plains, a localized feature near the intersection of Colorado, Nebraska, and Kansas appears in the empirical map and is recovered by the Matrix/Tensor estimators, particularly for $s=4$, followed by $s=3$. This feature is not clearly represented by the kernel estimator, which diffuses the signal over a broader area.

We next support this visual comparison with a distributional metric computed on the full four-dimensional normalized trajectory space. Specifically, we report the sliced Wasserstein-2 ($\mathrm{SW}_2$) distance between samples generated from each fitted estimator and the held-out test data, as defined in \Cref{sec:real_data}. Smaller values indicate closer agreement between the distribution of samples generated from the fitted intensity and the empirical distribution of the held-out data.

\begin{table}[ht]
\centering
\scalebox{0.9}{
\begin{tabular}{lc}
\toprule
\textbf{Estimator} & $\mathbf{\mathrm{SW}_2(\widehat{\lambda},\,\text{test})}$ \\
\midrule
$\widehat{\lambda}_{\mathrm{Mat/Ten}}(s=4)$ & $\mathbf{0.0959}$ \\
$\widehat{\lambda}_{\mathrm{Mat/Ten}}(s=3)$ & $0.1064$ \\
$\widehat{\lambda}_{\mathrm{Mat/Ten}}(s=2)$ & $0.1431$ \\
$\widehat{\lambda}_{\mathrm{KIE}}$ & $0.1793$ \\
\bottomrule
\end{tabular}}
\caption{Sliced Wasserstein-2 ($\mathrm{SW}_2$) distance between samples generated from each fitted estimator and the held-out test set for the tornado dataset, averaged over $30$ random splits. Smaller values indicate closer agreement with the empirical distribution.}
\label{tab:tornado_sw2}
\end{table}

Table~\ref{tab:tornado_sw2} provides quantitative support for the visual comparison in Figure~\ref{fig:tornado_maps}. The Matrix/Tensor estimators achieve smaller sliced Wasserstein-2 distances than KIE, indicating closer distributional agreement with the held-out tornado data. Among the Matrix/Tensor configurations, $s=4$ yields the smallest distance, followed closely by $s=3$, while $s=2$ performs noticeably worse. This ordering is consistent with the spatial maps, where the higher-resolution tensor factorizations capture the localized trajectory structure more sharply than the coarser $s=2$ estimator.

Overall, these findings are consistent with those in \Cref{sec:real_data}: the proposed Matrix/Tensor estimators effectively capture both large-scale spatial patterns and finer localized structures, while kernel-based methods tend to oversmooth and obscure these features. The spatial maps therefore explain the quantitative results above: higher-resolution tensor decompositions, especially $s=4$ and $s=3$, better represent the localized trajectory structure in the data and consequently achieve smaller $\mathrm{SW}_2$ distances to the held-out sample.

\subsection{Regional earthquake data: California--Nevada}
\label{sec:real_data_eq_regional}

We further apply the proposed methods to a regional subset of the U.S.~Geological Survey Earthquake Catalog, available at \url{https://earthquake.usgs.gov/earthquakes/search/}. This dataset complements the national earthquake analysis in \Cref{sec:real_data}. Unlike the nationwide, daily-aggregated setting considered there, we focus on a California--Nevada region and analyze the data at the event level.

The filtered dataset contains 15,150 seismic events recorded from 2015-01-01 to 2025-12-30, with $D=4$ attributes: latitude, longitude, depth, and magnitude. The events are overwhelmingly earthquakes (15,108 records), with a small number of mining explosions, quarry blasts, explosions, and sonic booms. Since these non-earthquake records account for less than one percent of the data, we retain them in the analysis.

The geographic region is defined by latitude in $[32.53,\,41.99]$ and longitude in $[-124.41,\,-114.18]$, corresponding primarily to California, Nevada, and nearby areas. Depth ranges from $-2.30$ to $74.16$, and magnitude from $2.5$ to $7.1$. All variables are rescaled to $[0,1]^4$ prior to estimation.

Each event is represented as a point in the four-dimensional space
\[
(\text{latitude}, \text{longitude}, \text{depth}, \text{magnitude}).
\]
This representation corresponds to estimating a joint first-order intensity over spatial location and event characteristics, capturing both the geographic distribution of seismic activity and the variation in depth and magnitude across the region.

From the perspective of \Cref{sec:math_setup}, this dataset departs from the multi-realization setting considered in the main analysis. Rather than observing $n$ independent point processes $\{N^{(i)}\}_{i=1}^n$, we observe a single realization with a large number of events, so inference is driven by the joint distribution of spatial location, depth, and magnitude.

Accordingly, unlike the temporally aggregated analysis in \Cref{sec:real_data}, the event-level representation emphasizes fine-scale spatial and structural variation within a seismically active region, providing a complementary setting focused on recovering localized features of the intensity function.

Since no ground truth intensity is available, we evaluate the fitted estimators using two complementary forms of evidence. First, we examine longitude--latitude marginal intensity maps, which provide a direct visual comparison with the held-out empirical spatial distribution. Second, we support this qualitative comparison using the sliced Wasserstein-2 distance between samples generated from each fitted estimator and the held-out test data.

We begin with the spatial marginal comparison. Figure~\ref{fig:regional_eq_maps} displays longitude--latitude marginal intensity maps obtained by integrating out depth and magnitude. As in \Cref{sec:real_data}, the empirical reference map $\lambda_{\mathrm{Test}}$ is constructed by pooling the held-out events, binning them over a longitude--latitude grid, and normalizing by the number of test observations and the grid-cell area.

\begin{figure}[]
\centering
\includegraphics[width=1\textwidth]{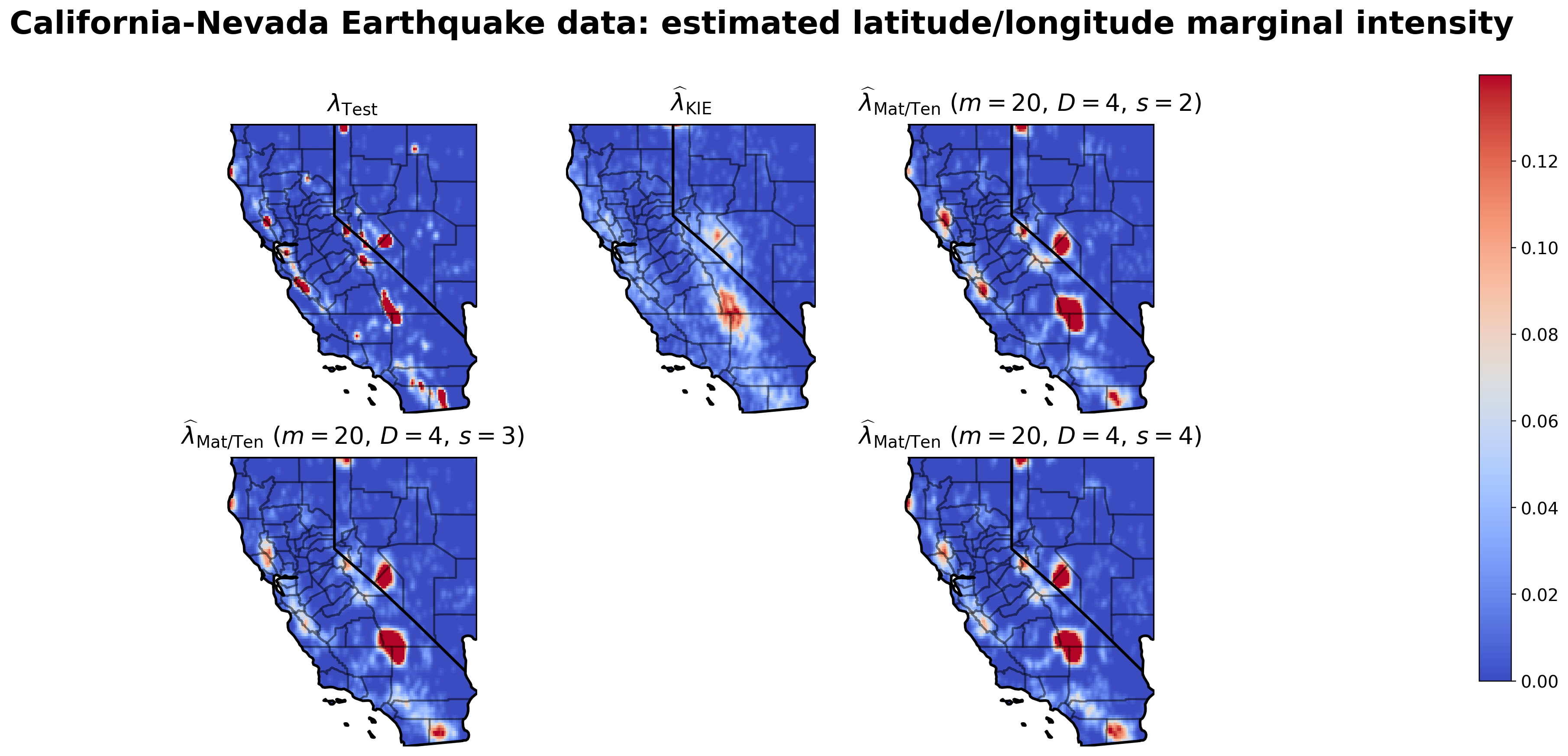}
\caption{Longitude--latitude marginal intensity maps for the California--Nevada earthquake dataset. The first two panels correspond to the empirical reference $\lambda_{\mathrm{Test}}$ and $\widehat{\lambda}_{\mathrm{KIE}}$, followed by $\widehat{\lambda}_{\mathrm{Mat/Ten}}$ for $s=2,3,4$.. The maps correspond to marginal intensities over longitude and latitude obtained by integrating out depth and magnitude.}
\label{fig:regional_eq_maps}
\end{figure}

The empirical reference map exhibits multiple highly localized regions of elevated seismic activity, rather than a single continuous structure. These appear as distinct, isolated clusters distributed across the California coast, the California--Nevada boundary, and southern California, together with additional localized activity in western Nevada.

In particular, six prominent hotspots are visible when ordered from north to south. A compact coastal cluster is observed in far northern California around Humboldt County. Slightly inland, a small isolated hotspot appears in western Nevada, centered in Washoe County. Moving south, a second coastal cluster is visible in the San Francisco Bay Area. Along the central California--Nevada boundary, a compact cluster is present in the Alpine--Mono region, extending into western Nevada. This is followed further south by a broader inland cluster centered in Inyo County, extending toward the Kern and San Bernardino boundary. Finally, in southern California, seismic activity is concentrated along the Imperial--San Diego--Riverside region, forming an elongated structure near the southern border.

These patterns highlight the fragmented and heterogeneous spatial organization of seismicity in this region, with clearly separated hotspots occurring across multiple spatial scales.

The Matrix/Tensor estimators $\widehat{\lambda}_{\mathrm{Mat/Ten}}$ recover these localized structures much more faithfully than the kernel estimator. Across the six regions described above, the estimator with $s=2$ provides the best overall balance between capturing spatial extent and preserving global structure, while $s=4$ yields sharper but sometimes overly localized representations. The estimator with $s=3$ tends to produce more diffuse and less stable structures.

Proceeding in the same north-to-south order, the northern coastal cluster in Humboldt County is best recovered by $\widehat{\lambda}_{\mathrm{Mat/Ten}}(s=4)$, which preserves a localized peak in the correct coastal region. The estimator with $s=2$ also captures this hotspot, although less sharply, while $s=3$ gives a more diffuse representation. In contrast, $\widehat{\lambda}_{\mathrm{KIE}}$ does not clearly recover this localized feature.

The nearby northern Nevada hotspot (Washoe County) is clearly captured by $s=2$, while it appears weaker and slightly less aligned for $s=4$ and $s=3$. This isolated feature is not clearly recovered by $\widehat{\lambda}_{\mathrm{KIE}}$.

The Bay Area cluster is consistently identified by the Matrix/Tensor estimators. The estimator with $s=2$ better preserves its spatial spread and multi-peak structure, while $s=4$ yields a more concentrated representation of the dominant peaks; $s=3$ again appears more diffuse. This cluster is not clearly recovered by $\widehat{\lambda}_{\mathrm{KIE}}$, which substantially attenuates the localized structure.

Along the central California--Nevada boundary, the compact Alpine--Mono cluster is recovered by all Matrix/Tensor estimators. The representations for $s=2$ and $s=4$ are very similar and closely match the location and separation observed in the empirical map, while $s=3$ yields a slightly more diffuse and less well-defined structure. In contrast, $\widehat{\lambda}_{\mathrm{KIE}}$ does not clearly resolve this cluster and instead merges it into a diffuse background.

Further south, the inland Inyo cluster is the most prominent feature and is recovered across all values of $s$, with $s=2$, $s=3$, and $s=4$ producing very similar representations in both location and shape. In contrast, $\widehat{\lambda}_{\mathrm{KIE}}$ also captures this region, but in a substantially more diffuse manner with reduced peak intensity.

In southern California, the Imperial--San Diego--Riverside region forms an extended chain of activity that is best preserved by $s=2$, which maintains the continuity of the structure. The estimators with $s=3$ and $s=4$ recover this region as well, but tend to smooth or fragment the chain into less clearly connected components. In contrast, $\widehat{\lambda}_{\mathrm{KIE}}$ does not clearly recover this southern structure and largely misses the associated localized activity.

Taken together, the kernel estimator $\widehat{\lambda}_{\mathrm{KIE}}$ substantially oversmooths the intensity surface. While it captures the dominant inland Inyo region, it does so in a diffuse manner and does not clearly recover most of the smaller, well-separated hotspots.

We next support this visual comparison with a distributional metric computed on the full four-dimensional normalized event space. Specifically, we report the sliced Wasserstein-2 ($\mathrm{SW}_2$) distance between samples generated from each fitted estimator and the held-out test data, as defined in \Cref{sec:real_data}. Smaller values indicate closer agreement between the distribution of samples generated from the fitted intensity and the empirical distribution of the held-out events.

\begin{table}[ht]
\centering
\scalebox{0.9}{
\begin{tabular}{lc}
\toprule
\textbf{Estimator} & $\mathbf{\mathrm{SW}_2(\widehat{\lambda},\,\text{test})}$ \\
\midrule
$\widehat{\lambda}_{\mathrm{Mat/Ten}}(s=2)$ & $\mathbf{0.0636}$ \\
$\widehat{\lambda}_{\mathrm{Mat/Ten}}(s=4)$ & $0.0672$ \\
$\widehat{\lambda}_{\mathrm{Mat/Ten}}(s=3)$ & $0.0733$ \\
$\widehat{\lambda}_{\mathrm{KIE}}$ & $0.0815$ \\
\bottomrule
\end{tabular}}
\caption{Sliced Wasserstein-2 ($\mathrm{SW}_2$) distance between samples generated from each fitted estimator and the held-out test set for the California--Nevada earthquake dataset, averaged over $30$ random splits. Smaller values indicate closer agreement with the empirical distribution.}
\label{tab:regional_eq_sw2}
\end{table}

Table~\ref{tab:regional_eq_sw2} provides quantitative support for the visual comparison in Figure~\ref{fig:regional_eq_maps}. The Matrix/Tensor estimators achieve smaller sliced Wasserstein-2 distances than KIE, indicating closer distributional agreement with the held-out California--Nevada earthquake data. Among the Matrix/Tensor configurations, $s=2$ yields the smallest distance, followed by $s=4$ and $s=3$. This ordering is consistent with the spatial maps, where $s=2$ provides the best balance between capturing fine-scale localized features and preserving the broader spatial organization of the region.

Overall, this example complements the findings in \Cref{sec:real_data} by highlighting the advantage of the proposed Matrix/Tensor estimators in settings with strong spatial localization. The $\mathrm{SW}_2$ results indicate closer agreement with the empirical distribution of the held-out events, while the spatial analysis explains this improvement by showing that moderate partitioning ($s=2$) provides the best balance between capturing fine-scale structure and preserving global spatial organization. In contrast, kernel-based methods tend to oversmooth and obscure localized features.

\subsection{Relation to ETAS Models in the Earthquake Application}\label{sec:ETAS_earthquake}
Epidemic-Type Aftershock Sequence (ETAS) models \citep{ogata1988statistical}, built on self-exciting (Hawkes) processes \citep{hawkes1974cluster}, are a leading framework for earthquake catalog analysis, declustering, and short-term forecasting. To clarify the relationship between our method and this widely used class of models, we contrast their inferential goals, compare their modeling structures, and identify a shared nonparametric ingredient.

\smallskip
\noindent\textbf{Inferential goal.}
ETAS targets two distinct objects: the \emph{history-dependent conditional intensity} together with a separately estimated \emph{background spatial intensity}. In the standard space--time formulation \citep{zhuang2002stochastic}, an earthquake catalog is represented by event tuples
\(
(t_i,x_i,y_i,m_i),\,i=1,2,\ldots,
\)
where $t_i$ is the occurrence time, $(x_i,y_i)$ is the epicentral location, and $m_i$ is the magnitude of the $i$th earthquake. Let
\(
\mathcal{H}_t
=
\sigma\!\left(\{(t_i,x_i,y_i,m_i):t_i<t\}\right)
\)
denote the $\sigma$-field generated by all events prior to time $t$. The ETAS conditional intensity is then written as
\[
\lambda_{\mathrm{ETAS}}(t,x,y \mid \mathcal{H}_t)
=
\mu(x,y)
+
\sum_{i:t_i<t}
\kappa(m_i)\, g(t-t_i)\, f(x-x_i,y-y_i \mid m_i),
\]
where $\mu(x,y)$ is the background spatial intensity, and $(\kappa,g,f)$ are parametric triggering kernels capturing magnitude-dependent productivity, temporal decay, and magnitude-dependent spatial spread of aftershocks. Our method, by contrast, targets a single object---the \emph{joint history-averaged first-order intensity}
\[
\lambda^*(u)
=
\mathbb{E}\!\left[N(\mathrm{d}u)\right]/\mathrm{d}u,
\qquad u\in\mathbb{X}\subset\mathbb{R}^D,
\]
defined on the full attribute space without conditioning on past events. In the earthquake application, the four physical attributes $(\mathrm{lat},\mathrm{lon},\mathrm{depth},\mathrm{mag})$ are coordinatewise rescaled to the unit cube, giving $u\in[0,1]^4$, and $\lambda^*(u)$ is estimated by basis projection and low-rank tensor regularization (Section~\ref{sec:method}).

\smallskip
\noindent\textbf{Modeling structure.}
The two approaches also differ substantially in their modeling structure. ETAS imposes a parametric self-excitation form on the triggering component while leaving the background nonparametric on the 2D spatial plane, with the two components fit under different criteria (history-conditional likelihood for the triggering kernels; weighted kernel smoothing for the background). Our method instead places a single nonparametric basis representation on the full 4D attribute space and imposes low-rank tensor structure on the resulting coefficient tensor, fit under a unified criterion evaluated via the
sliced Wasserstein-2 distance between samples generated from each fitted estimator and
the held-out test data. The 4D estimator captures structured interactions across location, depth, and magnitude that 2D ETAS background estimators do not directly model, and tensor low-rank regularization remains tractable in multivariate mark spaces where kernel-based estimators are constrained by the curse of dimensionality. Explicit modeling of self-excitation, aftershock cascades, and external drivers, however, remains the domain of ETAS-type models.

\smallskip
\noindent\textbf{A shared nonparametric ingredient.}
Although recovering the 2D spatial intensity is not the primary target of our method, our framework can be used to obtain it as a byproduct of the joint 4D estimator. ETAS-based methods such as \citet{zhuang2002stochastic} estimate the spatial background $\mu(x,y)$ nonparametrically on the 2D plane, typically via a kernel-based method with each event weighted by its inferred background probability. Our method, in contrast, induces a 2D spatial first-order intensity by marginalizing the 4D basis-plus-low-rank-tensor estimator. Specifically, following the marked-process discussion in \Cref{remark:beyond_poisson}, identify the spatial coordinates with
\(
x_0=(\mathrm{lat},\mathrm{lon})
\)
and the remaining marks with
\(
m=(\mathrm{depth},\mathrm{mag}).
\)
Then the joint first-order intensity can be written as $\lambda^*(x_0,m)$, and its spatial marginal is
\(
\lambda_0^*(x_0)
=
\int_{[0,1]^2}
\lambda^*(x_0,\mathrm{depth},\mathrm{mag})
\,\d(\mathrm{depth})\,\d(\mathrm{mag}).
\)
Similarly, our estimator induces
\(
\widehat{\lambda}_0(x_0)
=
\int_{[0,1]^2}
\widehat{\lambda}(x_0,\mathrm{depth},\mathrm{mag})
\,\d(\mathrm{depth})\,\d(\mathrm{mag}).
\)
Because $\widehat{\lambda}$ is represented in a tensor-product basis, this marginalization reduces to finite sums involving one-dimensional basis integrals over the depth and magnitude coordinates; as noted in \Cref{remark:beyond_poisson}, these integrals depend only on the chosen basis and can be precomputed. Thus, both approaches can produce a 2D spatial first-order intensity, but through different routes: ETAS-based methods estimate a background component on the projected spatial plane, whereas our method obtains a spatial marginal by integrating a structured 4D estimator. The latter can borrow strength from depth and magnitude when the joint intensity exhibits low-rank coupling across coordinates.

\color{Black}

\section{Examples of low-rank functions}\label{sec:example_low-rank}
Let $A: \mathbb{X} \to \mathbb R$ be an $D$-variable function in $\mathbb{L}_2(\mathbb{X})$, where $\mathbb{X} = \mathbb{X}_1 \times \cdots \times \mathbb{X}_s \subset \mathbb R^{D}$, with $\mathbb{X}_j \subset \mathbb{R}^{d_j}$ for all $j \in [s]$ and $2 \leq s \leq D$.  For each $j$, let $x_{-j} = (x_1, \dots, x_{j-1}, x_{j+1}, \dots, x_s) \in \mathbb{X}_{-j} = \mathbb{X}_1 \times \cdots \times \mathbb{X}_{j-1} \times \mathbb{X}_{j+1} \times \cdots \times \mathbb{X}_s \subset \mathbb{R}^{D-d_1}$.
\subsection{Example 1: Additive functions}
In nonparametric multiple regression \citep{friedman1981projection}, it is often assumed that the unknown function $A$ is additive, in the sense that for all $2 \leq s \leq D$
\[
A(x_1, \dots, x_s) = A_1(x_1) + \cdots + A_s(x_s), \;\; \text{for all} \;\; (x_1, \dots, x_s) \in \mathbb{X}.
\]
Rewrite the above equation in the form of the function SVD, for each $j \in [s]$,
\[
A(x_j, x_{-j}) = \{A_j(x_j) \cdot 1\} + \{1 \cdot A_{-j}(x_{-j})\},
\]
where $A_{-j}(x_{-j}) = \sum_{l \neq j}A_l(x_l)$.
This indicates that the Tucker rank of $A$ is $(2, \dots, 2)$.

\subsection{Example 2: Multiplicative functions}
It is also commonly assumed that the unknown function $A$ is multiplicative \citep{blei2017variational}, in the sense that for all $2 \leq s \leq D$
\[
A(x_1, \dots, x_s) = A_1(x_1) \cdots   A_s(x_s), \;\; \text{for all} \;\; (x_1, \dots, x_s) \in \mathbb{X}.
\]
Rewrite the above equation in the form of the function SVD, for each $j \in [s]$,
\[
A(x_j, x_{-j}) = A_j(x_j) \cdot A_{-j}(x_{-j}),
\]
where $A_{-j}(x_{-j}) = \prod_{l \neq j}A_l(x_l)$.
This indicates that the Tucker rank of $A$ is $(1, \dots, 1)$. Note that multiplicative functions are special cases of the mean-field models given in the next example.

\subsection{Example 3: Mean-field models}
Mean-field theory is widely used in computational physics, Bayesian statistics, and statistical mechanics.    One of the main challenges in solving statistical mechanics models is the existence of correlations in the system arising from interactions between particles. If we can approximate the model with a non-interacting counterpart, solving it becomes significantly simpler. Mean-field approximation treats these variables as independent and simplifies the complexity of handling their interactions.  We refer the readers to \cite{blei2017variational} for more details.

Specifically, an unknown density function $A : \mathbb{X} \to \mathbb{R}_+$ can be well approximated by a mixture of mean-field densities. Let $\{\tau_{\rho}\}_{\rho = 1}^{r}$ be a sequence of probabilities summing to $1$. In the mean-field mixture model, with probability $\tau_{\rho}$, data are sampled from a mean-field density 
\[
A_{\rho}(x_1, \dots, x_s) = A_{\rho,1}(x_1) \cdots A_{\rho,s}(x_s).
\]
Thus,
\begin{align*}
A(x_1, \dots, x_s) =& \sum_{\rho = 1}^{r}\tau_{\rho}A_{\rho,1}(x_1) \cdots A_{\rho,s}(x_s)\\
=& \sum_{\rho = 1}^{r}\tau_{\rho}A_{\rho,j}(x_j) \cdot A_{\rho,-j}(x_{-j}),
\end{align*}
where $A_{\rho,-j}(x_{-j}) = \prod_{l \neq j}A_{\rho,l}(x_l)$, for each $j \in [s]$. This indicates that the Tucker rank of $A$ is $(r, \dots, r)$.

\subsection{Example 4: Multivariate Taylor expansion}
Consider a function $A$ that is continuously differentiable up to order $\alpha$. By Taylor's theorem, for points $x = (x_1, \dots, x_s) \in \mathbb{X}$ and $t = (t_1, \dots, t_s) \in \mathbb{X}$, we have
\[
A(x) \approx T_t(x) = A(t) + \sum_{k=1}^\alpha \frac{1}{k!} \mathcal D^k A(t, x - t),
\]
where $\mathcal D^k A(l, m) = \sum_{i_1, \dots, i_k = 1}^s \partial_{i_1}\cdots\partial_{i_k}A(l) \cdot m_{i_1}\cdots m_{i_k}$, for $l,m \in \mathbb{X}$.  For simplicity, consider \( t = 0 \in \mathbb{X} \), and then the expansion becomes
\[
T_0(x) = A(0) + \sum_{i=1}^s \partial_i A(0) x_i + \frac{1}{2!} \sum_{i=1}^s
 \sum_{j=1}^s \partial_{i}\partial_{j} A(0) x_i x_j + \cdots + \frac{1}{\alpha!} \sum_{i_1, \dots, i_\alpha=1}^s \partial_{i_1}\cdots\partial_{i_\alpha} A(0) x_{i_1} \cdots x_{i_\alpha}.
\]
Rewrite the above equation in the form of the function SVD, we have that $A$ can be well-approximated by a finite-rank function with the Tucker rank $(\alpha + 1, \dots, \alpha+1)$.

\section{Sobolev space and RKHS basis}\label{sec:approximation_theory}
The approximation error between a function $A: \mathbb{X} \to \mathbb{R}$ and its projection onto a finite-dimensional tensor product subspace relies on both the smoothness of $A$ and the choice of orthonormal basis functions.

Let \(\mathbb{X}  \subset \mathbb{R}^D\) be any measurable set. For a multi-index \(\beta = (\beta_1, \dots, \beta_D) \in \mathbb{N}^D\) and a function \(f : \mathbb{X} \to \mathbb{R}\), the \(\beta\)-derivative of \(f\) is defined as
\[
\mathcal D^\beta f = \partial_1^{\beta_1} \cdots \partial_D^{\beta_D} f.
\]
The Sobolev space \(W_2^\alpha(\mathbb{X})\) is defined as
\[
W_2^\alpha(\mathbb{X}) = \{ f \in \mathbb{L}_2(\mathbb{X}) : \mathcal D^\beta f \in \mathbb{L}_2(\mathbb{X}) \text{ for all } |\beta|_1 \leq \alpha \},
\]
where \(|\beta|_1 = \beta_1 + \cdots + \beta_D\), and \(\alpha\) represents the total order of derivatives. The (squared) Sobolev norm of \(f \in W_2^\alpha(\mathbb{X})\) is
\[
\|f\|_{W_2^\alpha(\mathbb{X})}^2 = \sum_{0 \leq |\beta|_1 \leq \alpha} \|\mathcal D^\beta f\|_{\mathbb{L}_2(\mathbb{X})}^2.
\]

We briefly introduce the reproducing kernel Hilbert space (RKHS).
For \(x, y \in \Omega\), let \(\mathcal{K} : \Omega \times \Omega \to \mathbb{R}\) be a continuous and positive semidefinite kernel function such that
\begin{align}\label{eq:kernel_fct}
\mathcal{K}(x, y) = \sum_{k=1}^\infty \lambda_k^\mathcal{K} \psi_k^\mathcal{K}(x) \psi_k^\mathcal{K}(y),
\end{align}
where \(\{\lambda_k^\mathcal{K}\}_{k=1}^\infty \subset \mathbb{R}_+ \cup \{0\}\) are eigenvalues ordered non-increasingly, and \(\{\psi_k^\mathcal{K}\}_{k=1}^\infty\) is a collection of basis functions in \(\mathbb{L}_2(\Omega)\).

The reproducing kernel Hilbert space generated by the kernel \(\mathcal{K}\) is defined as
\begin{align}\label{eq:rkhs_kernel}
\mathcal{H}(\mathcal{K}) = \left\{ f \in \mathbb{L}_2(\Omega) : \|f\|^2_{\mathcal{H}(\mathcal{K})} = \sum_{k=1}^\infty (\lambda_k^\mathcal{K})^{-1} \langle f, \psi_k^\mathcal{K} \rangle^2 < \infty \right\},
\end{align}
where $\|\cdot\|_{\mathcal{H}(\mathcal{K})}$ is the RKHS norm induced by the inner product.
For all functions \(f, g \in \mathcal{H}(\mathcal{K})\), the inner product in \(\mathcal{H}(\mathcal{K})\) is given by
\[
\langle f, g \rangle_{\mathcal{H}(\mathcal{K})} = \sum_{k=1}^\infty (\lambda_k^\mathcal{K})^{-1} \langle f, \psi_k^\mathcal{K} \rangle \langle g, \psi_k^\mathcal{K} \rangle.
\]
Let \(\varphi_k^\mathcal{K} = (\lambda_k^\mathcal{K})^{-1/2} \psi_k^\mathcal{K}\), and then \(\{\varphi_k^\mathcal{K}\}_{k=1}^\infty\) are the orthonormal basis functions in \(\mathcal{H}(\mathcal{K})\), as we have
\[
\langle \varphi_{k_1}^\mathcal{K}, \varphi_{k_2}^\mathcal{K} \rangle_{\mathcal{H}(\mathcal{K})} = 
\begin{cases}
1, & \text{if } k_1 = k_2, \\
0, & \text{if } k_1 \neq k_2,
\end{cases}
\]
and the induced RKHS norm is
\[
\|f\|^2_{\mathcal{H}(\mathcal{K})} = \sum_{k=1}^\infty (\lambda_k^\mathcal{K})^{-1} \langle f, \psi_k^\mathcal{K} \rangle^2 = \sum_{k=1}^\infty \langle f, \varphi_k^\mathcal{K} \rangle^2.
\]

We refer the readers to the Section~B.1 in \cite{khoo2024nonparametric} for the approximation theory of multi-dimensional Sobolev spaces using the RKHS basis.

\section{Computational and storage costs of matrix- and tensor-based methods}\label{sec:computation_cost}
\subsection{Costs of \Cref{alg0}}
\label{secE1}
\noindent \textbf{Computational cost.}
    The computational costs of \Cref{alg0} can be decomposed into three parts.  (i) The cost of computing $\widehat b$ is due to matrix multiplications, which is of $O(nm^D\mathcal{N})$, where $\mathcal{N} = \sum_{i = 1}^n|N^{(i)}|/n$ is the averaged number of points over the $n$ observed point processes.  (ii) The soft-SVT in computing $T_{\gamma}(\widehat b)$ has a cost of $O(m^{d_1}\cdot m^{d_2}\cdot r_{\gamma}) = O(m^{D} \cdot r_{\gamma})$, where $r_{\gamma}$ is the smallest integer such that $ \sigma_{r_{\gamma}}(\widehat{b}) \leq \gamma$ in \Cref{alg0}.  (iii) Evaluating the resulting estimator $\widehat\lambda_{\mathrm{Matrix}}$ at $n_{\text{test}}$ points using the rank-$r_\gamma$ factorization costs $O(n_{\text{test}}(m^{d_1} + m^{d_2})r_{\gamma})$. Therefore, the total cost of \Cref{alg0} is
    \[
    O\left(nm^D\mathcal{N} + m^{D} r_{\gamma} + n_{\text{test}}(m^{d_1} + m^{d_2})r_{\gamma}\right).
    \]
    The first and last terms in the total cost involve basis evaluations at observed events and test points, respectively. 
    These computations are embarrassingly parallel and can be cached for reuse (if one wants to reuse the observed points as the test points).  
    If the rank $r_{\gamma}$, implied by the soft-thresholding parameter $\gamma$, is a bounded constant, using parallelism (with sufficient number of nodes) would improves wall-clock time to
    $$O(m^D) = O(n^{(d_{\max}+d_{\min})/(2\alpha+d_{\max})})$$ 
    plugging the choice $m \asymp n^{1/(2\alpha+d_{\max})}$ given in \eqref{eq:choice_matrix}.
    
    \noindent\textbf{Storage cost.} The storage cost can be decomposed into two parts: (i) forming the empirical coefficient matrix $\widehat{b}$ and (ii) storing the low-rank representation of the soft-SVT coefficient matrix $T_{\gamma}(\widehat{b})$. For (i), a streaming implementation that iterates over events $u\in N^{(i)}$ only needs to keep the current basis vectors
    $\phi^{(1)}(u_1)\in\mathbb{R}^{m^{d_1}}$ and $\phi^{(2)}(u_2)\in\mathbb{R}^{m^{d_2}}$ in memory, costing
    $O(m^{d_1}+m^{d_2})$ working space per event, but updating and storing the full matrix $\widehat b$ requires $O(m^D)$ memory. 
    For (ii), instead of storing $T_{\gamma}(\widehat{b})$, which costs $O(m^D)$, the low-rank representation allows us to store only $\widehat U$, $\widehat \Sigma$ and $\widehat V$, which requires 
    \[
O\left((m^{d_1}+m^{d_2})r_\gamma\right).
\]
Note that the peak memory is dominated by the storage of $\widehat{b}$. If $m^D$ is large and storing $\widehat{b}$ is infeasible, one can implement randomized SVD by computing products of $\widehat{b}$ with random vectors, which reduces the peak memory from $O(m^D)$ to $O\left((m^{d_1}+m^{d_2})r_\gamma\right)$.
    

\subsection{Costs of \Cref{alg1}}
\label{secE2}
\noindent \textbf{Computational cost.}    The computational costs of \Cref{alg1} can be decomposed into five parts.  The cost of computing empirical measures is due to matrix multiplications, which is of $O(nm^D\mathcal{N})$.  In the first for-loop, each truncated SVD on $\mathcal{M}_j(\widehat b)$ has a computational cost of $O(m^{d_j}\cdot m^{D-d_j}\cdot R_j) = O(m^{D}R_j)$, and the total cost is of $O(m^{D}\sum_{j \in [s]}R_j)$.  In the second for-loop, the computational cost of the sketched matrix $\mathcal{M}_j(\widehat b^{H_2}) \cdot \otimes_{k \neq j} \widehat{U}_{k}^{(0)}$ for each $j$ is of $O(m^{D}\sum_{k \neq j}R_k)$, and the cost of truncated SVD on the sketched matrix is of $O(m^{d_j}\cdot \prod_{k\neq j}R_k \cdot R_j) = O(m^{d_j}\prod_{j \in [s]}R_j)$. The total cost of the second for-loop is of $O((s-1)m^{D}\sum_{j \in [s]}R_j + (\sum_{j \in [s]}m^{d_j})\prod_{j \in [s]}R_j)$. Evaluating the final estimator $\widehat{\lambda}_{\text{tensor}}$ at $n_{\text{test}}$ points using the Tucker decomposition form costs $O(n_{\text{test}}(\sum_{j \in [s]}m^{d_j}R_j +\prod_{j \in [s]}R_j))$. Therefore, the total cost of \Cref{alg1} is
    \[
    O\left(nm^D\mathcal{N} +  sm^{D}\sum_{j \in [s]}R_j + \Big(\sum_{j \in [s]}m^{d_j}\Big)\prod_{j \in [s]}R_j + n_{\text{test}}\Big(\sum_{j \in [s]}m^{d_j}R_j +\prod_{j \in [s]}R_j\Big)\right).
    \]
    Similarly to the matrix-based method, the first and last terms in the total cost involve basis evaluations at observed events and test points, respectively. 
    These computations are embarrassingly parallel and can be cached for reuse.    If the Tucker ranks $R_j$ are all bounded constants,  using parallelism (with sufficient number of nodes) would improves wall-clock time to
    $$O(m^D) = O(n^{D/(2\alpha + d_{\max})}) = O(n^{(2\alpha + d_{\max} + d_{\min})/(2\alpha + d_{\max})}),$$
    using the choice $m \asymp n^{1/(2\alpha+d_{\max})}$ given in \eqref{eq:choice_m} and the second equality follows from Condition \eqref{eq:condition_dim}.
    
    \noindent\textbf{Storage.} For (i) forming the empirical coefficient tensor $\widehat b^{H_k}$, we need to store basis vectors $\phi^{(j)}(u_j)\in\mathbb{R}^{m^{d_j}}$ for $j\in[s]$, which costs $O(\sum_{j=1}^s m^{d_j})$ working space per event, and store each full coefficient tensor $\widehat b^{H_k}$ requiring $O(m^D)$ memory. For (ii) storing the final Tucker estimator, we can store its factor matrices and core tensor, requiring
\[
O\left(\sum_{j\in[s]} m^{d_j}R_j + \prod_{j\in[s]}R_j\right)
\]
memory, which is substantially smaller than $O(m^D)$ when $R_j\ll m^{d_j}$.
Similarly to the matrix-based method, if $m^D$ is large and storing $\widehat{b}^{H_k}$ is infeasible, one can implement randomized HOSVD by computing products of $\mathcal{M}_j(\widehat{b}^{H_k})$ with random test matrices, which avoid materializing the tensors and reduces the peak memory from $O(m^D)$ to $O(\sum_{j\in[s]} m^{d_j}R_j + \prod_{j\in[s]}R_j)$.

\subsection{Comparison with KIE}
\label{secE3}
The naive KIE has quadratic training cost $O(n^2\mathcal N^2)$ and evaluation cost $O(n\mathcal{N})$ per test point, for $n_{\text{test}}$ points.
In contrast, once the low-rank factors are obtained, $\widehat\lambda_{\mathrm{Matrix}}$ can be evaluated at cost
$O((m^{d_1}+m^{d_2})r_\gamma)$ per test point. Similarly, once the Tucker factors are computed, $\widehat\lambda_{\mathrm{Tensor}}$ can be evaluated at
$O(\sum_{j\in[s]}m^{d_j}R_j + \prod_{j\in[s]}R_j)$
per test point. Both can be dramatically smaller than $O(n\mathcal{N})$ due to our choices of $m$, coordinate partitions and rank parameters.
Therefore, when $n_{\text{test}}$ is large, the overall runtime is often dominated by evaluation, and the matrix- and tensor-based estimators yield particularly large gains over the KIE due to both their compact storage and  low-rank evaluation cost.
Moreover, the KIE typically requires storing all observed events (and potentially additional kernel-related quantities), whereas the matrix- and tensor-based estimators can be stored and reused in compressed form.

\subsection{Empirical runtime per Monte Carlo repetition}
\label{sec:runtime_empirical}

To complement the computational analysis, Table~\ref{tab:runtime} reports the empirical wall time per Monte Carlo replicate for the kernel intensity estimator (KIE) and the Matrix/Tensor estimator with $m = 4$, across all seven scenarios from \Cref{sec:simu_setup} and \Cref{sec:supp_scenarios}. Each timing includes the full pipeline (clustering, basis aggregation, soft-SVT or Tucker decomposition with the rank rule from \Cref{sec:coor_part_cluster}, and grid evaluation), and uses the simulation setup of \Cref{sec:simu_setup} ($n = 10^5$ processes per replicate, $100$ Monte Carlo repetitions). For the Matrix/Tensor estimator, the partition size $s$ is selected to minimize the averaged relative error at each $(\text{scenario}, D)$.

\begin{table}[ht]
\centering
\small
\setlength{\tabcolsep}{6pt}
\begin{tabular}{c|ccccc}
\toprule
 & $D=2$ & $D=3$ & $D=4$ & $D=5$ & $D=6$ \\
\midrule
S1 & $\textbf{0.05}\,/\,0.42$ & $\textbf{0.17}\,/\,0.43$ & $0.81\,/\,\textbf{0.75}$ & $4.82\,/\,\textbf{1.63}$ & $25.90\,/\,\textbf{3.96}$ \\
S2 & $\textbf{0.02}\,/\,0.06$ & $\textbf{0.08}\,/\,\textbf{0.08}$ & $0.36\,/\,\textbf{0.12}$ & $2.18\,/\,\textbf{0.30}$ & $13.29\,/\,\textbf{1.00}$ \\
S3 & $\textbf{0.04}\,/\,0.21$ & $\textbf{0.17}\,/\,0.21$ & $1.03\,/\,\textbf{0.36}$ & $5.80\,/\,\textbf{0.80}$ & $28.53\,/\,\textbf{2.13}$ \\
S4 & $\textbf{0.01}\,/\,0.03$ & $0.06\,/\,\textbf{0.03}$ & $0.26\,/\,\textbf{0.06}$ & $1.45\,/\,\textbf{0.17}$ & $\phantom{0}7.68\,/\,\textbf{0.56}$ \\
S5 & $\textbf{0.02}\,/\,0.03$ & $0.08\,/\,\textbf{0.04}$ & $0.37\,/\,\textbf{0.09}$ & $2.30\,/\,\textbf{0.28}$ & $12.95\,/\,\textbf{0.94}$ \\
S6 & $\textbf{0.02}\,/\,0.03$ & $0.08\,/\,\textbf{0.04}$ & $0.37\,/\,\textbf{0.09}$ & $2.25\,/\,\textbf{0.27}$ & $13.01\,/\,\textbf{0.95}$ \\
S7 & $\textbf{0.02}\,/\,0.03$ & $0.08\,/\,\textbf{0.04}$ & $0.37\,/\,\textbf{0.09}$ & $2.31\,/\,\textbf{0.28}$ & $13.11\,/\,\textbf{0.95}$ \\
\bottomrule
\end{tabular}
\caption{Mean wall time per Monte Carlo replicate (seconds), KIE\,/\,Matrix-Tensor at $m = 4$, across the seven scenarios. Timings include the full pipeline (clustering, basis aggregation, decomposition, and grid evaluation). The Matrix-Tensor estimator uses the partition size $s$ that minimizes the rep-mean relative error at each $(\text{scenario}, D)$.}
\label{tab:runtime}
\end{table}

Table~\ref{tab:runtime} shows that the Matrix/Tensor estimator is uniformly faster than KIE from $D = 4$ onward, with the gap widening sharply as $D$ grows: at $D = 6$ it is approximately $7$--$14\times$ faster than KIE across all scenarios, with the smallest speedup ($\approx 6.5\times$) occurring on the Neyman--Scott scenario S1, where the much larger per-replicate event count makes basis aggregation the dominant cost. At $D = 2$ and $D = 3$ the comparison reverses by a small absolute margin: at these low dimensions, the fixed setup cost of the Matrix/Tensor pipeline, i.e.~correlation-clustering, basis construction, and Tucker rank selection, is more visible, whereas KIE has no such setup and is therefore the more direct choice.

These findings are consistent with the complexity analysis in Sections~\ref{secE1}--\ref{secE3}. Because the local hat basis has compact support, each event updates only $\mathcal{O}(2^D)$ cells, so the basis-aggregation cost scales as $\mathcal{O}(n\mathcal{N}\, 2^D)$, with an additional Tucker-decomposition cost of $\mathcal{O}(s R\, m^D)$, where $\mathcal{N}$ is the average number of events per process and $R = \max_j R_j$ is the largest selected mode rank. This is much lower than the cost of evaluating KIE on a $6^D$-cell test grid against all observed points. Together with the accuracy gains shown in Figure~\ref{fig:simu_main}, the main message is that the Matrix/Tensor estimator provides both lower error and lower computing time in multi-dimensions.

\color{black}

\section{Proofs for Section~\ref{sec:approx_low-rank}}\label{sec:proofs}

\subsection{Auxiliary lemmas}
The following lemma is from \cite{shah2016stochastically}, and we provide a proof for completeness.
\begin{lemma}[Soft-SVT]\label{lemma1}
    Let $Y=X+Z$, where $Z \in \mathbb{R}^{p_1 \times p_2}$ is a zero-mean matrix. If $\gamma \geq 1.01\|Z\|_{\op}$, then
$$
\left\|T_{\gamma}(Y)-X\right\|_{\mathrm{F}}^2 \leq C \sum_{k=1}^{\min\{p_1, p_2\}} \min \left\{\gamma^2, \sigma_k^2\left(X\right)\right\},
$$
where $C > 0$ is an absolute constant.
\end{lemma}
\begin{proof}
Fix $\delta=0.01$. Let $q \leq \min\{p_1, p_2\}$ be the number of singular values of $X$ above $\delta(1+\delta)^{-1} \gamma$, and let $X_{(q)}$ be the truncated SVD of $X$.  We then have
\begin{align*}
\left\|T_{\gamma}(Y)-X\right\|_{\mathrm{F}}^2 & \leq 2\left\|T_{\gamma}(Y)-X_{(q)}\right\|_{\mathrm{F}}^2+2\left\|X_{(q)} - X\right\|_{\mathrm{F}}^2 \\
& \leq 2 \rank\left(T_{\gamma}(Y)-X_{(q)}\right)\left\|T_{\gamma}(Y)-X_{(q)}\right\|_{\op}^2 + 2 \sum_{k= q+1}^{\min\{p_1, p_2\}} \sigma_k^2\left(X\right) .
\end{align*}
We claim that $T_{\gamma}(Y)$ has rank at most $q$. Indeed, for each $k \geq q+1$, by \Cref{thm:WeylSV} we have
$$
\sigma_k(Y) \leq \sigma_k\left(X\right)+\|Z\|_{\text {op }} \leq \gamma,
$$
where we have used the facts that $\sigma_k(X) \leq \delta(1+\delta)^{-1} \gamma$ for each $k \geq q+1$, and $\gamma \geq(1+\delta)\|Z\|_{\op}$. As a consequence we have $\sigma_k\left(T_{\gamma}(Y)\right)=0$ for each $k \geq q+1$, and hence $\rank\left(T_{\gamma}(Y)-X_{(q)}\right) \leq 2q$. Moreover, we have
\begin{align*}
\left\|T_{\gamma}(Y)-X_{(q)}\right\|_{\op} & \leq\left\|T_{\gamma}(Y)-Y\right\|_{\op}+\left\|Y-X\right\|_{\op}+\left\|X-X_{(q)}\right\|_{\op} \\
& \leq \gamma +\|Z\|_{\op}+\frac{\delta}{1+\delta} \gamma \\
& \leq 2 \gamma .
\end{align*}
Putting together the pieces, we conclude that
$$
\left\|T_{\gamma}(Y)-X\right\|_{\mathrm{F}}^2 \leq 16 q \gamma^2+2 \sum_{k=q+1}^{\min\{p_1, p_2\}} \sigma_k^2\left(X\right) \leq C \sum_{k=1}^{\min\{p_1, p_2\}} \min \left\{\sigma_k^2\left(X\right), \gamma^2\right\},
$$
for some constant $C$. Here the second inequality follows since $\sigma_k\left(X\right) \leq \delta(1+\delta)^{-1} \gamma$ whenever $k \geq q+1$ and $\sigma_k\left(X\right)> \delta(1+\delta)^{-1} \gamma$ whenever $k \leq q$.
\end{proof}

\begin{lemma}[Fundamental bound for Poisson point process]\label{lemma:deviation_op_ortho}
    Let $\{N^{(i)}\}_{i = 1}^n$ be a set of i.i.d.~inhomogeneous Poisson point processes, with intensity function $\lambda^*$. Let $\widehat b \in \mathbb{R}^{m^{d_1} \times \cdots \times m^{d_s}}$ be the empirical coefficient tensor defined in \eqref{eq:empirical_coefficient_tensor}, and $b^*$ be the corresponding population coefficient tensor.  For any deterministic $V_j \in \mathbb O_{m^{D-d_j}, r_{V_j}}$ and $W_j \in \mathbb O_{m^{d_j}, r_{W_j}}$with $r_{V_j} \leq m^{D-d_j}$ and $r_{W_j} \leq m^{d_j}$, we have that with probability at least $1 - m^{-5}$, uniformly for all $j \in [s]$,
    \begin{align*}
        \left\|W_j^{\top} \cdot \mathcal{M}_j(\widehat b - b^*) \cdot V_j\right\|_{\op}
        \leq C\left\{\sqrt{\frac{\|\lambda^*\|_{\infty}(r_{V_j}+r_{W_j})\log(m)}{n}} + \frac{m^{D/2}\log(m)}{n}\right\},
    \end{align*}
    where $C > 0$ is an absolute constant.
\end{lemma}
\begin{proof}
    We obtain the upper bound using \Cref{coro:MatrixBernstein_PPP}, and we only focus on the matricization of $\widehat{b}$ at mode $j = 1$. For notational simplicity, we let $W = W_1$ and $V = V_1$. 
    Rewrite $$W^{\top} \cdot \mathcal{M}_1(\widehat b) \cdot V = \frac{1}{n}\sum_{i = 1}^n\sum_{X \in N^{(i)}}F(X) \in \mathbb{R}^{r_W \times r_V},$$
    where $X = (X_1^{\top}, \dots, X_s^{\top})^{\top} \in \mathbb{R}^D$ with $X_j \in \mathbb{R}^{d_j}$, and $x \mapsto F(x)$ is a $\mathbb{R}^{r_W \times r_V}$-valued function with the $(j,l)$ entry \begin{align*}
    F_{(j; l)}(x) & = \sum_{\mu_1 = 1}^{m^{d_1}}W_{(\mu_1;j)}\phi_{1,\mu_1}(x_1)\sum_{\mu_2 = 1}^{m^{d_2}}\cdots\sum_{\mu_s = 1}^{m^{d_s}}V_{(\mu_2, \dots, \mu_s; l)}\phi_{2,\mu_2}(x_2)\cdots\phi_{s,\mu_s}(x_s)\\
    &= \psi_{j}(x_1)\sum_{\mu_2 = 1}^{m^{d_2}}\cdots\sum_{\mu_s = 1}^{m^{d_s}}V_{(\mu_2, \dots, \mu_s; l)}\phi_{2,\mu_2}(x_2)\cdots\phi_{s,\mu_s}(x_s),
    \end{align*}
    where $W_{(\mu_1;j)}$ is the $(\mu_1,j)$ entry of $W$. Each combination of $(\mu_2, \dots, \mu_s)$ corresponds to a row index of $V$, and the corresponding row is denoted by $V_{(\mu_2, \dots, \mu_s; \cdot)}$.  For $j \in [r_W]$, we let  $\psi_{j}(\cdot) = \sum_{\mu_1 = 1}^{m^{d_1}}W_{(\mu_1;j)}\phi_{\mu_1}(\cdot)$. Note that $\{\psi_{j}\}_{j = 1}^{r_W}$ is a set of orthonormal basis functions, since $\{\phi_{1,\mu_1}\}_{\mu_1 = 1}^{m^{d_1}}$ is a set of orthonormal basis functions and $\{W_{(\cdot; j)}\}_{j=1}^{r_W}$ is a set of orthonormal vectors.
    
    We can write
    $$W^{\top} \cdot \mathcal{M}_1(b^*) \cdot V = \int F(x)\lambda^*(x)\d x.$$
    We verify the conditions of \Cref{coro:MatrixBernstein_PPP}.  It follows that for all $x$,
    \begin{align*}
        \|F(x)\|_{\op} &\leq \|F(x)\|_{\mathrm{F}} \leq \sqrt{\sum_{\mu_1=1}^{m^{d_1}}\cdots\sum_{\mu_s = 1}^{m^{d_s}}\left\{\phi_{1,\mu_1}(x_1)\cdots\phi_{s,\mu_s}(x_s)\right\}^2}\\
        &\leq m^{D/2}\prod_{j = 1}^s\left\|\phi_{j,\mu_j}(x_j)\right\|_{\infty}\\
        &\leq C_{\phi}^sm^{D/2},
    \end{align*}
    where the second inequality follows from $\|W\|_{\op} \leq 1$, $\|V\|_{\op} \leq 1$ and $D = \sum_{j = 1}^sd_j$ by definition, and the last inequality holds because the basis function satisfying $\|\phi_{j,\mu_j}\|_{\infty} \leq C_{\phi} < \infty$. Recall the matrix variance statistic $\nu$ in \Cref{coro:MatrixBernstein_PPP}, defined as
    $$ \nu = n \max\left\{\left\|\int F(x)(F(x))^{\top} \lambda^*(x) \d x\right\|_{\op}, \left\|\int (F(x))^{\top}F(x) \lambda^*(x) \d x\right\|_{\op} \right\}.$$
    We only derive the bound on $\|\int F(x)(F(x))^{\top} \lambda^*(x) \d x\|_{\op}$, since the bound on the other term can be obtained similarly.
    Note that the $(p,q)$ entry of $[F(x)(F(x))^{\top}]$ is
    \begin{align*}
        &\left[F(x)(F(x))^{\top}\right]_{(p; q)} = \sum_{l = 1}^{r_V} F_{(p;l)}(x)F_{(q;l)}(x)\\
        =& \psi_{p}(x_1)\psi_{q}(x_1)\sum_{l = 1}^{r_V}\left(\sum_{\mu_2 = 1}^{m^{d_2}}\cdots\sum_{\mu_s = 1}^{m^{d_s}} V_{(\mu_2, \dots, \mu_d; l)}\phi_{2,\mu_2}(x_2)\cdots\phi_{s,\mu_s}(x_s)\right)^2.
    \end{align*}
    Furthermore,
    \begin{align*}
        &\left\|\int F(x)(F(x))^{\top} \lambda^*(x) \d x\right\|_{\op} = \sup_{\|v\|_2 = 1} v^{\top}\left[\int F(x)(F(x))^{\top} \lambda^*(x) \d x\right] v\\
        =& \sup_{\|v\|_2 = 1} \int \left(\sum_{p = 1}^{r_W}\sum_{q = 1}^{r_W} v_p\left[F(x)(F(x))^{\top}\right]_{(p; q)}v_q\right)\lambda^*(x)\d x \\
        =& \sup_{\|v\|_2 = 1} \idotsint\left(\sum_{p = 1}^{r_W}\sum_{q = 1}^{r_W} v_p\psi_{p}(x_1)\psi_{q}(x_1)v_q\right)\\
        &\quad\quad\quad\quad\quad\quad\quad\left\{\sum_{l = 1}^{r_V}\left(\sum_{\mu_2 = 1}^{m^{d_2}}\cdots\sum_{\mu_s = 1}^{m^{d_s}} V_{(\mu_1, \dots,\mu_s; l)}\phi_{2,\mu_2}(x_2)\cdots\phi_{s,\mu_s}(x_s)\right)^2\right\}\lambda^*(x_1,\cdots,x_s)\d x_1 \cdots \d x_s\\
        =& \sup_{\|v\|_2 = 1} \idotsint\left(\sum_{k = 1}^{r_W} v_k\psi_{k}(x_1)\right)^2\\
        &\quad\quad\quad\quad\quad\quad\quad\left\{\sum_{l = 1}^{r_V}\left(\sum_{\mu_2 = 1}^{m^{d_2}}\cdots\sum_{\mu_s = 1}^{m^{d_s}} V_{(\mu_1, \dots,\mu_s; l)}\phi_{2,\mu_2}(x_2)\cdots\phi_{s,\mu_s}(x_s)\right)^2\right\}\lambda^*(x_1,\cdots,x_s)\d x_1 \cdots \d x_s\\
        \leq& \|\lambda^*\|_{\infty}\sup_{\|v\|_2 = 1} \int\left(\sum_{k = 1}^{r_W} v_k\psi_{k}(x_1)\right)^2\d x_1\\
        &\quad\quad\quad\quad\quad\quad\quad\left\{\sum_{l = 1}^{r_V}\idotsint\left(\sum_{\mu_2 = 1}^{m^{d_2}}\cdots\sum_{\mu_s = 1}^{m^{d_s}}V_{(\mu_2, \dots,\mu_s; l)}\phi_{2,\mu_2}(x_2)\cdots\phi_{s,\mu_s}(x_s)\right)^2\d x_2\cdots \d x_s\right\}\\
        =& \|\lambda^*\|_{\infty}\sup_{\|v\|_2 = 1} \int\sum_{k = 1}^{r_W} v_k^2\psi^2_{k}(x_1) \d x_1\\
        &\quad\quad\quad\quad\quad\quad \left\{\sum_{l = 1}^{r_V}\idotsint\sum_{\mu_2 = 1}^{m^{d_2}}\cdots\sum_{\mu_s = 1}^{m^{d_s}}\left\{V_{(\mu_2,\dots,\mu_s; l)}\phi_{2,\mu_2}(x_2)\cdots\phi_{s,\mu_s}(x_s)\right\}^2\d x_2\cdots \d x_s\right\}\\
        =& \|\lambda^*\|_{\infty}r_V,
    \end{align*}
    where the last two lines follows from the fact that $\{\psi_k\}_{k = 1}^{r_W}$, $\{\phi_{j,\mu_j}\}_{\mu_j = 1}^{m^{d_j}}$ are collections of othonormal functions, and $V \in \mathbb O_{m^{D-d_1}, r_V}$. Similarly, we can show that $$\left\|\int(F(x))^{\top}F(x) \lambda^*(x) \d x\right\|_{\op} \leq \|\lambda^*\|_{\infty}r_W.$$
    Therefore, we have $\nu \leq n\|\lambda^*\|_{\infty}(r_V + r_W)$ and $L = C_{\phi}^sm^{D/2}$. By \Cref{coro:MatrixBernstein_PPP}, we have that with probability at least $1 - m^{-5}$
    \begin{align*}
        \left\|\frac{1}{n}\sum_{i = 1}^n\sum_{X \in N^{(i)}}F(X) - \int F(x)\lambda(x)\d x\right\|_{\op} \leq C\left\{\sqrt{\frac{\|\lambda^*\|_{\infty}(r_V+r_W)\log(m)}{n}} + \frac{C_{\phi}^sm^{D/2}\log(m)}{n}\right\},
    \end{align*}
    where $C > 0$ is an absolute constant.
    The same argument leads to the similar bounds for $j = 2, \dots, s$, which concludes the proof.
\end{proof}

\subsection{Proof for \Cref{sec:theory_matrix}}
\begin{proof}[Proof of \Cref{theorem_matrix}]
Define the finite-dimensional subspaces $\mathcal{U}_1 = \mathrm{Span}\{\phi_{1,\mu_1}: \mu_1 \in [m^{d_1}]\}$ and $\mathcal{U}_2 = \mathrm{Span}\{\phi_{2,\mu_2}: \mu_2 \in [m^{d_2}]\}$, as well as the corresponding projection operators $\mathcal{P}_{\mathcal{U}_1}$ and $\mathcal{P}_{\mathcal{U}_2}$ (see~\Cref{sec:notation_app} for definitions).
Observe that
\begin{align*}
        \|\lambda^*-\widehat{\lambda}_{\mathrm{Matrix}}\|_{\mathbb{L}_2(\mathbb{X})}^2 &\leq 2\|\lambda^*- \lambda^*\times_1 \mathcal{P}_{\mathcal{U}_1} \times_2 \mathcal{P}_{\mathcal{U}_2}\|_{\mathbb{L}_2(\mathbb{X})}^2 + 2\|\lambda^*\times_1 \mathcal{P}_{\mathcal{U}_1} \times_2 \mathcal{P}_{\mathcal{U}_2} - \widehat{\lambda}_{\mathrm{Matrix}}\|_{\mathbb{L}_2(\mathbb{X})}^2\\
        &=: 2\mathrm{Term}_1 + 2\mathrm{Term}_2.
\end{align*}
For $\mathrm{Term}_1$, under Assumptions~\ref{ass:approx_err1} and \ref{ass:approx_err2}, \eqref{eq:approx_error} (see also Lemma~2 in \citealt{khoo2024nonparametric}) yields 
\begin{align}
    \label{Eq3-T2-UP}
        \mathrm{Term}_1 = O({\|\lambda^*\|_{W_2^{\alpha}(\mathbb{X})}^2  }m^{-2 \alpha}).
    \end{align}
    For $\mathrm{Term}_2$, we have 
    \begin{align}\label{eq:proof_mat}
        \mathrm{Term}_2 =& \|\lambda^*\times_1 \mathcal{P}_{\mathcal{U}_1} \times_2 \mathcal{P}_{\mathcal{U}_2} - \widehat{\lambda}_{\mathrm{Matrix}}\|_{\mathbb{L}_2}^2\nonumber\\
        =& \sum_{\mu_1 = 1}^{m^{d_1}}\sum_{\mu_2 = 1}^{m^{d_2}} \left\{\left(\lambda^*\times_1 \mathcal{P}_{\mathcal{U}_1} \times_2 \mathcal{P}_{\mathcal{U}_2} - \widehat{\lambda}_{\mathrm{Matrix}}\right)[\phi_{1,\mu_1}, \phi_{2,\mu_2}]\right\}^2\nonumber\\
    =& \sum_{\mu_1 = 1}^{m^{d_1}}\sum_{\mu_2 = 1}^{m^{d_2}} \left\{\iint\left[\sum_{v_1 = 1}^{m^{d_1}}\sum_{v_2 = 1}^{m^{d_2}} \{ b_{v_1,v_2}^* - T_{\gamma}(\widehat b)_{v_1,v_s}\}  \phi_{1,v_1}(x_1)\phi_{2,v_2}(x_2)\phi_{1,\mu_1}(x_1)\phi_{2,\mu_2}(x_2)\right]\d x_1 \d x_2\right\}^2\nonumber\\
    =& \sum_{\mu_1 = 1}^{m^{d_1}}\sum_{\mu_2 = 1}^{m^{d_2}} \{ b_{\mu_1,\mu_2}^* - T_{\gamma}(\widehat b)_{\mu_1,\mu_s}\}^2\nonumber\\
    =& \|b^* - T_{\gamma}(\widehat b)\|_{\mathrm{F}}^2,
    \end{align}
    where the fourth equality follows from the orthonormality of $\{\phi_{j,\mu_j}\}_{\mu_j = 1}^{\infty}$.
    Write 
$$\widehat{b} = b^*+Z,$$
    where $Z$ has is mean zero by \Cref{lemma:campbell}. Recall that we choose
    \begin{align*}
    m = ({\|\lambda^*\|_{W_2^{\alpha}(\mathbb{X})}}n)^{1/(2\alpha + d_{\max})}.
    \end{align*}
    We have, under \Cref{ass:approx_err1},  $$\sqrt{\frac{\|\lambda^*\|_{\infty}m^{d_{\max}}\log(m)}{n}} \gtrsim \sqrt{\frac{m^{d_{\max}}\log(m)}{n}} \gtrsim \frac{m^{D/2}\log(m)}{n},$$
    for some constant $C > 0$.
    Consequently, by \Cref{lemma:deviation_op_ortho} with $s = 2$, $W = I_{m^{d_1}}$ and $V = I_{m^{d_2}}$, we have 
\begin{align*}
    1.01\|Z\|_{\op} =& 1.01\|\widehat{b} - b^*\|_{\op}\\
    \lesssim& \sqrt{\frac{\|\lambda^*\|_{\infty}m^{d_{\max}}\log(m)}{n}}\\
    =& \sqrt{\frac{\|\lambda^*\|_{\infty}\|\lambda^*\|_{W_2^\alpha(\mathbb{X})}^{d_{\max}/(2\alpha + d_{\max})}\log(m)}{n^{2\alpha/(2\alpha + d_{\max})}}}\\
    =& \sqrt{\frac{\|\lambda^*\|_{\infty}^{2\alpha/(2\alpha + d_{\max})}\|\lambda^*\|_{W_2^\alpha(\mathbb{X})}^{2d_{\max}/(2\alpha + d_{\max})}(\|\lambda^*\|_{\infty}/\|\lambda^*\|_{W_2^\alpha(\mathbb{X})})^{d_{\max}/(2\alpha + d_{\max})}\log(m)}{n^{2\alpha/(2\alpha + d_{\max})}}}\\
    \le& C_{\gamma}\sqrt{ \frac{\|\lambda^*\|_{\infty}^{2\alpha/(2\alpha+d_{\max})}\|\lambda^*\|_{W_2^\alpha(\mathbb{X})}^{2d_{\max}/(2\alpha+d_{\max})}\log(n)}{n^{2\alpha/(2\alpha+d_{\max})}}} = \gamma,
\end{align*}
where the last inequality following from \Cref{ass:approx_err1}.
    Therefore, by Lemma \ref{lemma1}, we conclude that
\begin{align*}
   \mathrm{Term}_2=\|b^* - T_{\gamma}(\widehat b)\|_{\mathrm{F}}^2 \le c \sum_{j=1}^\infty \min \left\{\gamma^2, \sigma_j^2\left(b^*\right)\right\}
\end{align*}
    for some constant $c>0.$ For each rank $R>0$, we have
    \begin{align*}
        \mathrm{Term}_2 \leq cR\gamma^2+c\sum_{j=R+1}^\infty\sigma_j^2\left(b^*\right).
    \end{align*}
    The last inequality together with \eqref{Eq3-T2-UP} give
    \begin{align*}
        \|\lambda^*-\widehat{\lambda}_{\mathrm{Matrix}}\|_{\mathbb{L}_2}^2\le 2\mathrm{Term}_1+2\mathrm{Term}_{2} \le O({\|\lambda^*\|_{W_2^{\alpha}(\mathbb{X})}^2}m^{-2 \alpha})+2cR\gamma^2+2c\sum_{j=R+1}^{\infty}\sigma_j^2\left(b^*\right).
    \end{align*}
    It remains to bound the term      $$\sum_{j=R+1}^{\infty}\sigma_j^2\left(b^*\right) = \sum_{j=R+1}^{\infty}\sigma_j^2\left(\lambda^* \times_1 \mathcal{P}_{\mathcal{U}_1} \times_2 \mathcal{P}_{\mathcal{U}_2}\right).$$
    Let $\llbracket\lambda^* \times_1 \mathcal{P}_{\mathcal{U}_1} \times_2 \mathcal{P}_{\mathcal{U}_2}\rrbracket_{(R)}$ and $\llbracket \lambda^*\rrbracket_{(R)}$ denote the best rank-$R$ approximation of $\lambda^* \times_1 \mathcal{P}_{\mathcal{U}_1} \times_2 \mathcal{P}_{\mathcal{U}_2}$ and $\lambda^*$, respectively (see \Cref{sec:approximately_low-rank_fct} for details). We have
    \begin{align*}
    &\sum_{j=R+1}^{\infty}\sigma_j^2\left(\lambda^* \times_1 \mathcal{P}_{\mathcal{U}_1} \times_2 \mathcal{P}_{\mathcal{U}_2}\right) = \sum_{j=1}^{\infty}\left|\sigma_j\left(\lambda^* \times_1 \mathcal{P}_{\mathcal{U}_1} \times_2 \mathcal{P}_{\mathcal{U}_2}\right) - \sigma_j\left(\llbracket\lambda^* \times_1 \mathcal{P}_{\mathcal{U}_1} \times_2 \mathcal{P}_{\mathcal{U}_2}\rrbracket_{(R)}\right)\right|^2\\
    \leq& \|\lambda^* \times_1 \mathcal{P}_{\mathcal{U}_1} \times_2 \mathcal{P}_{\mathcal{U}_2} - \llbracket\lambda^* \times_1 \mathcal{P}_{\mathcal{U}_1} \times_2 \mathcal{P}_{\mathcal{U}_2}\rrbracket_{(R)}\|_{\mathbb{L}_2}^2\\
    \leq& \|\lambda^* \times_1 \mathcal{P}_{\mathcal{U}_1} \times_2 \mathcal{P}_{\mathcal{U}_2} - \llbracket\lambda^*\rrbracket_{(R)} \times_1 \mathcal{P}_{\mathcal{U}_1} \times_2 \mathcal{P}_{\mathcal{U}_2}\|_{\mathbb{L}_2}^2\\
    \leq& \|\lambda^* - \llbracket\lambda^*\rrbracket_{(R)} \|_{\mathbb{L}_2}^2\\
    =& \xi^2_{(R)},
    \end{align*}
    where the first inequality follows from \Cref{lem:Mirsky_Hilbert}. The second inequality follows since $\llbracket\lambda^*\rrbracket_{(R)} \times_1 \mathcal{P}_{\mathcal{U}_1} \times_2 \mathcal{P}_{\mathcal{U}_2}$ is of rank $R$.  The last equality follows from the definition \eqref{eq:tucker_low-rank_approx}.
    Note that by \Cref{ass:approx_err1} and for all $R \geq 1$
    \begin{align*}
        {\|\lambda^*\|_{W_2^{\alpha}(\mathbb{X})}^2}m^{-2 \alpha} =& {\|\lambda^*\|_{W_2^{\alpha}(\mathbb{X})}^{(2 \alpha + 2d_{\max})/(2\alpha + d_{\max})}} n^{-2 \alpha/(2\alpha + d_{\max})}\\
        =& {\|\lambda^*\|_\infty^{2 \alpha/(2\alpha + d_{\max})}}({\|\lambda^*\|_{W_2^{\alpha}(\mathbb{X})}/\|\lambda^*\|_{\infty})^{2 \alpha/(2\alpha + d_{\max})}}{\|\lambda^*\|_{W_2^{\alpha}(\mathbb{X})}^{2d_{\max}/(2\alpha + d_{\max})}} n^{-2 \alpha/(2\alpha + d_{\max})}\\
        \lesssim& R\left(\frac{\|\lambda^*\|_{\infty}^{2\alpha/(2\alpha+d_{\max})}\|\lambda^*\|_{W_2^\alpha(\mathbb{X})}^{2d_{\max}/(2\alpha+d_{\max})}}{n^{2\alpha/(2\alpha+d_{\max})}}\right)\log(n)\\
        =& R\gamma^2.
    \end{align*}

    Therefore, by plugging in the choice of $m$ and $\gamma$, we have that for each rank $R>0$,
    \begin{align*}
        \|\lambda^*-\widehat{\lambda}_{\mathrm{Matrix}}\|_{\mathbb{L}_2}^2 \le& O_p\left({\|\lambda^*\|_{W_2^{\alpha}(\mathbb{X})}^2}m^{-2 \alpha} +R\gamma^2+ \xi_{(R)}^2\right)\\
        =& O_{p}\left( \frac{\|\lambda^*\|_{\infty}^{2\alpha/(2\alpha+d_{\max})}\|\lambda^*\|_{W_2^\alpha(\mathbb{X})}^{2d_{\max}/(2\alpha+d_{\max})}}{n^{2\alpha/(2\alpha+d_{\max})}}\log(n) +\xi_{(R)}^2 \right).
    \end{align*}
    This concludes the proof.
\end{proof}

\subsection{Proof for \Cref{sec:poisson}}
\subsubsection{Auxiliary results for \Cref{alg1}}\label{sec:alg1_results}

We present key propositions that establish error bounds on the initial and refined singular vector estimators as well as the final low-rank tensor estimation obtained in \Cref{alg1}.

\begin{proposition}[Error bound on the initial singular vectors]\label{lemma:hosvd_approx}
    Let $\widehat{b}^{H_1}$ denote the empirical coefficient tensor computed based on the subset $H_1$ of the observation (see~\Cref{alg1}). Let $\widehat{U}_{j}^{(0)} = \text{SVD}_{(R_j)} (\mathcal{M}_j(\widehat b^{H_1})) \in \mathbb{O}_{m^{d_j}, R_j}$, for $j \in [s]$, be the truncated SVD obtained in the initialization step of \Cref{alg1}, and $U_{j} = \text{SVD}_{(R_j)} (\mathcal{M}_j(b^*)) \in \mathbb{O}_{m^{d_j}, R_j}$.
    Suppose it holds that with probability at least $1 - m^{-5}$, uniformly for all $j \in [s]$
    \begin{align}\label{eq:fundamental_general}
        \left\|W_j^{\top} \cdot \mathcal{M}_j(\widehat b^{H_1} - b^*) \cdot M_j\right\|_{\op}
        \leq a_1\sqrt{\frac{\|\lambda^*\|_{\infty}(r_{W_j}+r_{M_j})\log(m)}{n}} + a_2\frac{m^{D/2}\log(m)}{n},
    \end{align}
    where $W_j \in \mathbb O_{m^{d_j}, r_{W_j}}$ and $M_j \in \mathbb O_{m^{D-d_j}, r_{M_j}}$ are some deterministic orthonormal matrices with ranks $r_{W_j}$  and $r_{M_j}$, respectively, and $a_1, a_2 > 0$ are some bounded constants.
    Choose $m$ and $\{d_j\}_{j \in [s]}$ such that $$m = \lceil({\|\lambda^*\|_{W_2^{\alpha}(\mathbb{X})}}n)^{1/(2\alpha + d_{\max})}\rceil \;\; \text{and}\;\; d_{\max} + d_{\min} > D - 2\alpha,$$
    where $d_{\max} = \max\{d_1, \dots, d_s\}$ and $d_{\min} = \min\{d_1, \dots, d_s\}$.
    Suppose Assumptions~\ref{ass:approx_err1} and \ref{ass:approx_err2} hold and 
    \begin{align}\label{eq:singular_gap_proposition}
    \min_{j = 1}^s\{\sigma_{j,R_j}(\lambda^*)-\sigma_{j,R_j+1}(\lambda^*) \} \geq C_{\mathrm{gap}}\max\left\{{\|\lambda^*\|_{W_2^{\alpha}(\mathbb{X})}}m^{-\alpha}, \sqrt{\frac{\|\lambda^*\|_{\infty}m^{(D-d_{\min}) \vee d_{\max}}\log(m)}{n}}\right\},
    \end{align}
    where $C_{\mathrm{gap}} > 0$ is a sufficiently large constant.
    We have that with probability at least $1 - 3m^{-5}$,
    \begin{align*}
        \left\|\sin\Theta(\widehat U_{j}^{(0)}, U_{j})\right\|_{\op} \leq \frac{C}{\sigma_{j,R_j}(\lambda^*) - \sigma_{j,R_j+1}(\lambda^*)}\sqrt{\frac{\|\lambda^*\|_{\infty}(m^{D-d_j}+ m^{d_j})\log(m)}{n}}, \;\; \text{for all} \;\; j \in [s],
    \end{align*}
    where $C > 0$ is an absolute constant.
\end{proposition}

\begin{proposition}[Error bound on the refined coefficient tensor]\label{lemma:oshoop_approx}
    Let $\widehat{U}_{j}^{(1)} \in \mathbb{O}_{m^{d_j}, R_j}$, for $j \in [s]$, denote the outputs from \Cref{alg1}, and denote $U_{j} = \text{SVD}_{(R_j)} (\mathcal{M}_j(b^*)) \in \mathbb{O}_{m^{d_j}, R_j}$.
    Suppose the assumptions in \Cref{lemma:hosvd_approx} hold.
   Then
    \begin{align}\label{eq:tensor_error_frobenius}
        \left\|\widetilde b - b^*\right\|_{\mathrm{F}}^2
        =& O_{p}\left(\frac{\|\lambda^*\|_{\infty}(\sum_{j = 1}^sR_jm^{d_j}+\prod_{j = 1}^sR_j)\log(m)}{n} + \sum_{j = 1}^s\sum_{k = R_j+1}^{\infty}\sigma^2_{j,k}(\lambda^*)\right).
    \end{align}
\end{proposition}
There are two terms in the error bound \eqref{eq:tensor_error_frobenius}. The first term represents the estimation variance, which matches the minimax lower bound up to a log factor \cite[][Theorem~3]{zhang2018tensor}. The second term accounts for the remaining singular values, representing the approximation bias using a low-rank tensor. To the best of our knowledge, the tensor estimation in the approximately low-rank tensor settings has not been studied in the literature. This result and the theoretical tools we developed may be of independent interest.

These propositions provide essential bounds on the estimation errors of the singular vectors, which are critical for ensuring the accuracy of the tensor-based estimator $\widehat \lambda_{\mathrm{Tensor}}$.

\subsubsection{Proof of \Cref{thm:main_result}}
\begin{proof}[Proof of \Cref{thm:main_result}]\label{sec:proof_main}
Recall that 
\begin{align*}
    \widehat{\lambda}_{\mathrm{Tensor}} = \sum_{\mu_1 = 1}^{m^{d_1}}\cdots\sum_{\mu_s = 1}^{m^{d_s}} \widetilde b_{\mu_1,\dots,\mu_s}  \phi_{1,\mu_1}\cdots\phi_{s,\mu_s},
\end{align*}
where $\widetilde b$ is the coefficient tensor output by \Cref{alg1}.
We have
\begin{align*}
        \|\lambda^*-\widehat{\lambda}_{\mathrm{Tensor}}\|_{\mathbb{L}_2(\mathbb{X})}
        &\leq \|\lambda^*- \lambda^*\times_1 \mathcal{P}_{\mathcal{U}_1} \cdots \times_s \mathcal{P}_{\mathcal{U}_s}\|_{\mathbb{L}_2(\mathbb{X})} + \|\lambda^*\times_1 \mathcal{P}_{\mathcal{U}_1} \cdots \times_s \mathcal{P}_{\mathcal{U}_s} - \widehat{\lambda}_{\mathrm{Tensor}}\|_{\mathbb{L}_2(\mathbb{X})}\\
        &=: \mathrm{Term}_1 + \mathrm{Term}_2.
\end{align*}
By Assumptions~\ref{ass:approx_err1} and \ref{ass:approx_err2}, Lemma~2 in \cite{khoo2024nonparametric} shows that 
\begin{align*}
    \mathrm{Term}_1 = \left\|\lambda^*-\lambda^* \times_1 \mathcal{P}_{\mathcal{U}_1} \cdots  \times_s \mathcal{P}_{\mathcal{U}_s} \right\|_{\mathbb L_2(\mathbb{X})}=O\left({\|\lambda^*\|_{W_2^{\alpha}(\mathbb{X})}}m^{- \alpha}\right).
\end{align*}
For the second term, we have
\begin{align*}
    &(\mathrm{Term}_2)^2 = \|\lambda^*\times_1 \mathcal{P}_{\mathcal{U}_1} \cdots \times_s \mathcal{P}_{\mathcal{U}_s}- \widehat{\lambda}_{\mathrm{Tensor}}\|_{\mathbb{L}_2(\mathbb{X})}^2\\
    =& \sum_{\mu_1 = 1}^{m^{d_1}}\cdots\sum_{\mu_s = 1}^{m^{d_s}} \left\{\left(\lambda^*\times_1 \mathcal{P}_{\mathcal{U}_1} \cdots \times_s \mathcal{P}_{\mathcal{U}_s} - \widehat{\lambda}_{\mathrm{Tensor}}\right)[\phi_{1,\mu_1}, \cdots, \phi_{s,\mu_s}]\right\}^2\\
    =& \sum_{\mu_1 = 1}^{m^{d_1}}\cdots\sum_{\mu_s = 1}^{m^{d_s}} \Bigg\{\idotsint\left[\sum_{v_1 = 1}^{m^{d_1}}\cdots\sum_{v_s = 1}^{m^{d_s}} ( b_{v_1,\dots,v_s} - \widetilde b_{v_1,\dots,v_s})  \phi_{1,v_1}(x_1)\cdots\phi_{s,v_s}(x_s)\phi_{1,\mu_1}(x_1)\cdots\phi_{s,\mu_s}(x_s)\right]\nonumber\\
        &\quad\quad\quad\quad\quad\quad\quad\quad\quad \d x_1 \cdots\d x_s\Bigg\}^2\\
    =& \sum_{\mu_1 = 1}^{m^{d_1}}\cdots\sum_{\mu_s = 1}^{m^{d_s}} ( b^*_{\mu_1,\dots,\mu_s} - \widetilde b_{\mu_1,\dots,\mu_s})^2\\
    =& \|b^* - \widetilde b\|_{\mathrm{F}}^2\\
    =&O_{p}\left(\frac{\|\lambda^*\|_{\infty}(m^{d_{\max}}\sum_{j = 1}^sR_j+\prod_{j = 1}^sR_j)\log(m)}{n}  + \sum_{j = 1}^s\sum_{\mu_j = R_j+1}^{\infty}\sigma^2_{j,\mu_j}(\lambda^*)\right),
\end{align*}
where the fourth equality follows from the orthonormality of $\{\phi_{j,k}\}_{k = 1}^{m^{d_j}}$, and the last equality follows from \Cref{lemma:oshoop_approx}.
We have
\begin{align*}
        &\|\lambda^*-\widehat{\lambda}_{\mathrm{Tensor}}\|_{\mathbb{L}_2(\mathbb{X})}^2\\
        =& O_{p}\left(\frac{\|\lambda^*\|_{\infty}(m^{d_{\max}}\sum_{j = 1}^sR_j+\prod_{j = 1}^sR_j)\log(m)}{n} + {\|\lambda^*\|_{W_2^{\alpha}(\mathbb{X})}^2}m^{-2\alpha} + \sum_{j = 1}^s\sum_{\mu_j = R_j+1}^{\infty}\sigma^2_{j,\mu_j}(\lambda^*)\right)\\
        =& O_{p}\Bigg(\left\{\frac{\|\lambda^*\|_{\infty}{\|\lambda^*\|_{W_2^{\alpha}(\mathbb{X})}^{d_{\max}/(2\alpha + d_{\max})}}\sum_{j = 1}^sR_j}{n^{2 \alpha/(2\alpha + d_{\max})}}+ \frac{\|\lambda^*\|_{\infty}\prod_{j=1}^sR_j}{n}\right\}\log(m)\\
        &\qquad+ \frac{\|\lambda^*\|_{W_2^{\alpha}(\mathbb{X})}^{(2 \alpha + 2d_{\max})/(2\alpha + d_{\max})}}{n^{2 \alpha/(2\alpha + d_{\max})}} + \sum_{j = 1}^s\sum_{\mu_j = R_j+1}^{\infty}\sigma^2_{j,\mu_j}(\lambda^*)\Bigg)\\
        =& O_p\left(\left\{\frac{\|\lambda^*\|_{\infty}^{2\alpha/(2\alpha + d_{\max})}{\|\lambda^*\|_{W_2^{\alpha}}^{2d_{\max}/(2\alpha + d_{\max})}}\sum_{j = 1}^sR_j}{n^{2\alpha/(2\alpha+d_{\max})}} + \frac{\|\lambda^*\|_{\infty}\prod_{j=1}^sR_j}{n}\right\}\log(n) + \sum_{j = 1}^s\sum_{\mu_j = R_j+1}^{\infty}\sigma^2_{j,\mu_j}(\lambda^*)\right),
\end{align*}
where the second equality follows from the choice $m = ({\|\lambda^*\|_{W_2^{\alpha}(\mathbb{X})}}n)^{1/(2\alpha+d_{\max})}$, and the last equality follows from \Cref{ass:approx_err1}.

The proof concludes by noting that
\[
\max_{j = 1}^s\sum_{\mu_j = R_j+1}^{\infty}\sigma^2_{j,\mu_j}(\lambda^*)  \leq \xi_{(R_1, \dots, R_s)}^2,
\]
which follows from the definitions in \Cref{sec:approximately_low-rank_fct}, since for each $j \in [s]$, $\sum_{\mu_j = R_j+1}^{\infty}\sigma^2_{j,\mu_j}(\lambda^*)$ is the approximation error in squared $\mathbb{L}_2$ norm of the best rank-$(R_1, \dots, R_s)$ approximation function.
\end{proof}

\subsection{Proof for \Cref{sec:low_bound}}

\begin{proof}[Proof of \Cref{Minimax-Theorem-Tensor}]
Suppose $N^{(1)},\dots,N^{(n)}$ i.i.d.\ PPPs with common intensity $\lambda$. Let $P_\lambda$ denote the marginal law of one PPP and $P_\lambda^{(n)} = P_\lambda^{\otimes n}$ denote the joint law. Recall the function class for the true intensity function:
\begin{align*}
    \Lambda_{(R_1, \dots, R_s)}^{\alpha, s} = \left\{ \lambda^* : \mathbb{X} \to \mathbb{R}_+ \;\middle|\; \|\lambda^*\|_{W_2^\alpha(\mathbb{X})} < \infty, \; \|\lambda^*\|_{\infty} < \infty, \right. \nonumber \\ 
    \left.  \; \text{and} \ \inf_{\lambda \in \mathcal{T}_{(R_1, \dots, R_s)}} \|\lambda - \lambda^*\|_{\mathbb{L}_2(\mathbb{X})} \leq \xi_{(R_1, \dots, R_s)} \right\},
\end{align*}
where \[
\mathcal{T}_{(R_1, \dots, R_s)} = \left\{\lambda \in \mathbb{L}_2(\mathbb{X}): \rank(\lambda_j(x_j,x_{-j})) \leq R_j, \forall j \in [s]\right\}
\]
is the set of functions on $\mathbb{X}$, whose Tucker ranks are bounded by $(R_1, \dots, R_s)$.

\noindent\textbf{Step 1: Estimation lower bound: $n^{-2\alpha/(2\alpha+d_{\max})}$.}

We construct a finite subset of $\Lambda_{(R_1,\dots,R_s)}^{\alpha,s}$ and apply
Assouad's lemma.  Let $j^*\in[s]$ such that $d_{j^*}=d_{\max}$. Without loss of generality we assume $[0,1]^{d_j} \subset\mathbb X_j$ for all $j$, and 
$\prod_{j=1}^s[0,1]^{d_j}\subset\mathbb X$.

\smallskip
\noindent\textbf{Construct a subclass.}
Fix nonnegative, infinitely differentiable and compact supported functions $\eta_j\in C_c^\infty([0,1]^{d_j})$ for all $j\neq j^*$,
with $\eta_j \neq 0$, and define $\eta_-(x_{-j^*}) =\prod_{j\neq j^*}\eta_j(x_j)$.
For any $f:[0,1]^{d_{j^*}}\to\mathbb R_+$ with $\|f\|_{\infty} < \infty$ and $\|f\|_{W_2^\alpha([0,1]^{d_{j^*}})} < \infty$, define
\begin{equation}
\label{eq:lambda-f}
\lambda_f(x_1,\dots,x_s) = \eta_-(x_{-j^*})\,f(x_{j^*}).
\end{equation}
Note that $\lambda_f$ is multiplicative, so its Tucker rank is $(1, \dots, 1)$. Thus, $\lambda_f\in\mathcal T_{(R_1,\dots,R_s)}$ whenever $R_j\ge 1$ for all $j$. We have
\[
\|\lambda_f\|_{\infty} \le \|\eta_-\|_{\infty}\|f\|_{\infty} < \infty.
\]
Moreover, for any multi-index $\beta=(\beta^{(1)},\dots,\beta^{(s)})$ with $|\beta|_1\le \alpha$,
\[
\mathcal D^\beta \lambda_f(x)
=
\big(\mathcal D^{\beta^{(j^*)}} f\big)(x_{j^*})
\prod_{j\neq j^*}\big(\mathcal D^{\beta^{(j)}}\eta_j\big)(x_j).
\]
Note that by Fubini 
\begin{align*}
\|\lambda_f\|_{W_2^\alpha(\mathbb X)}^2
&= \sum_{0 \leq |\beta|_1 \leq \alpha} \|\mathcal D^\beta \lambda_f\|_{\mathbb{L}_2(\mathbb{X})}^2 = \sum_{0 \leq |\beta|_1 \leq \alpha} \big\|\mathcal D^{\beta^{(j^*)}} f\big\|_{\mathbb{L}_2}^2
\prod_{j\neq j^*}\big\|\mathcal D^{\beta^{(j)}}\eta_j\big\|_{\mathbb{L}_2}^2\\
&= \sum_{0 \leq |\gamma|_1 \leq \alpha} \big\|\mathcal D^{\gamma} f\big\|_{\mathbb{L}_2}^2
\sum_{\sum_{j \neq j^*}|\beta^{(j)}|_1 \leq \alpha - |\gamma|_1}\prod_{j\neq j^*}\big\|\mathcal D^{\beta^{(j)}}\eta_j\big\|_{\mathbb{L}_2}^2\\
&= \sum_{0 \leq |\gamma|_1 \leq \alpha} C_{\eta}(\gamma)\,\big\|\mathcal D^{\gamma} f\big\|_{\mathbb{L}_2}^2\\
&\le
C_\eta\,\|f\|_{W_2^\alpha([0,1]^{d_{j^*}})}^2 < \infty,
\end{align*}
for a constant $C_\eta = \max_{0 \le |\gamma|_1 \leq \alpha}C_{\eta}(\gamma) < \infty$ depending only on $\alpha$ and $\{\eta_j\}_{j\neq j^*}$.

Next, we provide the construction of a set of such $f$ functions. Fix $\phi\in C_c^\infty((0,1)^{d_{j^*}})$, i.e.~$\phi$ is infinitely differentiable on $(0,1)^d$ and has compact support, such that $\|\phi\|_{\mathbb L_2}=1$. Let $h \in (0, 1/4]$ whose value will be specified later.  Let $m = \lfloor 1/h\rfloor$, and partition $[0,1]^{d_{j^*}}$ into
$M = m^{d_{j^*}}$ disjoint cubes $\{Q_k\}_{k=1}^M$ of side length $\asymp h$. Since $\mathrm{supp}(\phi)\subset (0,1)^{d_{j^*}}$, for each cube $Q_k$ one can choose $x(k)\in Q_k$ such that
\[
x(k)+h\mathrm{supp}(\phi)\subset Q_k.
\]
Therefore, for
\[
\phi_k(x) = h^{-d_{j^*}/2}\phi\left(\frac{x-x(k)}{h}\right),
\]
we have
\[
\mathrm{supp}(\phi_k)=x(k)+h\mathrm{supp}(\phi)\subset Q_k.
\]
Then the supports are disjoint,
$\|\phi_k\|_{\mathbb L_2}=1$, and for any multi-index $\gamma$,
\[
\|\mathcal D^\gamma \phi_k\|_{\mathbb L_2} = h^{-|\gamma|_1}\|\mathcal D^\gamma\phi\|_{\mathbb L_2},
\qquad
\|\phi_k\|_\infty \le h^{-d/2}\|\phi\|_\infty.
\]
Let $\theta\in \Theta = \{-1,+1\}^M$. Fix $f_0>0$ and define
\begin{equation}\label{eq:f_theta}
f_\theta(x) =f_0 + a\sum_{k=1}^M \theta_k \phi_k(x),
\qquad x\in(0,1)^{d_{j^*}},
\end{equation}
with amplitude $a = c_a h^{\alpha + d_{j^*}/2}$ for a sufficiently small constant $c_a>0$.
Define the intensity functions on $\mathbb X$:
\begin{equation*}
\lambda_\theta(x_1,\dots,x_s) =\eta_-(x_{-j^*})\,f_\theta(x_{j^*}).
\end{equation*}
Then $\lambda_\theta\in\mathcal T_{(R_1,\dots,R_s)}$, whenever $R_j\ge 1$ for all $j$, and has
$\inf_{\lambda\in\mathcal T_{(R_1,\dots,R_s)}}\|\lambda-\lambda_\theta\|_{\mathbb{L}_2}=0$.
Since the supports of $\{\phi_k\}$ are disjoint, for each $x$ at most one $\phi_k(x)$ is nonzero, hence
\[
|f_\theta(x)-f_0|\le a\max_{1\le k \le M}\|\phi_k\|_\infty \le a h^{-d_{j^*}/2}\|\phi\|_\infty = c_a\|\phi\|_\infty h^\alpha.
\]
Choose $c_a > 0$ sufficiently small so that $c_a\|\phi\|_\infty h^\alpha\le f_0/2$. Then
\begin{equation}
\label{eq:f-positive}
f_\theta(x)\in[f_0/2,3f_0/2]\quad \forall x,
\end{equation}
so $\lambda_\theta\ge 0$ and $\|\lambda_\theta\|_\infty\le (3f_0/2)\|\eta_-\|_\infty$.

By disjoint supports and scaling, for $|\gamma|_1\le \alpha$,
\[
\|\mathcal D^\gamma(f_\theta-f_0)\|_{\mathbb L_2}^2
=
a^2\sum_{k=1}^M \|\mathcal D^\gamma \phi_k\|_{\mathbb L_2}^2
=
a^2 M h^{-2|\gamma|_1}\|\mathcal D^\gamma\phi\|_{\mathbb L_2}^2 \lesssim c_a^2h^{2\alpha+d_{j^*}}h^{-d_{j^*}}h^{-2|\gamma|_1}\|\mathcal D^\gamma\phi\|_{\mathbb L_2}^2 \lesssim 1,
\]
which follows from \(M=m^d\asymp h^{-d}\), \(a=c_a h^{\alpha+d/2}\) and $\phi\in C_c^\infty((0,1)^{d_{j^*}})$.
Hence $\|f_\theta\|_{W_2^\alpha([0,1]^{d_{j^*}})} < \infty$.
Therefore $\|\lambda_\theta\|_{W_2^\alpha(\mathbb X)} < \infty$ and $\lambda_\theta\in \Lambda_{(R_1,\dots,R_s)}^{\alpha,s}$ for every $\theta \in \Theta$.

\smallskip
\noindent\textbf{Separatability of the constructed subset $\{\lambda_{\theta}\}_{\theta \in \Theta}$.} 
For any $\theta, \theta^{\prime} \in \Theta$, we have by Fubini
\begin{align}\label{eq:reduction_Hamming}
\|\lambda_\theta-\lambda_{\theta^{\prime}}\|_{\mathbb{L}_2}^2 &= a^2\|\eta_-\|_{\mathbb{L}_2}^2\Big\|\sum_{k = 1}^M(\theta_k - \theta_k^{\prime})\phi_k\Big\|_{\mathbb{L}_2}^2 = a^2\|\eta_-\|_{\mathbb{L}_2}^2\sum_{k = 1}^M(\theta_k - \theta_k^{\prime})^2\|\phi_k\|_{\mathbb{L}_2}^2\nonumber\\
&= 4a^2\|\eta_-\|_{\mathbb{L}_2}^2d_H(\theta, \theta^{\prime}),
\end{align}
where the second equality follows from the disjoint supports, and the last equality follows from the fact that $\|\phi_k\|_{\mathbb{L}_2} = 1$ and $\sum_{k = 1}^M(\theta_k - \theta^{\prime}_k) = 4\sum_{k = 1}^M\mathbbm{1}\{\theta_k \neq \theta_k^{\prime}\} = 4d_H(\theta, \theta^{\prime})$ where $d_H(\theta, \theta^{\prime})$ is the Hamming distance between $\theta$ and $\theta^{\prime}$.

More specifically, let $\theta^{(k)}$ be $\theta$ with the $k$th bit flipped. Then $d_H(\theta,\theta^{(k)})=1$, and hence \eqref{eq:reduction_Hamming},
\begin{equation}
\label{eq:L2-sep}
\|\lambda_\theta-\lambda_{\theta^{(k)}}\|_{\mathbb L_2}^2
=
4a^2\|\eta_-\|_{\mathbb L_2}^2.
\end{equation}
Note that on $\mathbb{X} \setminus S = \{x \in \mathbb{X}: \eta_-(x_{-j^*}) = 0\}$, we have $\lambda_\theta=\lambda_{\theta^{(k)}}=0$.
By \Cref{thm:KL}, we have that
\begin{align*}
\mathrm{KL}(P_{\lambda_\theta}\|P_{\lambda_{\theta^{(k)}}})
&=
\int_{\mathbb X}\Big[
\lambda_\theta(x)\log\frac{\lambda_\theta(x)}{\lambda_{\theta^{(k)}}(x)}-\lambda_\theta(x)+\lambda_{\theta^{(k)}}(x)
\Big]\,d\upsilon(x)\\
&= \int_{S}\Big[
\lambda_\theta(x)\log\frac{\lambda_\theta(x)}{\lambda_{\theta^{(k)}}(x)}-\lambda_\theta(x)+\lambda_{\theta^{(k)}}(x)
\Big]\,d\upsilon(x),
\end{align*}
where the second equality follows from the convention that $0 \log 0 = 0$. On $S$, set $$t(x) = (\lambda_\theta(x)-\lambda_{\theta^{(k)}}(x))/\lambda_{\theta^{(k)}}(x),$$
so that 
\[
\lambda_\theta(x) = \lambda_{\theta^{(k)}}(x)(1 + t(x)),
\]
and
\begin{align*}
\mathrm{KL}(P_{\lambda_\theta}\|P_{\lambda_{\theta^{(k)}}})
&= \int_{S}
\lambda_{\theta^{(k)}}(x)\Big[(1 + t(x))\log(1+t(x))-t(x)\Big]\,d\upsilon(x).
\end{align*}
On $S$, by \eqref{eq:f-positive}, we choose $c_a$ so that
$|\lambda_\theta(x)-\lambda_{\theta^{(k)}}(x)| \le \lambda_{\theta^{(k)}}(x)/2$ for all $x$. Then we have $|t(x)| \leq 1/2$ for all $x$. Since for all $t$ with $|t| \le 1/2$, $(1+t)\log(1+t)-t \le t^2$, we have
\begin{align*}
\mathrm{KL}(P_{\lambda_\theta}\|P_{\lambda_{\theta^{(k)}}})
&\le \int_{S}
\frac{(\lambda_\theta(x)-\lambda_{\theta^{(k)}}(x))^2}{\lambda_{\theta^{(k)}}(x)}\,d\upsilon(x)\\
&\le \frac{\|\lambda_\theta-\lambda_{\theta^{(k)}}\|_{\mathbb L_2}^2}{\inf_{x \in S}\lambda_{\theta^{(k)}}(x)}\\
&\le Ca^2,
\end{align*}
where the last inequality follows the fact that $\inf_{x \in S}\lambda_{\theta^{(k)}}(x) > 0$ is bounded away from $0$ by construction, and $C > 0$ is a constant depending on $f_0$ and $\eta_-$. Since we are given the i.i.d.~PPPs $\{N^{i}\}_{i = 1}^n$, by \Cref{thm:KL} again, we have
\[
\mathrm{KL}(P_{\lambda_\theta}^{(n)}\|P_{\lambda_{\theta^{(k)}}}^{(n)}) \le Cna^2.
\]
By Pinsker's inequality, we have the TV distance is bounded by
\begin{equation}\label{eq:TV}
\mathrm{TV}(P_{\lambda_\theta}^{(n)},P_{\lambda_{\theta^{(k)}}}^{(n)}) \le \sqrt{\mathrm{KL}(P_{\lambda_\theta}^{(n)}\|P_{\lambda_{\theta^{(k)}}}^{(n)})/2} \le \sqrt{Cna^2/2}.
\end{equation}

\smallskip
\noindent\textbf{Apply Assouad’s lemma.}  To ensure $\mathrm{TV}(P_{\lambda_\theta}^{(n)},P_{\lambda_{\theta^{(k)}}}^{(n)}) \le 1/2$, according to \eqref{eq:TV}, we choose $h$ such that $n a^2 = n c_a^2 h^{2\alpha+d}\le c_0$ for a small absolute $c_0>0$,
i.e.\ $h\asymp n^{-1/(2\alpha+d)}$.

Given an arbitrary estimator $\widehat{\lambda}$, we define
\[
\widehat{\theta} \in \argmin_{\theta \in \Theta}\|\widehat{\lambda}- \lambda_{\theta}\|_{\mathbb{L}_2}^2.
\]
Then for any $\theta$,
\[
\|\lambda_{\widehat{\theta}}- \lambda_{\theta}\|_{\mathbb{L}_2} \le \|\widehat{\lambda}- \lambda_{\widehat\theta}\|_{\mathbb{L}_2} + \|\widehat{\lambda}- \lambda_{\theta}\|_{\mathbb{L}_2} \le 2\|\widehat{\lambda}- \lambda_{\theta}\|_{\mathbb{L}_2},
\]
so by \eqref{eq:reduction_Hamming}
\[
\|\widehat{\lambda}- \lambda_{\theta}\|_{\mathbb{L}_2}^2 \ge \frac{1}{4}\|\lambda_{\widehat{\theta}}- \lambda_{\theta}\|_{\mathbb{L}_2}^2 \ge a^2\|\eta_-\|_{\mathbb{L}_2}^2d_H(\widehat\theta, \theta).
\]
Note that $a = c_ah^{\alpha + d/2}$, $M = m^d$, $m = \lfloor 1/h\rfloor$, and we have chosen $h \asymp n^{-1/(2\alpha+d)}$. By Assouad's lemma \citep[Lemma 2.12 in][]{Tsybakov2009}, we have that
\begin{align*}
    \inf_{\widehat\lambda}\sup_{\theta}\mathbb{E}_{\theta}\|\widehat{\lambda}- \lambda_{\theta}\|_{\mathbb{L}_2}^2 \ge a^2\|\eta_-\|_{\mathbb{L}_2}^2\inf_{\widehat\theta}\sup_{\theta}d_H(\widehat\theta, \theta) \ge \frac{a^2\|\eta_-\|_{\mathbb{L}_2}^2M}{4} \asymp h^{2\alpha+d}h^{-d} \asymp n^{-2\alpha/(2\alpha+d)}.
\end{align*}
Hence
\[
\inf_{\widehat\lambda \in T_{(R_1, \dots, R_s)}}\sup_{\lambda^* \in \Lambda_{(R_1, \dots, R_s)}^{\alpha, s}}\|\widehat\lambda - \lambda^*\|_{\mathbb{L}_2}^2 \ge \inf_{\widehat\lambda \in T_{(R_1, \dots, R_s)}}\sup_{\theta}\|\widehat\lambda - \lambda_\theta\|_{\mathbb{L}_2}^2 \ge Cn^{-2\alpha/(2\alpha+d)}.
\]

\bigskip
\noindent\textbf{Step 2: Approximation lower bound: $\xi_{(R_1, \dots, R_s)}^2$.}
This step is deterministic. Note that for any estimator $\widehat\lambda\in\mathcal T_{(R_1, \dots, R_s)}$, we have 
\[
\|\widehat\lambda - \lambda^*\|_{\mathbb{L}_2} \ge \inf_{g \in \mathcal T_{(R_1, \dots, R_s)}}\|g - \lambda^*\|_{\mathbb{L}_2}.
\]
Let $\lambda^* \in \Lambda_{(R_1, \dots, R_s)}^{\alpha, s}$ such that $\inf_{g \in \mathcal{T}_{(R_1, \dots, R_s)}} \|g - \lambda^*\|_{\mathbb{L}_2(\mathbb{X})} = \xi_{(R_1, \dots, R_s)}$. Hence
\[
\inf_{\widehat\lambda \in T_{(R_1, \dots, R_s)}}\sup_{\lambda^* \in \Lambda_{(R_1, \dots, R_s)}^{\alpha, s}}\|\widehat\lambda - \lambda^*\|_{\mathbb{L}_2} \ge \xi_{(R_1, \dots, R_s)}.
\]

\bigskip
Combining \textbf{Step 1} and \textbf{Step 2} concludes the proof.
\end{proof}

\subsection{Proof for \Cref{sec:alg1_results}}

\begin{proof}[Proof of \Cref{lemma:hosvd_approx}]
    We simplify the notation. Let $\widehat b = \widehat b^{H_1}$.
    We only upper bound $\|\sin\Theta(\widehat U_{j}^{(0)}, U_{j})\|_{\op}$ with $j = 1$, since the same argument applies with $j = 2,\dots,s$.
    Let $Y = \mathcal{M}_1(\widehat b) \in \mathbb R^{m^{d_1} \times m^{D-d_1}}$, $X = \mathcal{M}_1(b^*) \in \mathbb R^{m^{d_1} \times m^{D-d_1}}$, $Z = Y - X = \mathcal{M}_1(\widehat b - b^*)$, $\widehat U = \widehat U_1^{(0)} \in \mathbb O_{m^{d_1}, R_1}$ and $\widehat V = \widehat V_1^{(0)} \in \mathbb O_{m^{D-d_1}, R_1}$ be the left and right singular vectors of $Y$, as well as $U = U_1 \in \mathbb O_{m^{d_1}, R_1}$ and $V = V_1 \in \mathbb O_{m^{D-d_1}, R_1}$ be the left and right singular vectors of $X$.

\noindent\textbf{Step 1.}
    We start by giving some deviation bounds to be used in the rest of the proof. 
    We apply  \eqref{eq:fundamental_general} with $W_1 = I_{m^{d_1}}$ and $W_2 = I_{m^{D-d_1}}$.
    Recall that we choose
    \begin{align*}
    m = ({\|\lambda^*\|_{W_2^{\alpha}(\mathbb{X})}}n)^{1/(2\alpha + d_{\max})}.
    \end{align*}
    We have, under \Cref{ass:approx_err1},  $$\sqrt{\frac{\|\lambda^*\|_{\infty}(m^{D-d_1}+ m^{d_1})\log(m)}{n}} \gtrsim \sqrt{\frac{m^{d_{\max}}\log(m)}{n}} \gtrsim \frac{m^{D/2}\log(m)}{n},$$
    for some constant $C > 0$, where the first inequality follows because (i) if $d_1 = d_{\max}$ then $m^{D-d_1} + m^{d_1} > m^{d_{\max}}$ or (ii) if $d_1 \neq d_{\max}$ then $D - d_1 > d_{\max}$.
    Thus, there exists an absolute constant $a > 0$ such that the following event
    \begin{align}\label{eq:Z_op_bound}
        \mathcal{E}_1 = \left\{\left\|Z\right\|_{\op} \leq a\sqrt{\frac{\|\lambda^*\|_{\infty}(m^{D-d_1}+ m^{d_1})\log(m)}{n}}\right\}
    \end{align}
    holds with probability at least $1 - m^{-5}$. 

    We apply \eqref{eq:fundamental_general} with $W_1 = I_{m^{d_1}}$ and $W_2 = V \in \mathbb{O}_{m^{D-d_1}, R_1}$. Since $m = ({\|\lambda^*\|_{W_2^{\alpha}(\mathbb{X})}}n)^{1/(2\alpha + d_{\max})} =  C(\|\lambda^*\|_{\infty}n)^{1/(2\alpha + d_{\max})}$, for some constant $C > 0$, and $d_{\max} + d_{\min} > D - 2\alpha$, we have that there exists an absolute constant $C_1 > 0$ such that $\|\lambda^*\|_{\infty}n \geq C_1m^{2\alpha + d_{\max}}$
    and
    \[
        \sqrt{\frac{\|\lambda^*\|_{\infty}(m^{d_1}+ R_1)\log(m)}{n}} \geq \sqrt{\frac{\|\lambda^*\|_{\infty}m^{d_{\min}}\log(m)}{n}} \geq \sqrt{\frac{\|\lambda^*\|_{\infty}m^{D-2\alpha-d_{\max}}}{n}}\log(m) \geq C_1^{1/2}\frac{m^{D/2}\log(m)}{n}.
    \]
    Thus, there exists an absolute constant $a > 0$ such that the following event
    \begin{align}\label{eq:Z_op_bound_2}
        \mathcal{E}_2 = \left\{\left\|ZV\right\|_{\op} \leq a\sqrt{\frac{\|\lambda^*\|_{\infty}(m^{d_1} +R_1)\log(m)}{n}}\right\}
    \end{align}
    holds with probability at least $1 - m^{-5}$. 

    We apply \eqref{eq:fundamental_general} with $W_1 = U \in \mathbb{O}_{m^{d_1}, R_1}$ and $W_2 = I_{m^{D-d_1}}$. Since $m = ({\|\lambda^*\|_{W_2^{\alpha}(\mathbb{X})}}n)^{1/(2\alpha + d_{\max})} = C(\|\lambda^*\|_{\infty}n)^{1/(2\alpha + d_{\max})}$, for some constant $C > 0$, and $d_{\max} + d_{\min} > D - 2\alpha$, we have there exists an absolute constant $C_1 > 0$ such that $\|\lambda^*\|_{\infty}n \geq C_1m^{2\alpha + d_{\max}} \geq C_1m^{d_{\max}}\log(m)$
    and
    \[
        \sqrt{\frac{\|\lambda^*\|_{\infty}(m^{D-d_1}+ R_1)\log(m)}{n}} \geq \sqrt{\frac{\|\lambda^*\|_{\infty}m^{D-d_{\max}}\log(m)}{n}}  \geq C_1^{1/2}\frac{m^{D/2}\log(m)}{n}.
    \]
    Thus, there exists an absolute constant $a > 0$ such that the following event
    \begin{align}\label{eq:Z_op_bound_3}
        \mathcal{E}_3 = \left\{\left\|Z^{\top}U\right\|_{\op} \leq a\sqrt{\frac{\|\lambda^*\|_{\infty}(m^{D-d_1} +R_1)\log(m)}{n}}\right\}
    \end{align}
    holds with probability at least $1 - m^{-5}$. 
    
    We assume $\mathcal{E}_1 \cap \mathcal{E}_2 \cap \mathcal{E}_3$ holds  throughout the rest of proof.

\noindent\textbf{Step 2.}
    By \Cref{thm:wedin}, we have
    \begin{align}\label{eq:proof_sintheta_approx}
        \left\|\sin\Theta(\widehat U, U)\right\|_{\op} \leq \frac{\|Z\|_{\op}}{\sigma_{R_1}(X) - \sigma_{R_1+1}(X) - \|Z\|_{\op}}.
    \end{align}
    By \eqref{eq:Z_op_bound}, we can upper bound $\|Z\|_{\op}$. To upper bound $\|\sin\Theta(\widehat U, U)\|_{\op}$, we only need to lower bound the denominator of \eqref{eq:proof_sintheta_approx}:
    \begin{align}
        &\sigma_{R_1}(X) - \sigma_{R_1+1}(X) - \|Z\|_{\op}\nonumber\\
        =& \sigma_{1,R_1}(\lambda^* \times_1 \mathcal{P}_{\mathcal{U}_1} \cdots \times_s \mathcal{P}_{\mathcal{U}_s}) - \sigma_{1,R_1+1}(\lambda^* \times_1 \mathcal{P}_{\mathcal{U}_1} \cdots \times_s \mathcal{P}_{\mathcal{U}_s}) - \|Z\|_{\op}\nonumber\\
        \geq& \sigma_{1,R_1}(\lambda^*) - \sigma_{1,R_1+1}(\lambda^*) - \left|\sigma_{1,R_1}(\lambda^*) - \sigma_{1,R_1}(\lambda^* \times_1 \mathcal{P}_{\mathcal{U}_1} \cdots \times_s \mathcal{P}_{\mathcal{U}_s}) \right|\nonumber\\
        &- \left|\sigma_{1,R_1+1}(\lambda^*) - \sigma_{1,R_1+1}(\lambda^* \times_1 \mathcal{P}_{\mathcal{U}_1} \cdots \times_s \mathcal{P}_{\mathcal{U}_s}) \right| -\|Z\|_{\op}\nonumber\\
        \geq& \sigma_{1,R_1}(\lambda^*) - \sigma_{1,R_1+1}(\lambda^*) - 2\|\mathcal{M}_1(\lambda^*) - \mathcal{M}_1(\lambda^* \times_1 \mathcal{P}_{\mathcal{U}_1} \cdots \times_s \mathcal{P}_{\mathcal{U}_s})\|_{\op} - \|Z\|_{\op}\nonumber\\
    \geq& \sigma_{1,R_1}(\lambda^*) - \sigma_{1,R_1+1}(\lambda^*) - 2\|\lambda^* - \lambda^* \times_1 \mathcal{P}_{\mathcal{U}_1} \cdots \times_s \mathcal{P}_{\mathcal{U}_s}\|_{\mathbb{L}_2(\mathbb{X})} - \|Z\|_{\op}\nonumber\\
    \geq& \sigma_{1,R_1}(\lambda^*) - \sigma_{1,R_1+1}(\lambda^*) - O({\|\lambda^*\|_{W_2^{\alpha}(\mathbb{X})}}m^{-\alpha}) - a\sqrt{\frac{\|\lambda^*\|_{\infty}(m^{D-d_1}+ m^{d_1})\log(m)}{n}} \nonumber\\
    \geq& C_2\left\{\sigma_{1,R_1}(\lambda^*) - \sigma_{1,R_1+1}(\lambda^*)\right\} \label{eq:sin_gap_ineq},
    \end{align}
    where $C_2 \in (0,1)$ is an absolute constant. The first equality following from $\sigma_{k}(\mathcal{M}_1(b^*)) = \sigma_{k}(\mathcal{M}_1(\lambda^* \times_1 \mathcal{P}_{\mathcal{U}_1} \cdots \times_s \mathcal{P}_{\mathcal{U}_s}))$. The first inequality follows from the triangle inequality. The second inequality follows from \Cref{LemmaDaren}, since $\mathcal{M}_1(\lambda^*)$ and $\mathcal{M}_1(\lambda^* \times_1 \mathcal{P}_{\mathcal{U}_1} \cdots \times_s \mathcal{P}_{\mathcal{U}_s})$ are two compact operators on Hilbert space. The fourth inequality follows from \eqref{eq:approx_error} and \eqref{eq:Z_op_bound}. The last inequality follows from the condition \eqref{eq:singular_gap_proposition}.
    \
    \\
    This yields that, for the absolute constant $C_3 = a/C_2 > 0$,
    \begin{align*}
        \left\|\sin\Theta(\widehat U, U)\right\|_{\op} \leq& \frac{\|Z\|_{\op}}{\sigma_{R_1}(X) - \sigma_{R_1+1}(X) - \|Z\|_{\op}}\\
        & \leq \frac{C_3}{\sigma_{1,R_1}(\lambda^*) - \sigma_{1,R_1+1}(\lambda^*)}\sqrt{\frac{\|\lambda^*\|_{\infty}(m^{D-d_1}+ m^{d_1})\log(m)}{n}},
    \end{align*}
    which concludes the proof. 
\end{proof}

\begin{proof}[Proof of \Cref{lemma:oshoop_approx}]
\
\\
\noindent\textbf{Step 0: notation.}
Let $p_j = m^{d_j}$, $p_{-j} = \prod_{\ell\neq j}p_\ell$, $R_{-j} = \prod_{\ell\neq j}R_\ell$ and $D =\sum_{j=1}^s d_j$. For $j \in [s]$, let
\[W^*_{-j} = \otimes_{\ell\neq j} U_\ell^*\in\mathbb O_{p_{-j}, R_{-j}} \quad W^{(0)}_{-j} = \otimes_{\ell\neq j}\widehat U_\ell^{(0)}\in\mathbb O_{p_{-j}, R_{-j}} \quad W^{(1)}_{-j} = \otimes_{\ell\neq j}\widehat U_\ell^{(1)}\in\mathbb O_{p_{-j}, R_{-j}}.\]
For a matrix $U\in\mathbb O_{p,r}$, write
$\mathcal P_U = UU^\top$ and $\mathcal P_{U^\perp} =I-\mathcal P_U$. Using kronecker product properties, we have 
\[\mathcal P_{\otimes_{\ell \neq j}U_{\ell}} = \otimes_{\ell \neq j}\mathcal P_{U_\ell} = \otimes_{\ell \neq j}(U_\ell U_\ell^{\top}) = (\otimes_{\ell \neq j} U_\ell) \cdot (\otimes_{\ell \neq j} U_\ell)^{\top}.\]

For each split $H_k$ ($k\in[3]$), define the noise tensor
\[
Z^{(k)} = \widehat b^{H_k}-b^*.
\]
Recall that Algorithm~\ref{alg1} outputs
\[
\widetilde b
=
\widehat b^{H_3}\times_1\mathcal P_{\widehat U_1^{(1)}}\times_2\cdots\times_s\mathcal P_{\widehat U_s^{(1)}}.
\]
Let $\delta_R = \min_{j = 1}^s\{\sigma_{j,R_j}(\lambda^*)-\sigma_{j,R_j+1}(\lambda^*) \}$.

The first for-loop in \Cref{alg1} outputs for each $j \in [s]$
    \begin{align*}
        \widehat U^{(0)}_j = \mathrm{SVD}_{(R_j)}(\mathcal{M}_j(\widehat b^{H_1})) \in \mathbb O_{p_j, R_j}.
    \end{align*}
    By \Cref{lemma:hosvd_approx}, we have that with probability at least $1 - 3m^{-5}$,
    \begin{align}\label{eq:initial sin theta}
        & \max_{j = 1}^s\left\|\sin\Theta(\widehat U_j^{(0)}, U_j^*)\right\|_{\op}\nonumber\\
        \leq&
        C\max_{j = 1}^s\left\{\frac{1}{\sigma_{j,R_j}(\lambda^*) - \sigma_{j,R_j+1}(\lambda^*)}\sqrt{\frac{\|\lambda^*\|_{\infty}(p_{-j}+ p_j)\log(m)}{n}}\right\}\nonumber\\
        \leq& c,
    \end{align}
    where $c \in (0,1)$ is some sufficiently small constant, and this follows from \eqref{eq:singular_gap_proposition} where $C_{\mathrm{gap}}$ is sufficiently large.

    \medskip
\noindent\textbf{Step 1: decompose $\|\widetilde b-b^*\|_{\mathrm F}$.}
Write
\[
\widetilde b-b^*
=
(b^*+Z^{(3)})\times_1\mathcal P_{\widehat U_1^{(1)}}\cdots\times_s\mathcal P_{\widehat U_s^{(1)}}-b^*,
\]
so by the triangle inequality,
\begin{equation}
\label{eq:I1I2_def}
\|\widetilde b-b^*\|_{\mathrm F}
\le
\mathrm{Term}_1+\mathrm{Term}_2,
\end{equation}
where
\begin{align*}
\mathrm{Term}_1
&=
\big\|b^*\times_1\mathcal P_{\widehat U_1^{(1)}}\times_2\cdots\times_s\mathcal P_{\widehat U_s^{(1)}}-b^*\big\|_{\mathrm F},\\
\mathrm{Term}_2
&=
\big\|Z^{(3)}\times_1\mathcal P_{\widehat U_1^{(1)}}\times_2\cdots\times_s\mathcal P_{\widehat U_s^{(1)}}\big\|_{\mathrm F}.
\end{align*}

\medskip
\noindent\textbf{Step 2: bound $I_2$.} 
Fix any $j\in[s]$ and define
\[
 T
=
Z^{(3)}\times_1\mathcal P_{\widehat U_1^{(1)}}\times_2\cdots\times_s\mathcal P_{\widehat U_s^{(1)}}.
\]
Using the standard unfolding identity for multilinear products,
\[
\mathcal M_j( T)
=
\mathcal P_{\widehat U_j^{(1)}}\cdot
\mathcal M_j(Z^{(3)})\cdot
\big(\otimes_{\ell\neq j}\mathcal P_{\widehat U_\ell^{(1)}}\big) = \mathcal P_{\widehat U_j^{(1)}}\cdot
\mathcal M_j(Z^{(3)})\cdot
W_{-j}^{(1)} \cdot {W_{-j}^{(1)}}^{\top}.
\]
By \Cref{lemma:Fobenius norm preserve},
\[
\mathrm{Term}_2 = \|T\|_{\mathrm F}
=
\|\mathcal M_j( T)\|_{\mathrm F}
=
\big\|\widehat U_j^{(1)\top}\cdot \mathcal M_j(Z^{(3)})\cdot W_{-j}^{(1)}\big\|_{\mathrm F}.
\]
Note that, $\{w_t\}_{t = 1}^{R_{-j}}$, the columns of $W_{-j}^{(1)}$, are orthonormal.
We have
\[
\big\|\widehat U_j^{(1)\top}\cdot \mathcal M_j(Z^{(3)})\cdot W^{(1)}_{-j}\big\|_{\mathrm F}^2
=
\sum_{t=1}^{R_{-j}}\big\|\widehat U_j^{(1)\top}\cdot \mathcal M_j(Z^{(3)})\cdot w_t\big\|_2^2
=
\sum_{t=1}^{R_{-j}}\big\|\widehat U_j^{(1)\top}\cdot \mathcal M_j(Z^{(3)})\cdot w_t\big\|_{\op}^2.
\]
Due to the sample splitting, we have $\{\widehat{U}_j^{(1)}, W_{-j}^{(1)}\}$ is independent of $Z^{(3)}$.
Apply \Cref{lemma:deviation_op_ortho} with $W_1=\widehat U_j^{(1)}$ of rank $R_j$ and
$W_2=w_t$ of rank $1$, we have with probability at least $1 - Cm^{-5}$, for all $t$
\[
\max_{j = 1}^s\big\|\widehat U_j^{(1)\top}\cdot \mathcal M_j(Z^{(3)})\cdot w_t\big\|_{\op}
\le
a_1\sqrt{\frac{\|\lambda^*\|_\infty (R_j+1)\log m}{n}}
+
a_2\frac{m^{D/2}\log m}{n}.
\]
Therefore,
\begin{equation}
\label{eq:I2_bound_raw}
\mathrm{Term}_2
=
\|T\|_{\mathrm F}
\le
\sqrt{R_{-j}}\left(
a_1\sqrt{\frac{\|\lambda^*\|_\infty (R_j+1)\log m}{ n}}
+
a_2\frac{m^{D/2}\log m}{ n}
\right).
\end{equation}

\medskip
\noindent\textbf{Step 3: Bound $\mathrm{Term}_1$ by a telescoping sum.}
Using the fact that $\|\mathcal P_{\widehat U_\ell^{(1)}}\|_{\op}\le 1$,
a standard telescoping argument yields
\begin{equation}
\label{eq:telescoping_I1}
\mathrm{Term}_1
=
\big\|b^*\times_1\mathcal P_{\widehat U_1^{(1)}}\cdots\times_s\mathcal P_{\widehat U_s^{(1)}}-b^*\big\|_{\mathrm F}
\le
\sum_{k=1}^s
\|b^*\times_k(I-\mathcal P_{\widehat U_k^{(1)}})\|_{\mathrm F}
=
\sum_{k=1}^s
\|b^*\times_k\mathcal P_{\widehat U_{k\,\perp}^{(1)}}\|_{\mathrm F}.
\end{equation}

\medskip
\noindent\textbf{Step 3.1: control of $\|b^*\times_k\mathcal P_{\widehat U_{k\,\perp}^{(1)}}\|_{\mathrm F}$.}
We consider only the case $k = 1$. Recall that $U_j^*\in\mathbb O_{p_j,R_j}$ is the top-$R_j$ left singular subspace of
$\mathcal M_j(b^*)$.
We have that
\begin{align}\label{eq:main_b^*_0}
    & \left\| b^* \times_1 \mathcal P_{\widehat U_{1\,\perp}^{(1)}}\right\|_{\mathrm{F}}\nonumber\\
    &\leq \left\| b^* \times_1 \mathcal P_{\widehat U_{1\,\perp}^{(1)}} \times_{2} \mathcal{P}_{U_{2}^*}\right\|_{\mathrm{F}} + \left\| b^* \times_1 \mathcal P_{\widehat U_{1\,\perp}^{(1)}} \times_{2} \mathcal{P}_{U_{2\,\perp}^{*}}\right\|_{\mathrm{F}}\nonumber\\
    &\leq \left\| b^* \times_1 \mathcal P_{\widehat U_{1\,\perp}^{(1)}} \times_{2} \mathcal{P}_{U_{2}^*}\right\|_{\mathrm{F}} + \left\| b^* \times_{2} \mathcal{P}_{U_{2\,\perp}^{*}}\right\|_{\mathrm{F}}\nonumber\\
    &\leq \left\| b^* \times_1 \mathcal P_{\widehat U_{1\,\perp}^{(1)}} \times_{2} \mathcal{P}_{U_{2}^*} \times_{3} \mathcal{P}_{U_{3}^*} \right\|_{\mathrm{F}} + 
    \left\|   b^* \times_1 \mathcal P_{\widehat U_{1\,\perp}^{(1)}} \times_{2}   \mathcal{P}_{U_{2}^*} \times_{3} \mathcal{P}_{U_{3\,\perp}^{*}}\right\|_{\mathrm{F}} +  
    \left\| b^* \times_{2} \mathcal{P}_{U_{2\,\perp}^{*}}\right\|_{\mathrm{F}}   
    \nonumber\\
    &
    \le   \left\| b^* \times_1 \mathcal P_{\widehat U_{1\,\perp}^{(1)}} \times_{2} \mathcal{P}_{U_{2}^*} \times_{3} \mathcal{P}_{U_{3}^*} \right\|_{\mathrm{F}} + 
    \left\|   b^*  \times_{3} \mathcal{P}_{U_{3\,\perp}^{*}}\right\|_{\mathrm{F}}+ 
    \left\| b^* \times_{2} \mathcal{P}_{U_{2\,\perp}^{*}}\right\|_{\mathrm{F}}   \nonumber\\ 
    &\le \dots\nonumber\\
    &\le \left\| b^* \times_1 \mathcal P_{\widehat U_{1\,\perp}^{(1)}} \times_{2} \mathcal{P}_{U_{2}^*} \dots \times_{s} \mathcal{P}_{U_{s}^*} \right\|_{\mathrm{F}} +  \sum_{j \neq 1}\left\| b^* \times_{j} \mathcal{P}_{U_{j\,\perp}^{*}}\right\|_{\mathrm{F}}\nonumber\\
    &= \left\| b^* \times_1 \mathcal P_{\widehat U_{1\,\perp}^{(1)}} \times_{2} \mathcal{P}_{U_{2}^*} \dots \times_{s} \mathcal{P}_{U_{s}^*} \right\|_{\mathrm{F}} + \sum_{j=2}^s\sqrt{\sum_{k = R_j+1}^{p_j \wedge p_{-j}}\sigma_k^2(\mathcal{M}_j(b^*))}.
\end{align}

\medskip
\noindent\textbf{Step 3.2: bound $\| b^* \times_1 \mathcal P_{\widehat U_{1\,\perp}^{(1)}} \times_{2} \mathcal{P}_{U_{2}^*} \dots \times_{s} \mathcal{P}_{U_{s}^*} \|_{\mathrm{F}}$.}
Consider a different but relevant  quantity   
\begin{align*}
    &\left\| b^* \times_1 \mathcal P_{\widehat U_{1\,\perp}^{(1)}} \times_{2} \mathcal{P}_{\widehat U_{2}^{(0)}} \dots \times_{s} \mathcal{P}_{\widehat U_{s}^{(0)}} \right\|_{\mathrm{F}}\\
    &=  \left\|\mathcal P_{\widehat U_{1\,\perp}^{(1)}} \cdot \mathcal{M}_1(b^*) \cdot (\otimes_{\ell \neq 1} \mathcal{P}_{\widehat U_{\ell}^{(0)}}) \right\|_{\mathrm{F}}\\
    &=  \left\|\mathcal P_{\widehat U_{1\,\perp}^{(1)}} \cdot \mathcal{M}_1(b^*) \cdot W_{-1}^{(0)} \cdot {W_{-1}^{(0)}}^{\top} \right\|_{\mathrm{F}}\\
    &=  \left\|\mathcal P_{\widehat U_{1\,\perp}^{(1)}} \cdot \mathcal{M}_1(b^*) \cdot W_{-1}^{(0)} \right\|_{\mathrm{F}}\\
    &\geq  \left\|\mathcal P_{\widehat U_{1\,\perp}^{(1)}} \cdot \mathcal{M}_1(b^*) \cdot \mathcal{P}_{\otimes_{\ell \neq 1} U_\ell^*} \cdot W_{-1}^{(0)} \right\|_{\mathrm{F}}\\
    & \quad - \left\|\mathcal P_{\widehat U_{1\,\perp}^{(1)}} \cdot \mathcal{M}_1(b^*) \cdot (I - \mathcal{P}_{\otimes_{\ell \neq 1} U_\ell^*}) \cdot W_{-1}^{(0)} \right\|_{\mathrm{F}}\\
    &=: {II}_1 - {II}_2,
\end{align*}
where the third equality follows from \Cref{lemma:Fobenius norm preserve}. Before analyzing the terms ${II}_1$ and ${II}_2$, we firstly note that
\begin{equation}\label{eq:facts}
\begin{aligned}
     (W_{-1}^*)^{\top} \cdot W_{-1}^{(0)} &= \otimes_{\ell \neq 1} (U_\ell^{*\,\top} U_{\ell}^{(0)}),\\
     \sigma_{\min}\left(\otimes_{\ell \neq 1}(U_\ell^{*\,\top} U_{\ell}^{(0)})\right) &= \prod_{\ell \neq 1}\sigma_{\min}\left(U_\ell^{*\,\top} U_{\ell}^{(0)}\right),\\
    \sigma_{\min}^2\left(U_k^{*\,\top} U_{k}^{(0)}\right) &= 1 - \left\|U_{k\,\perp}^{*\,\top} U_{k}^{(0)}\right\|_{\op}^2 = 1 - \left\|\sin\Theta(U_{k}^{*}, U_{k}^{(0)})\right\|_{\op}^2,
\end{aligned}
\end{equation}
which hold due to the properties of the Kronecker product and \Cref{lemma:sintheta}.
\\
For the term ${II}_1$, we have that
\begin{align*}
    {II}_1 &=  \left\|\mathcal P_{\widehat U_{1\,\perp}^{(1)}} \cdot \mathcal{M}_1(b^*) \cdot W_{-1}^* \cdot (W_{-1}^*)^{\top} \cdot W_{-1}^{(0)} \right\|_{\mathrm{F}}\\
    &\geq \left\|\mathcal P_{\widehat U_{1\,\perp}^{(1)}} \cdot \mathcal{M}_1(b^*) \cdot W_{-1}^*\right\|_{\mathrm{F}} \cdot \prod_{\ell \neq 1}\sigma_{\min}\left(U_\ell^{*\,\top} U_{\ell}^{(0)}\right)\\
    &= \left\|b^* \times_1 \mathcal{P}_{\widehat U_{1\,\perp}^{(1)}} \times_2 \mathcal{P}_{U_{2}^{*}} \dots \times_s \mathcal{P}_{U_{s}^{*}}\right\|_{\mathrm{F}} \cdot \prod_{\ell \neq 1}\sqrt{1 - \left\|\sin\Theta(U_{\ell}^{*}, U_{\ell}^{(0)})\right\|_{\op}^2}
    \\
    &\ge 
    \Big(\frac{3}{4}\Big)^{(s-1)/2}\left\|b^* \times_1 \mathcal{P}_{\widehat U_{1\,\perp}^{(1)}} \times_2 \mathcal{P}_{U_{2}^{*}} \dots \times_s \mathcal{P}_{U_{s}^{*}}\right\|_{\mathrm{F}},
\end{align*}
where the first inequality follows from \Cref{lemma:lb_forbuious}, and the second equality follows from \eqref{eq:facts}, and the last inequality follows from \eqref{eq:initial sin theta}, i.e.~$\max_{\ell}\|\sin\Theta(U_{\ell}^{*}, U_{\ell}^{(0)})\| \le c \le 1/2$.
\\
For the term ${II}_2$, we have that
\begin{align*}
    {II}_2 &= \left\|\mathcal P_{\widehat U_{1\,\perp}^{(1)}} \cdot \mathcal{M}_1(b^*) \cdot (I - \mathcal{P}_{\otimes_{\ell \neq 1} U_\ell^*}) \cdot W_{-1}^{(0)} \right\|_{\mathrm{F}}\\
    &= \left\|\mathcal P_{\widehat U_{1\,\perp}^{(1)}} \cdot \mathcal{M}_1(b^*) \cdot (W_{-1}^*)_{\perp} \cdot (W_{-1}^*)_{\perp}^{\top} \cdot W_{-1}^{(0)} \right\|_{\mathrm{F}}\\
    &\leq \left\|\mathcal P_{\widehat U_{1\,\perp}^{(1)}}\right\|_{\op} \left\|\mathcal{M}_1(b^*) \cdot (W_{-1}^*)_{\perp} \right\|_{\mathrm{F}} \left\| (W_{-1}^*)_{\perp}^{\top} \cdot W_{-1}^{(0)} \right\|_{\op}\\
    &\le   \left\|\mathcal{M}_1(b^*) \cdot (W_{-1}^*)_{\perp} \right\|_{\mathrm{F}},
\end{align*}
where the last inequality follows from $\|P_{\widehat U_{1\,\perp}^{(1)}}\|_{\op} \le 1$ and $\| (W_{-1})_{-1}^\top W_{-1}^{(0)} \|_{\op} \le \|(W_{-1})_{\perp} \|_{\op} \| W_{-1}^{(0)} \|_{\op} \le 1$.
By \Cref{lemma:complement_d},
\begin{align*}
    II_2 \le \left\| \mathcal{M}_1(b^*) \cdot (W_{-1}^*)_{\perp}\right\|_{\mathrm{F}} 
    \le  \sum_{j = 2}^s\left\| b^* \times_{j} \mathcal{P}_{U_{j\,\perp}^{*}}\right\|_{\mathrm{F}} = \sum_{j=2}^s\sqrt{\sum_{k = R_j+1}^{p_j \wedge p_{-j}}\sigma_k^2(\mathcal{M}_j(b^*))}.
\end{align*}
Combining ${II}_1$ and ${II}_2$, we have
\begin{align}\label{eq:mode_k_bound_1}
      &\left\| b^* \times_1 \mathcal{P}_{\widehat U_{1\,\perp}^{(1)}} \times_{2} \mathcal{P}_{U_{2}^*} \dots \times_{s} \mathcal{P}_{U_{s}^*} \right\|_{\mathrm{F}}\nonumber\\ 
    \le \; &  \Big(\frac{4}{3}\Big)^{(s-1)/2}\left\{\left\| b^* \times_1 P_{\widehat U_{1\,\perp}^{(1)}} \times_{2} \mathcal{P}_{\widehat U_{2}^{(0)}} \dots \times_{s} \mathcal{P}_{\widehat U_{s}^{(0)}} \right\|_{\mathrm{F}}  + \sum_{j=2}^s\sqrt{\sum_{k = R_j+1}^{p_j \wedge p_{-j}}\sigma_k^2(\mathcal{M}_j(b^*))}\right\}.
\end{align}
It remains to bound $\| b^* \times_1 P_{\widehat U_1^{(1)\perp}} \times_{2} \mathcal{P}_{\widehat U_{2}^{(0)}} \dots \times_{s} \mathcal{P}_{\widehat U_{s}^{(0)}}\|_{\mathrm{F}}$.

\medskip
\noindent\textbf{Step 3.3: bound $\| b^* \times_1 P_{\widehat U_1^{(1)\perp}} \times_{2} \mathcal{P}_{\widehat U_{2}^{(0)}} \dots \times_{s} \mathcal{P}_{\widehat U_{s}^{(0)}} \|_{\mathrm{F}}$.}
Recall that $W_{-k}^{(0)} =\otimes_{\ell\neq k}\widehat U_\ell^{(0)}\in \mathbb O_{p_{-k},R_{-k}}$. For a fixed $k\in[s]$, we define
\[
A_k
=
\mathcal M_k(\widehat b^{H_2})\cdot W_{-k}^{(0)},
\quad
B_k
=
\mathcal M_k(b^*)\cdot W_{-k}^{(0)},
\quad
Z_k
=
A_k-B_k
=
\mathcal M_k(Z^{(2)})\cdot W_{-k}^{(0)}.
\]
Hence \Cref{lemma:perterbation_approx} gives
\begin{align}
\label{eq:lem7_apply}
&\left\| b^* \times_1 P_{\widehat U_1^{(1)\perp}} \times_{2} \mathcal{P}_{\widehat U_{2}^{(0)}} \dots \times_{s} \mathcal{P}_{\widehat U_{s}^{(0)}} \right\|_{\mathrm{F}}\nonumber\\
&= \|\mathcal P_{\widehat U_1^{(1)\perp}}\cdot B_1\|_{\mathrm F}\nonumber\\
&\le
3\|B_1-(B_1)_{(R_1)}\|_{\mathrm F}
+
2\min\{\sqrt{R_1}\|Z_1\|_{\op},\,\|Z_1\|_{\mathrm F}\}\nonumber\\
&\le
3\|B_1-(B_1)_{(R_1)}\|_{\mathrm F}
+
2\sqrt{R_1}\|Z_1\|_{\op}.
\end{align}
Since $\otimes_{\ell\neq k}\widehat U_\ell^{(0)}$ is orthonormal, by \Cref{lemma:singular values preserve}, $0 \le \sigma_t(B_k) \le \sigma_t(\mathcal{M}_k(b^*))$ for any $t$. We have that
\begin{align*}
\|B_1-(B_1)_{(R_1)}\|_{\mathrm F}^2
&= \sum_{t = R_1+1}^{p_1 \wedge R_{-1}}\sigma_t^2(B_1) \le \sum_{t = R_1+1}^{p_1 \wedge p_{-1}}\sigma_t^2(\mathcal{M}_1(b^*)).
\end{align*}
    
Due to the sample splitting used in \Cref{alg1}, $Z^{(2)}$ and $W_{-1}^{(0)}$ are independent.
Apply \Cref{lemma:deviation_op_ortho} with $W_1=I_{p_1}$ and $W_2=W_{-1}^{(0)}$ with rank $R_{-1}$, we have with probability at least $1 - m^{-5}$
\begin{equation*}
\label{eq:Zk_op_bound}
\|Z_1\|_{\op}
=
\Big\|\mathcal M_1(Z^{(2)})\cdot W_{-1}^{(0)}\Big\|_{\op}
\le
a_1\sqrt{\frac{\|\lambda^*\|_\infty(p_1+R_{-1})\log m}{n}}
+
a_2\frac{m^{D/2}\log m}{n}.
\end{equation*}
We obtain
\begin{align}
\label{eq:mode_k_bound}
&\left\| b^* \times_1 \mathcal P_{\widehat U_1^{(1)\perp}} \times_{2} \mathcal{P}_{\widehat U_{2}^{(0)}} \dots \times_{s} \mathcal{P}_{\widehat U_{s}^{(0)}} \right\|_{\mathrm{F}}\nonumber\\
&\le
C\left\{
\sqrt{\sum_{t = R_1+1}^{p_1 \wedge p_{-1}}\sigma_t^2(\mathcal{M}_k(b^*))}
+
\sqrt{R_1}\left(
a_1\sqrt{\frac{\|\lambda^*\|_\infty(p_1+R_{-1})\log m}{n}}
+
a_2\frac{m^{D/2}\log m}{n}
\right)
\right\}.
\end{align}
Combining \eqref{eq:mode_k_bound_1} and \eqref{eq:mode_k_bound}, we have
\begin{align*}
      &\left\| b^* \times_1 \mathcal{P}_{\widehat U_1^{(1)\perp}} \times_{2} \mathcal{P}_{U_{2}^*} \dots \times_{s} \mathcal{P}_{U_{s}^*} \right\|_{\mathrm{F}}\\ 
    \le \; &  \Big(\frac{4}{3}\Big)^{(s-1)/2}C_1
\sqrt{R_1}\left(
a_1\sqrt{\frac{\|\lambda^*\|_\infty(p_1+R_{-1})\log m}{\bar n}}
+
a_2\frac{m^{D/2}\log m}{\bar n}
\right)\\
&+ \Big(\frac{4}{3}\Big)^{(s-1)/2}C_2\sum_{j=1}^s\sqrt{\sum_{k = R_j+1}^{p_j \wedge p_{-j}}\sigma_k^2(\mathcal{M}_j(b^*))}.
\end{align*}
Note that for each $j$
\[
 \sum_{k = R_j+1}^{p_j \wedge p_{-j}}\sigma_k^2(\mathcal{M}_j(b^*)) \le \sum_{k=R_j+1}^{\infty}\sigma_k^2\left(\mathcal{M}_j(b^*)\right)
= \sum_{k=R_j+1}^{\infty}\sigma_k^2\left(\mathcal{M}_j(\lambda^* \times_1 \mathcal{P}_{U^*_1} \cdots \times_s \mathcal{P}_{U^*_s})\right).
\]
Let $\llbracket \mathcal{M}_j(\lambda^* \times_1 \mathcal{P}_{
    U^*_1} \cdots \times_s \mathcal{P}_{U^*_s})\rrbracket_{(R_j)}$ and $\llbracket \mathcal{M}_j(\lambda^*)\rrbracket_{(R_j)}$ denote the best rank-$R_j$ approximations of $\mathcal{M}_j(\lambda^* \times_1 \mathcal{P}_{U^*_1} \cdots \times_s \mathcal{P}_{U^*_s})$ and $\mathcal{M}_j(\lambda^*)$, respectively. We have
    \begin{align}\label{eq:proof_hosvd_V_approx}
    &\sum_{k=R_j+1}^{\infty}\sigma_k^2\left(\mathcal{M}_j(\lambda^* \times_1 \mathcal{P}_{U^*_1} \cdots \times_s \mathcal{P}_{U^*_s})\right)\nonumber\\
    =& \sum_{k=1}^{\infty}\left|\sigma_k\left(\mathcal{M}_j(\lambda^* \times_1 \mathcal{P}_{U^*_1} \cdots \times_s \mathcal{P}_{U^*_s})\right) - \sigma_k\left(\llbracket\mathcal{M}_j(\lambda^* \times_1 \mathcal{P}_{U^*_1} \cdots \times_s \mathcal{P}_{U^*_s})\rrbracket_{(R_j)}\right)\right|^2\nonumber\\
    \leq& \left\|\mathcal{M}_j(\lambda^* \times_1 \mathcal{P}_{U^*_1} \cdots \times_s \mathcal{P}_{U^*_s}) - \llbracket\mathcal{M}_j(\lambda^* \times_1 \mathcal{P}_{U^*_1} \cdots \times_s \mathcal{P}_{U^*_s})\rrbracket_{(R_j)}\right\|_{\mathbb{L}_2}^2\nonumber\\
    \leq& \left\|\mathcal{P}_{U^*_j}\cdot\mathcal{M}_j(\lambda^*) \cdot W_{-j}^*  - \mathcal{P}_{U_j^*}\cdot \llbracket \mathcal{M}_j(\lambda^*) \rrbracket_{(R_j)} \cdot W_{-j}^* \right\|_{\mathbb{L}_2}^2\nonumber\\
    \leq& \left\|\mathcal{M}_j(\lambda^*) - \llbracket \mathcal{M}_j(\lambda^*)\rrbracket_{(R_j)} \right\|_{\mathbb{L}_2}^2\nonumber\\
    =& \sum_{k=R_j+1}^{\infty}\sigma_{j,k}^2\left(\lambda^*\right),
    \end{align}
    where the first inequality follows from \Cref{lem:Mirsky_Hilbert}. The second inequality follows since $\mathcal{P}_{U^*_j}\cdot \llbracket \mathcal{M}_j(\lambda^*) \rrbracket_{(R_j)} \cdot W_{- j}^*$ is of rank at most $R_j$. 
Therefore,
\begin{align}\label{eq:main_b^*}
      &\left\| b^* \times_1 \mathcal{P}_{\widehat U_1^{(1)\perp}} \times_{2} \mathcal{P}_{U_{2}^*} \dots \times_{s} \mathcal{P}_{U_{s}^*} \right\|_{\mathrm{F}}\nonumber\\ 
    \le \; &  \Big(\frac{4}{3}\Big)^{(s-1)/2}C_1
\sqrt{R_1}\left(
a_1\sqrt{\frac{\|\lambda^*\|_\infty(p_1+R_{-1})\log m}{n}}
+
a_2\frac{m^{D/2}\log m}{n}
\right)\nonumber\\
&+ \Big(\frac{4}{3}\Big)^{(s-1)/2}C_2\sum_{j=1}^s\sqrt{\sum_{k=R_j+1}^{\infty}\sigma_{j,k}^2\left(\lambda^*\right)}.
\end{align}
Combining \eqref{eq:telescoping_I1}, \eqref{eq:main_b^*_0} and \eqref{eq:main_b^*}, we have
\begin{align}
\label{eq:telescoping_I1_final}
\mathrm{Term}_1
&\le 
\sum_{k=1}^s
\|b^*\times_k\mathcal P_{\widehat U_{k\,\perp}^{(1)}}\|_{\mathrm F}\nonumber\\
&\leq \Big(\frac{4}{3}\Big)^{(s-1)/2}C_1\sum_{k = 1}^s
\sqrt{R_k}\left(\sqrt{\frac{\|\lambda^*\|_\infty(p_k+R_{-k})\log m}{n}}
+
\frac{m^{D/2}\log m}{n}
\right)\nonumber\\
&\quad + s\bigg(\Big(\frac{4}{3}\Big)^{(s-1)/2}C_2+1\bigg)\sum_{j=1}^s\sqrt{\sum_{k=R_j+1}^{\infty}\sigma_{j,k}^2\left(\lambda^*\right)}.
\end{align}

\medskip
\noindent\textbf{Step 4.}
Finally, combining \eqref{eq:I1I2_def}, \eqref{eq:I2_bound_raw} and \eqref{eq:telescoping_I1_final}, we have
\begin{align*}
    &\|\widetilde b-b^*\|_{\mathrm F}\nonumber\\
&\leq \Big(\frac{4}{3}\Big)^{(s-1)/2}C_1\sum_{k = 1}^s
\sqrt{R_k}\left(\sqrt{\frac{\|\lambda^*\|_\infty(p_k+R_{-k})\log m}{n}}
+
\frac{m^{D/2}\log m}{n}
\right)\nonumber\\
&\quad + s\bigg(\Big(\frac{4}{3}\Big)^{(s-1)/2}C_2+1\bigg)\sum_{j=1}^s\sqrt{\sum_{k=R_j+1}^{\infty}\sigma_{j,k}^2\left(\lambda^*\right)}\nonumber\\
&\quad+C_3\sqrt{R_{-j}}\left(
\sqrt{\frac{\|\lambda^*\|_\infty (R_j+1)\log m}{ n}}
+\frac{m^{D/2}\log m}{ n}
\right).
\end{align*}
By \Cref{ass:approx_err1}, we have $\|\lambda^*\|_{W_2^{\alpha}(\mathbb{X})}\|\lambda^*\|_{\infty}^{-1} = O(1)$ and $m = ({\|\lambda^*\|_{W_2^{\alpha}(\mathbb{X})}}n)^{1/(2\alpha + d_{\max})} = C(\|\lambda^*\|_{\infty}n)^{1/(2\alpha + d_{\max})}$, for some constant $C > 0$ and $d_{\max} + d_{\min} > D - 2\alpha$, we have that there exists an absolute constant $C_3 > 0$ such that $\|\lambda^*\|_{\infty}n \geq C_3m^{2\alpha + d_{\max}}$
    and
    \[
        \sqrt{\frac{\|\lambda^*\|_{\infty}(p_k+ R_{-k})\log(m)}{n}} \geq \sqrt{\frac{\|\lambda^*\|_{\infty}m^{d_{\min}}\log(m)}{n}} \geq \sqrt{\frac{\|\lambda^*\|_{\infty}m^{D-2\alpha-d_{\max}}}{n}}\log(m) \geq C_3^{1/2}\frac{m^{D/2}\log(m)}{n}.
    \]
Thus, with proability at least $1 - Cm^{-5}$
\begin{align}
    &\|\widetilde b-b^*\|_{\mathrm F}\nonumber\\
&\leq C_4
\sqrt{\frac{\|\lambda^*\|_{\infty}(\sum_{k = 1}^sp_kR_k + \prod_{k=1}^sR_{k})\log m}{n}} + C_5\sum_{j=1}^s\sqrt{\sum_{k=R_j+1}^{\infty}\sigma_{j,k}^2\left(\lambda^*\right)}.
\end{align}
\end{proof}

\section{Other point processes}\label{sec:other_pp}
\subsection{Neyman-Scott point processes}
\label{sec:Neymann-Scott}
    A Cox process $N \subseteq \mathbb{X} \subset \mathbb{R}^D$ is a point process with  random intensity process $\{\Lambda(x): x \in \mathbb{X}\}$  characterized by the following two properties.
\begin{enumerate}
    \item $\{\Lambda(x): x \in \mathbb{X}\}$ is non-negative valued random process.
    \item Conditional on a realization $\lambda(\cdot)$ of $\Lambda(\cdot)$, $N$ is an inhomogeneous Poission point process with intensity function $\lambda(\cdot)$. In this context, $\lambda(\cdot)$ is also called the local intensity.
\end{enumerate}
    Special examples of the Cox processes include the Log Gaussian Cox processes \citep{moller1998log} and the Neyman-Scott processes \citep{neyman1958statistical}.
In this section, we consider the Neyman-Scott point processes, which belong to the Cox point processes with specific forms of the random intensity processes.  Let $N$ be an inhomogeneous Neyman-Scott point process with random intensity process $\{\Lambda(x): x \in \mathbb{X}\}$, such that
\begin{align}\label{eq:NSPP_random_intensity}
\Lambda(x) = \ell(x)\sum_{c \in N_C}k(x, c),
\end{align}
where $\ell: \mathbb{X} \to \mathbb{R}_+$ is a deterministic locally non-negative intergrable function, $N_C$ is an inhomogeneous Poisson point process defined on $\mathbb{X}$ with intensity function $\lambda_C(\cdot)$ assumed to be locally integrable, and $k: \mathbb{X} \times \mathbb{X} \to \mathbb{R}_+$ is a kernel density function, in the sense that for all $x \in \mathbb{X}$, $k(x, \cdot)$ is a density function on $\mathbb{X}$.  The intensity function of $N$ is
\[
\lambda^*(\cdot) = \mathbb{E}[\Lambda(\cdot)] = \ell(\cdot)\int_{\mathbb{X}}k(\cdot, c)\lambda_C(c)\d c.
\]

Let $\{N^{(i)}\}_{i = 1}^n$ be a set of i.i.d.~inhomogeneous Neyman-Scott point processes, with random intensity processes $\{\{\Lambda^{(i)}(x): x \in \mathbb{X}\}\}_{i = 1}^n$ and with intensity function $\lambda^*: \mathbb{X} \to \mathbb{R}_+$, for $D \in \mathbb{N}_+$. We apply our tensor-based method, describe in \Cref{alg1}, to estimate $\lambda^*$. The theoretical guarantees are provided in the following corollary.

\begin{corollary}\label{thm:main_result_NS}
    Let $\{N^{(i)}\}_{i = 1}^n$ be a set of i.i.d.~inhomogeneous Neyman-Scott point processes, with random intensity processes $\{\{\Lambda^{(i)}(x): x \in \mathbb{X}\}\}_{i = 1}^n$ and with intensity function $\lambda^*$. Assume that $\Lambda^{(i)}$ are uniformly bounded almost surely, i.e.~$\max_{i = 1}^n\|\Lambda^{(i)}\|_{\infty} \leq C_{\Lambda} < \infty$.  Let $\widehat{\lambda}_{\mathrm{Tensor}}$ be the tensor-based estimator output by \Cref{alg1} with the target Tucker rank $(R_1, \dots, R_s)$ and choose
    \[
m = ({\|\lambda^*\|_{W_2^{\alpha}(\mathbb{X})}}n)^{1/(2\alpha+d_{\max})},
\]
where $d_{\max} = \max\{d_1, \dots, d_s\}$ and $\alpha \geq 1$ is the smoothness parameter of $\lambda^*$.
Suppose Assumptions \ref{ass:approx_err1}, \ref{ass:approx_err2} and \ref{ass:SingularGap_B*} hold, and then
we have
\[
        \|\lambda^*-\widehat{\lambda}_{\mathrm{Tensor}}\|_{\mathbb{L}_2}^2
        = O_p\left(\left\{\frac{\|\lambda^*\|_{\infty}^{2\alpha/(2\alpha+d_{\max})}{\|\lambda^*\|_{W_2^{\alpha}(\mathbb{X})}^{2d_{\max}/(2\alpha+d_{\max})}}\sum_{j = 1}^sR_j}{n^{2\alpha/(2\alpha+d_{\max})}} + \frac{\prod_{j=1}^sR_j}{n}\right\}\log(n) + \xi_{(R_1, \dots, R_s)}^2\right),
\]
where $\xi_{(R_1, \dots, R_s)}$ represents the minimum approximation error to $\lambda^*$ for each  rank-$(R_1, \dots, R_s)$ tensor, as defined in \eqref{eq:tucker_low-rank_approx}.
\end{corollary}

\subsection{Joint density estimation for stationary $D$-dependent time series}\label{sec:time_series}
A point process can be viewed as a random sample where the sample size may be random and the sample points may exhibit dependence. In this sense, samples of i.i.d.~random variables and sequences of time series data can be viewed as special cases of point processes.

In this section, we consider a point process formed by a time series. For $D \in \mathbb{N}_+$, a sequence of random variables $\{X_i\}_{i \in \mathbb{Z}}$ is said to be $D$-dependent if, for each pair of integers satisfying $|r -s| > D$, the random variables $X_r$ and $X_s$ are independent.  If the process $\{X_i\}_{i \in \mathbb{Z}}$ is stationary and $D$-dependent, then the distribution of this process is fully determined by the joint density of $Y_s = (X_{s}, \dots, X_{s+D-1})^{\top}$ independent of the starting time index $s$.

Let $\{X_{t}\}_{t = 1}^n \subset \mathbb{R}$ be a stationary $D$-dependent process.  Let $f^*: \mathbb{R}^D \to \mathbb{R}_+$ denote the joint density function of a segment $Y_{i}$ of $D$ consecutive random variables, i.e.~$Y_{i} = (X_{i}, \dots, X_{i+D-1})^{\top}$.  Note that $f^*$ fully characterizes the distribution of the $D$-dependent process.  Suppose $f^{*}$ satisfies that $f^* \in \mathbb{L}_2(\mathbb{R}^D)$ and $\|f^*\|_\infty < \infty$.  

Consider the sequence of point processes $\{N^{(i)} = Y_i\}_{i = 1}^{n+1-D}$.  We apply our tensor-based method described in \Cref{alg1} to estimate $\lambda^*$. The theoretical guarantees are provided in the following corollary.

\begin{corollary}\label{thm:main_result_ts}
    Let $\{N^{(i)} = Y_i\}_{i = 1}^{n+1-D}$ be a set of dependent point processes formed by a stationary $D$-dependent process $\{X_{t}\}_{t = 1}^n$ described above.  Let $\widehat{f}_{\mathrm{Tensor}}$ be the tensor-based estimator output by \Cref{alg1} with the target Tucker rank $(R_1, \dots, R_s)$ and choose
    \[
m = \left\{{\|\lambda^*\|_{W_2^{\alpha}(\mathbb{X})}^2}(n+1-D)\right\}^{1/(2\alpha+d_{\max})},
\]
where $d_{\max} = \max\{d_1, \dots, d_s\}$ and $\alpha \geq 1$ is the smoothness parameter of $f^*$.
Suppose Assumptions \ref{ass:approx_err1}, \ref{ass:approx_err2} and \ref{ass:SingularGap_B*} (with $\lambda^*$ therein replaced by $f^*$) hold, and then
we have
\[
        \|f^*-\widehat{f}_{\mathrm{Tensor}}\|_{\mathbb{L}_2}^2
        = O_p\left(\left\{\frac{\|\lambda^*\|_{\infty}^{2\alpha/(2\alpha+d_{\max})}{\|\lambda^*\|_{W_2^{\alpha}(\mathbb{X})}^{2d_{\max}/(2\alpha+d_{\max})}}\sum_{j = 1}^sR_j}{(n+1-D)^{2\alpha/(2\alpha+d_{\max})}} + \frac{\prod_{j=1}^sR_j}{n+1-D}\right\}\log(n) + \xi_{(R_1, \dots, R_s)}^2\right),
\]
where $\xi_{(R_1, \dots, R_s)}$ represents the minimum approximation error to $f^*$ for each rank-$(R_1, \dots, R_s)$ tensor, as defined in \eqref{eq:tucker_low-rank_approx}. 
\end{corollary}

\subsection{Proof for Section~\ref{sec:other_pp}}
We recall some notations given in \Cref{sec:math_setup}. Suppose we factorizes the domain of the point processes
\begin{align*}
\mathbb{X} = \mathbb{X}_1 \times \cdots \times \mathbb{X}_s \subset \mathbb{R}^{d_1}\times \cdots \times \mathbb{R}^{d_s} = \mathbb{R}^{D}, \quad \text{with} \ \sum_{j=1}^s d_j = D.
\end{align*}
For each coordinate space $\mathbb{X}_j$, we select orthonormal basis \(\{\phi_{j,\mu_j}\}_{\mu_j = 1}^{m^{d_j}} \subset \mathbb{L}_2(\mathbb{X}_j)\). Projecting \(\lambda^*\) onto the corresponding finite-dimensional subspace yields the coefficients
\[
b^*_{\mu_1, \dots, \mu_s} = \lambda^*[\phi_{1,\mu_1}, \dots, \phi_{s,\mu_s}],
\]
which naturally organizes into a tensor \(b^* \in \mathbb{R}^{m^{d_1} \times \cdots \times m^{d_s}}\).
Define the empirical measure
\[
\widehat{\lambda} = \frac{1}{n}\sum_{i = 1}^n\sum_{u \in N^{(i)}}\delta_u,
\]
where $\delta_u$ is a point mass at $u$.
Denote the empirical coefficient tensor by $\widehat{b}$ whose entries are
\[
\widehat{b}_{\mu_1,\dots,\mu_s} =  \widehat{\lambda}[\phi_{1,\mu_1}, \dots, \phi_{s,\mu_s}]
= \frac{1}{n}\sum_{i = 1}^n\sum_{X^{(i)} \in N^{(i)}} \phi_{1,\mu_1}(X_{1}^{(i)}) \cdots \phi_{s,\mu_s}(X_{s}^{(i)}),
\]
where \(X^{(i)} = (X^{(i)}_{1}, \dots, X^{(i)}_{s}) \in \mathbb{X}\) represents a point in \(N^{(i)}\).

\begin{lemma}[Fundamental bound for Neyman-Scott point process]\label{lemma:fundamental_NS}
    Consider the same setting as in \Cref{thm:main_result_NS}. For all deterministic matrices $W_j \in \mathbb O_{m^{d_j}, r_{W_j}}$ and $V_j \in \mathbb O_{m^{D-d_j}, r_{V_j}}$ with ranks $r_{W}$ and $r_{V}$, respectively, we have that with probability at least $1 - 2m^{-5}$, uniformly for all $j \in [s]$
    \begin{align*}
        \left\|W_j^{\top} \cdot \mathcal{M}_j(\widehat b - b^*) \cdot V_j\right\|_{\op} \leq  C\left\{a_1\sqrt{\frac{(r_{W_j}+r_{V_j})\log(m)}{n}} + a_2\frac{m^{D/2}\log(m)}{n}\right\},
    \end{align*}
    where $a_1 = \sqrt{C_{\Lambda}} + \sqrt{\|\ell\|_{\infty}\|k\|_{\infty}\|\lambda_C\|_{\infty}}$, $a_2 = C_{\phi}^s (1 + \|\ell\|_{\infty} \|k\|_{\infty})$, and $C > 0$ is an absolute constant.
\end{lemma}
\begin{proof}[Proof of \Cref{lemma:fundamental_NS}]
    We obtain the upper bound using the same arguments in the proof of  \Cref{lemma:deviation_op_ortho}, and we only focus on the case with $j = 1$.
    We decompose
    \begin{align*}
    W^{\top} \cdot \mathcal{M}_1(\widehat b - b^*) \cdot V =& W^{\top} \cdot\left(\mathcal{M}_1(\widehat b) - \mathbb{E}\left[\mathcal{M}_1(\widehat b)\Big|\{\Lambda^{(i)}\}_{i = 1}^n\right]\right) \cdot V\\
    &+ W^{\top} \cdot\left(\mathbb{E}\left[\mathcal{M}_1(\widehat b)\Big|\{\Lambda^{(i)}\}_{i = 1}^n\right] - \mathcal{M}_1(b^*)\right) \cdot V.
    \end{align*}
    Let $$W^{\top} \cdot \mathcal{M}_1(\widehat b) \cdot V = \frac{1}{n}\sum_{i = 1}^n\sum_{X \in N^{(i)}}F(X) \in \mathbb{R}^{r_{W} \times r_{V}},$$
    where $X = (X_1^{\top}, \dots, X_d^{\top})^{\top} \in \mathbb{X}$ with $X_j \in \mathbb{X}_{j}$, and $x \mapsto F(x)$ is an $\mathbb{R}^{r_{W} \times r_{V}}$-valued function with the $(j,l)$ entry \begin{align*}
    F_{(j; l)}(x) =& \sum_{\mu_1 = 1}^{m^{d_1}}W_{(\mu_1;j)}\phi_{\mu_1}(x_1)\sum_{\mu_2 = 1}^{m^{d_2}}\cdots\sum_{\mu_s = 1}^{m^{d_s}}V_{(\mu_2, \dots, \mu_s; l)}\phi_{\mu_2}(x_2)\cdots\phi_{\mu_s}(x_s)\\
    =& \psi_{j}(x_1)\sum_{\mu_2 = 1}^{m^{d_2}}\cdots\sum_{\mu_s = 1}^{m^{d_s}}V_{(\mu_2, \dots, \mu_s; l)}\phi_{\mu_2}(x_2)\cdots\phi_{\mu_s}(x_s),
    \end{align*}
    where $W_{(\mu_1;j)}$ is the $(\mu_1,j)$ entry of $W$, and  each combination of $\mu_2, \dots, \mu_s$ corresponds to a row of $V$ denoted by $V_{(\mu_2, \dots, \mu_s; \cdot)}$.  For $j \in [r_W]$, we let  $\psi_{j}(\cdot) = \sum_{\mu_1 = 1}^{m^s}W_{(\mu_1;j)}\phi_{j}(\cdot)$. Note that $\{\psi_{j}\}_{j = 1}^{r_W}$ is a set of orthonormal basis, since $\{\phi_{\mu_1}\}_{\mu_1 = 1}^{m^{d_1}}$ is a set of orthonormal basis and $\{W_{(\cdot; j)}\}_{j=1}^{r_W}$ is a set of orthonormal vectors.
      We can also write
    $$W^{\top} \cdot \mathbb{E}\left[\mathcal{M}_1(\widehat b)\Big|\{\Lambda^{(i)}\}_{i = 1}^n\right]\cdot V = \frac{1}{n}\sum_{i = 1}^n\int_{\mathbb{X}}F(x)\Lambda^{(i)}(x)\d x = \frac{1}{n}\sum_{i = 1}^n\int_{\mathbb{X}}F(x)\ell(x)\sum_{c \in N_C^{(i)}}k(x, c)\d x,$$
    and
    $$W^{\top} \cdot \mathcal{M}_1(b^*) \cdot V = \int_{\mathbb{X}}F(x)\lambda^*(x)\d x = \int_{\mathbb{X}}\left\{\int_{\mathbb{X}}F(x)\ell(x)k(x, c)\d x\right\}\lambda_C(c)\d c.$$
    It follows that
    \begin{align}\label{eq:proof_NSPP_decomp}
        &\mathbb{P}\left(\left\|W^{\top} \cdot \mathcal{M}_1(\widehat b - b^*) \cdot V\right\|_{\op} \geq t_1 + t_2\right)\nonumber\\
        \leq& \mathbb{P}\bigg(\left\| W^{\top} \cdot \left(\mathcal{M}_1(\widehat b) - \mathbb{E}\left[\mathcal{M}_1(\widehat b)\Big|\{\Lambda^{(i)}\}_{i = 1}^n\right]\right) \cdot V\right\|_{\op}\nonumber\\
        & \quad+ \left\| W^{\top} \cdot \left(\mathbb{E}\left[\mathcal{M}_1(\widehat b)\Big|\{\Lambda^{(i)}\}_{i = 1}^n\right] - \mathcal{M}_1(b^*)\right) \cdot V\right\|_{\op} \geq t_1 + t_2\bigg) \nonumber\\
        \leq&  \mathbb{P}\left(\left\| W^{\top} \cdot \left(\mathcal{M}_1(\widehat b) - \mathbb{E}\left[\mathcal{M}_1(\widehat b)\Big|\{\Lambda^{(i)}\}_{i = 1}^n\right]\right) \cdot V\right\|_{\op} \geq t_1\right) \nonumber\\
        &+ \mathbb{P}\left(\left\| W^{\top} \cdot \left(\mathbb{E}\left[\mathcal{M}_1(\widehat b)\Big|\{\Lambda^{(i)}\}_{i = 1}^n\right] - \mathcal{M}_1(b^*)\right) \cdot V\right\|_{\op} \geq t_2\right)\nonumber\\
        =:& I + II,
    \end{align}
    where the first inequality follows from the triangle inequality, and the second inequality follows from the union bound. 
    Next, we obtain the tail probability bounds in \eqref{eq:proof_NSPP_decomp}.
    \
    \\
    \textbf{Step~1:} Tail probability bound on $I$.  Note that the random intensity processes $\Lambda^{(i)}$ are uniformly bounded almost surely, i.e.~$\max_{i = 1}^n\|\Lambda^{(i)}\|_{\infty} \leq C_{\Lambda} < \infty$. Conditional on $\Lambda^{(i)}$, $N^{(i)}$ is an inhomogeneous Poisson point process with intensity function $\Lambda^{(i)}$.  By the same arguments in the proof of \Cref{lemma:deviation_op_ortho}, we have
    \begin{align*}
        &\mathbb{P}\left(\left\| W^{\top} \cdot \left(\mathcal{M}_1(\widehat b) - \mathbb{E}\left[\mathcal{M}_1(\widehat b)\Big|\{\Lambda^{(i)}\}_{i = 1}^n\right]\right) \cdot V\right\|_{\op}
        \geq C\left\{\sqrt{\frac{C_{\Lambda}(r_V+r_W)\log(m)}{n}} + \frac{C_{\phi}^sm^{D/2}\log(m)}{n}\right\}\Bigg | \Lambda \right)\\
        &\leq m^{-5},
    \end{align*}
    where $C >0$ is an absolute constant.  Taking the expation with respective to the random intensity $\Lambda$, we have
    \begin{align*}
        &\mathbb{P}\left(\left\| W^{\top} \cdot \left(\mathcal{M}_1(\widehat b) - \mathbb{E}\left[\mathcal{M}_1(\widehat b)\Big|\{\Lambda^{(i)}\}_{i = 1}^n\right]\right) \cdot V\right\|_{\op}
        \geq C\left\{\sqrt{\frac{C_{\Lambda}(r_V+r_W)\log(m)}{n}} + \frac{C_{\phi}^sm^{D/2}\log(m)}{n}\right\} \right)\\
        &\leq m^{-5}.
    \end{align*}
    \
    \\
    \textbf{Step~2.} Tail probability bound on $II$.  Let \[
    F^{\prime}(c) = \int_{\mathbb{X}}F(x)\ell(x)k(x, c)\d x.
    \]
    We verify the conditions of \Cref{coro:MatrixBernstein_PPP}. Note that
    \begin{align*}
        L = &\sup_{c \in \mathbb{X}}\|F^{\prime}(c)\|_{\op} \leq \|\ell\|_{\infty}\sup_{x \in \mathbb{X}}\|F(x)\|_{\op}\sup_{c \in \mathbb{X}}\int_{[0,1]^D}k(x, c)\d x\\
        \leq& \|\ell\|_{\infty}\sup_{x \in \mathbb{X}}\|F(x)\|_{\op}\|k\|_{\infty}\\
        \leq& C_{\phi}^s \|\ell\|_{\infty} \|k\|_{\infty}m^{D/2}.
    \end{align*}
    Similarly, we have the variance statistics $\nu \leq n\|\ell\|_{\infty}\|k\|_{\infty}\|\lambda_C\|_{\infty}(r_V + r_W)$.
    By the same arguments in the proof of \Cref{lemma:deviation_op_ortho}, we have that with probability at least $1 - m^{-5}$,
    \begin{align*}
        &\left\| W^{\top} \cdot \left(\mathbb{E}\left[\mathcal{M}_1(\widehat b)\Big|\{\Lambda^{(i)}\}_{i = 1}^n\right] - \mathcal{M}_1(b^*)\right) \cdot V\right\|_{\op}\\
        \leq& C\left\{\sqrt{\frac{\|\ell\|_{\infty}\|k\|_{\infty}\|\lambda_C\|_{\infty}(r_V+r_W)\log(m)}{n}} + \frac{C_{\phi}^s \|\ell\|_{\infty} \|k\|_{\infty}m^{D/2}\log(m)}{n}\right\}.
    \end{align*}
    Consequently, we have that with probability at least $1 - 2m^{-5}$,
    \begin{align*}
        \left\|W^{\top} \cdot \mathcal{M}_1(\widehat b - b^*) \cdot V\right\|_{\op} \leq  C\left\{a_1\sqrt{\frac{(r_V+r_W)\log(m)}{n}} + a_2\frac{m^{D/2}\log(m)}{n}\right\},
    \end{align*}
    where $a_1 = \sqrt{C_{\Lambda}} + \sqrt{\|\ell\|_{\infty}\|k\|_{\infty}\|\lambda_C\|_{\infty}}$ and $a_2 = C_{\phi}^s (1 + \|\ell\|_{\infty} \|k\|_{\infty})$.
\end{proof}

\begin{proof}[Proof of \Cref{thm:main_result_NS}]
    The proof is a consequence of Lemma~\ref{lemma:fundamental_NS} and Propositions~\ref{lemma:hosvd_approx} and \ref{lemma:oshoop_approx}.  It follows the proof of \Cref{thm:main_result}, and thus it is omitted.
\end{proof}

\begin{lemma}[Fundamental bound for the maximal overlapping segments formed by $M$-dependent process]\label{lemma:fundamental_ts}
    Consider the same setting as for \Cref{thm:main_result_ts}. For all deterministic $W_j \in \mathbb O_{m^{d_j}, r_{W_j}}$ and $V_j \in \mathbb O_{m^{D-d_j}, r_{V_j}}$ with ranks $r_{W_j}$ and $r_{V_j}$, respectively, we have that with probability at least $1 - m^{-5}$, uniformly for all $j \in [s]$
    \begin{align*}
        \max_{j = 1}^s\left\|W_j^{\top} \cdot \mathcal{M}_j(\widehat b - b^*) \cdot V_j\right\|_{\op} \leq C\left\{a_1\sqrt{\frac{(r_{W}+r_{V})\log(m)}{n}} + a_2\frac{m^{D/2}\log(m)}{n}\right\},
    \end{align*}
    where $a_1 = C_{\gamma}C_{\mathrm{dep}}\|f^*\|_{\infty}$, $a_2 = C_{\phi}^s (\log(n))^2$, $0 < C_{\mathrm{dep}}, C_{\gamma} < \infty$ and $C > 0$ is an absolute constant.
\end{lemma}
\begin{proof}
    We obtain the upper bound using \Cref{lemma:matrix_bernstein_dep}, and we only focus on the case with $j = 1$. 
    Let $Z = W^{\top} \cdot \mathcal{M}_1(\widehat b
    - b^*) \cdot V \in \mathbb{R}^{r_{W} \times r_{V}}$. We further write $Z = (n+1-D)^{-1}\sum_{i = 1}^{n+1-D} Z^{(i)}$, where $Z^{(i)} = Q^{(i)} - \mathbb{E}(Q^{(i)})$ and the $(j,l)$ entry is 
    \begin{align*}
        Q^{(i)}_{(j; l)} =& \sum_{\mu_1 = 1}^{m^{d_1}}W_{(\mu_1; j)}\phi_{j}(Y_1^{(i)})\sum_{\mu_2 = 1}^{m^{d_2}}\cdots\sum_{\mu_s = 1}^{m^{d_s}}V_{(\mu_2, \dots, \mu_s; l)}\phi_{\mu_2}(Y_2^{(i)})\cdots\phi_{\mu_s}(Y_s^{(i)})\\
        =& \psi_{j}(Y_1^{(i)})\sum_{\mu_2 = 1}^{m^{d_2}}\cdots\sum_{\mu_s = 1}^{m^{d_s}}V_{(\mu_2, \dots, \mu_s; l)}\phi_{\mu_2}(Y_2^{(i)})\cdots\phi_{\mu_d}(Y_s^{(i)}),
    \end{align*}
        where $W_{(\mu_1;j)}$ is the $(\mu_1,j)$ entry of $W$. For each combination of $\mu_2, \dots, \mu_s$ corresponds to a row of $V$ denoted by $V_{(\mu_2, \dots, \mu_s; \cdot)}$.  For $j \in [r_W]$, we let  $\psi_{j}(\cdot) = \sum_{\mu_1 = 1}^{m^{d_1}}W_{(\mu_1;j)}\phi_{j}(\cdot)$. Note that $\{\psi_{j}\}_{j = 1}^{r_W}$ is a set of orthonormal basis, since $\{\phi_{\mu_1}\}_{\mu_1 = 1}^{m^{d_1}}$ is a set of orthonormal basis functions and $\{W_{(\cdot; j)}\}_{j=1}^{r_W}$ is a set of orthonormal vectors.
    We verify the conditions of \Cref{lemma:matrix_bernstein_dep}. We have
    \begin{align*}
        &\|Z^{(i)}\|_{\op} \leq \|Z^{(i)}\|_{\mathrm{F}}\\
        \leq& \sqrt{\sum_{\mu_1=1}^{m^{d_1}}\cdots\sum_{\mu_s = 1}^{m^{d_s}}\left\{\phi_{\mu_1}(Y_1^{(i)})\cdots\phi_{\mu_s}(Y_s^{(i)}) - \mathbb{E}\left[\phi_{\mu_1}(Y_1^{(i)})\cdots\phi_{\mu_s}(Y_s^{(i)})\right]\right\}^2}\\
        & \leq m^{D/2}\left\|\phi_{\mu_1}(Y_1^{(i)})\cdots\phi_{\mu_s}(Y_s^{(i)}) - \mathbb{E}\left[\phi_{\mu_1}(Y_1^{(i)})\cdots\phi_{\mu_s}(Y_s^{(i)})\right]\right\|_{\infty}\\
        & \leq C_{\phi}^sm^{D/2},
    \end{align*}
    where the second inequality follows from $\|W\|_{\op} \leq 1$ and $\|V\|_{\op} \leq 1$, and the last inequality holds because the basis functions satisfy $\|\phi_{j}\|_{\infty} \leq C_{\phi} < \infty$. Recall the matrix variance statistic
    \begin{align*}
    &\nu = (n+1-D)\\
    &\times \sup_{\mathcal{K} \subseteq[n+1-D]}\frac{1}{|\mathcal{K}|}\max\left\{\left\|\mathbb{E}\left[\left(\sum_{i \in \mathcal{K}}Z^{(i)}\right)\left(\sum_{i \in \mathcal{K}}Z^{(i)}\right)^{\top}\right]\right\|_{\mathrm{op}}, \left\|\mathbb{E}\left[\left(\sum_{i \in \mathcal{K}}Z^{(i)}\right)^{\top}\left(\sum_{i \in \mathcal{K}}Z^{(i)}\right)\right]\right\|_{\mathrm{op}}\right\}.
    \end{align*}
    We focus on deriving the bound for $\|\mathbb{E}[(\sum_{i \in \mathcal{K}}Z^{(i)})(\sum_{i \in \mathcal{K}}Z^{(i)})^{\top})]\|_{\op}$, and the bound for the second term can be obtained similarly.  

    Note that $\{Z^{(t)}\}_{t \leq j}$ is a $2D$-dependence process, and thus a $\tau$-mixing process with an exponential coefficient decay rate in time lag $l$, i.e.~$\exp(-\gamma l)$ for some absolute $\gamma > 0$. By Lemma~5.3 in \cite{dedecker2007weak}, we have for each integer $j$, there exists a sequence of random matrices $\{\widetilde{Z}^{(t)}\}_{t > j}$ which is independent of $\sigma(\{Z^{(t)}\}_{t \leq j})$, identically distributed as $\{Z^{(t)}\}_{t > j}$ and for each $k \geq j+1$
    \begin{align}\label{eq:tau-mixing}
        \left\|\mathbb{E}\left[(Z^{(k)} - \widetilde{Z}^{(k)})(Z^{(k)} - \widetilde{Z}^{(k)})^{\top}\right]\right\|_{\op}^{1/2} \leq C_{\mathrm{dep}}\left\|\mathbb{E}\left[Z^{(1)}(Z^{(1)})^{\top}\right]\right\|_{\op}^{1/2}\exp(-\gamma(k-j-1)/2),
    \end{align}
    for some absolute constants $C_{\mathrm{dep}} > 0$.
    \begin{align*}
    \nu =& (n+1-D)\sup_{\mathcal{K} \subseteq [n+1-D]}|\mathcal{K}|^{-1}\left\|\mathbb{E}\left[\left(\sum_{i \in \mathcal{K}}Z^{(i)}\right)\left(\sum_{i \in \mathcal{K}}Z^{(i)}\right)^{\top}\right]\right\|_{\mathrm{op}}\\
    \leq& (n+1-D)\sup_{\mathcal{K} \subseteq [n+1-D]}|\mathcal{K}|^{-1}\sum_{j \in \mathcal{K}}\sum_{k \in \mathcal{K}}\left\|\mathbb{E}\left[Z^{(j)}\left(Z^{(k)}\right)^{\top}\right]\right\|_{\mathrm{op}}\\
    =& (n+1-D)\sup_{\mathcal{K} \subseteq [n+1-D]}|\mathcal{K}|^{-1}\sum_{j \in \mathcal{K}}\sum_{k \in \mathcal{K}}\sup_{\|v\|_2 = 1}v^{\top}\mathbb{E}\left[Z^{(j)}\left(Z^{(k)}\right)^{\top}\right]v\\
    =& (n+1-D)\sup_{\mathcal{K} \subseteq [n+1-D]}|\mathcal{K}|^{-1}\sum_{j \in \mathcal{K}}\sup_{\|v\|_2 = 1}v^{\top}\mathbb{E}\left[Z^{(j)}\left(Z^{(j)}\right)^{\top}\right]v\\
    &+ 2(n+1-D)\sup_{\mathcal{K} \subseteq [n+1-D]}|\mathcal{K}|^{-1}\sum_{j,k \in \mathcal{K}; j < k} \sup_{\|v\|_2 = 1}\left|\mathbb{E}\left[\left(v^{\top}Z^{(j)}\right)\left(v^{\top}Z^{(k)}\right)^{\top}\right]\right|\\
    =& (n+1-D)\sup_{\|v\|_2 = 1}v^{\top}\mathbb{E}\left[Z^{(j)}\left(Z^{(j)}\right)^{\top}\right]v\\
    &+ 2(n+1-D)\sup_{\mathcal{K} \subseteq [n+1-D]}|\mathcal{K}|^{-1}\sum_{j,k \in \mathcal{K}; j < k} \sup_{\|v\|_2 = 1}\left|\mathbb{E}\left[\left(v^{\top}Z^{(j)}\right)\left(v^{\top}(Z^{(k)} - \widetilde{Z}^{(k)})\right)^{\top}\right]\right|\\
    \leq& (n+1-D)\sup_{\|v\|_2 = 1}v^{\top}\mathbb{E}\left[Z^{(j)}\left(Z^{(j)}\right)^{\top}\right]v\\
    &+ 2(n+1-D)\sup_{\mathcal{K} \subseteq [n+1-D]}|\mathcal{K}|^{-1}\\
    &\quad \times \sum_{j,k \in \mathcal{K}; j < k} \sup_{\|v\|_2 = 1}\sqrt{\mathbb{E}\left[\left(v^{\top}Z^{(j)}\right)\left(v^{\top}Z^{(j)}\right)^{\top}\right]\mathbb{E}\left[\left(v^{\top}(Z^{(k)} - \widetilde{Z}^{(k)})\right)\left(v^{\top}(Z^{(k)} - \widetilde{Z}^{(k)})\right)^{\top}\right]}\\
    \leq& (n+1-D)\sup_{\|v\|_2 = 1}v^{\top}\mathbb{E}\left[Z^{(j)}\left(Z^{(j)}\right)^{\top}\right]v\\
    & + 2(n+1-D)C_{\mathrm{dep}}\sup_{\mathcal{K} \subseteq [n+1-D]}|\mathcal{K}|^{-1}\sum_{j,k \in \mathcal{K}; j < k} \sup_{\|v\|_2 = 1}\mathbb{E}\left[\left(v^{\top}Z^{(j)}\right)\left(v^{\top}Z^{(j)}\right)^{\top}\right]\exp(-\gamma(k-j-1))\\
    \leq& (n+1-D)(1+2C_{\gamma}C_{\mathrm{dep}})\left\|\mathbb{E}\left[Z^{(j)}\left(Z^{(j)}\right)^{\top}\right]\right\|_{\mathrm{op}}\\
    \leq& 2(n+1-D)C_{\gamma}C_{\mathrm{dep}}\left\|\mathbb{E}\left[Q^{(j)}\left(Q^{(j)}\right)^{\top}\right]\right\|_{\mathrm{op}},
\end{align*}
    where the first equality follows from the triangle inequality, the third equality follows from the symmetry of the cross-covariances, the fourth equality follows from the stationarity of $\{Z^{(j)}\}$ and the independence between $Z^{(j)}$ and $\widetilde{Z}^{(k)}$, the second inequality follows from the Cauchy-Schwarz inequality for expectations, the third inequality follows from \eqref{eq:tau-mixing}, the fourth inequality follows from the fact $\sum_{j,k \in \mathcal{K}; j<k}\exp(-\gamma(k-j-1)) \leq |\mathcal{K}|\sum_{l \geq 0}\exp(-\gamma l) = |\mathcal{K}| C_{\gamma}$, with $C_{\gamma} < \infty$. The last inequality follows from Fact 8.3.2 of \cite{tropp2015introduction}, i.e.~$\var(Q^{(i)}) \preccurlyeq \mathbb{E}[Q^{(i)}(Q^{(i)})^{\top}]$. Note that 
    \begin{align*}
        &\left[Q^{(j)}\left(Q^{(j)}\right)^{\top}\right]_{p, q}\\
        =& \sum_{l = 1}^{r_V} Q_{(p;l)}^{(j)}Q_{(q;l)}^{(j)}
        = \psi_{p}(Y_1^{(j)})\psi_{q}(Y_1^{(j)})\sum_{l = 1}^{r_V}\left(\sum_{\mu_2 = 1}^{m^{d_2}}\cdots\sum_{\mu_s = 1}^{m^{d_s}} V_{(\mu_2, \dots, \mu_s; l)}\phi_{\mu_2}(Y_2^{(j)})\cdots\phi_{\mu_s}(Y_s^{(j)})\right)^2.
    \end{align*}
    Furthermore,
    \begin{align*}
        &\left\|\mathbb{E}[Q^{(j)}(Q^{(j)})^{\top}]\right\|_{\op} = \sup_{\|v\|_2 = 1} v^{\top}\mathbb{E}[Q^{(j)}(Q^{(j)})^{\top}]v = \sup_{\|v\|_2 = 1} \mathbb{E}\left(\sum_{p = 1}^{r_W}\sum_{q = 1}^{r_W} v_p\left[Q^{(j)}(Q^{(j)})^{\top}\right]_{p; q}v_q\right) \\
        =& \sup_{\|v\|_2 = 1} \idotsint\left(\sum_{p = 1}^{r_W}\sum_{q = 1}^{r_W} v_p\psi_{p}(x_1)\psi_{q}(x_1)v_q\right)\\
        &\quad\quad\quad\quad\quad\quad\quad\left\{\sum_{l = 1}^{r_V}\left(\sum_{\mu_2 = 1}^{m^{d_2}}\cdots\sum_{\mu_s = 1}^{m^{d_s}} V_{(\mu_1, \dots,\mu_s; l)}\phi_{\mu_2}(x_2)\cdots\phi_{\mu_s}(x_s)\right)^2\right\}f^*(x_1,\cdots,x_s)\d x_2 \cdots \d x_s\\
        \leq& \|f^*\|_{\infty}\sup_{\|v\|_2 = 1} \int\left(\sum_{k = 1}^{r_W} v_k\psi_{k}(x_1)\right)^2\d x_1\\
        &\times \left\{\sum_{l = 1}^{r_V}\idotsint\left(\sum_{\mu_2 = 1}^{m^{d_2}}\cdots\sum_{\mu_s = 1}^{m^{d_s}}V_{(\mu_2, \dots,\mu_s; l)}\phi_{\mu_2}(x_2)\cdots\phi_{\mu_s}(x_s)\right)^2\d x_2\cdots \d x_s\right\}\\
        =& \|f^*\|_{\infty}\sup_{\|v\|_2 = 1} \int\sum_{k = 1}^{r_W} v_k^2\psi^2_{k}(x_1) \d x_1\\
        &\times \left\{\sum_{l = 1}^{r_V}\idotsint\sum_{\mu_2 = 1}^{m^{d_2}}\cdots\sum_{\mu_s = 1}^{m^{d_s}}\left\{V_{(\mu_2,\dots,\mu_s; l)}\phi_{\mu_2}(x_2)\cdots\phi_{\mu_s}(x_s)\right\}^2\d x_2\cdots \d x_s\right\}\\
        =& \|f^*\|_{\infty}r_V,
    \end{align*}
    where the last two lines follows from the fact that $\{\psi_j\}$ and $\{\phi_j\}$ are othonormal basis and $V \in \mathbb O_{m^{D-d_1}, r_V}$. Similarly, we can show that $\|\mathbb{E}[(Q^{(i)})^{\top}Q^{(i)}]\|_{\op} \leq \|f^*\|_{\infty}r_W$.
    Therefore, we have $\nu \leq 2(n+1-D)C_{\gamma}C_{\mathrm{dep}}\|f^*\|_{\infty}(r_V + r_W)$. By \Cref{lemma:matrix_bernstein_dep}, we have
\begin{align*}
    &\mathbb{P}\left(\left\|W^{\top} \cdot \mathcal{M}_1(\widehat b - b^*) \cdot V\right\|_{\op} \geq t\right)\\
    \leq& (r_V+ r_W)\exp\left(-\frac{C_1(n+1-D)t^2}{2C_{\gamma}C_{\mathrm{dep}}\|f^*\|_{\infty}(r_V + r_W) }\right)\\
    &+ (r_V+ r_W)\exp\left(-\frac{C_2(n+1-M)^2t^2}{c^{-1}C_{\phi}^{2s}m^{D}}\right)\\
    &\leq (r_V+ r_W)\exp\left(-\frac{C_3(n+1-D)t}{C_{\phi}^sm^{D/2}(\log(n+1-D))^2}\right). 
\end{align*}
    It follows that with probability at least $1 - m^{-5}$,
    \begin{align*}
        \left\|W^{\top} \cdot \mathcal{M}_1(\widehat b - b^*) \cdot V\right\|_{\op} \leq C\left\{\sqrt{\frac{C_{\gamma}C_{\mathrm{dep}}\|f^*\|_{\infty}(r_V+r_W)\log(m)}{n}} + \frac{C_{\phi}^{s}m^{D/2}\log(m)(\log(n))^2}{n}\right\}.
    \end{align*}
    The same argument leads to the similar bounds for $j = 2, \dots, s$, which concludes the proof.
\end{proof}

\begin{proof}[Proof of \Cref{thm:main_result_ts}]
    The proof is a consequence of Lemma~\ref{lemma:fundamental_ts} and Propositions~\ref{lemma:hosvd_approx} and \ref{lemma:oshoop_approx}.  It follows from the proof of \Cref{thm:main_result}, and thus it is omitted.
\end{proof}

\section{Auxiliary results for tensor estimation under approximately low-rank settings}\label{sec:approx_low_rank_tensor}
\subsection{Notation}\label{sec:notation_app}
We recall the notation used for our main results.  Let $\mathbb{X} = \mathbb{X}_1 \times \cdots \times \mathbb{X}_s \subset \mathbb{R}^{d_1}\times  \cdots \times \mathbb{R}^{d_s}= \mathbb R^{D}$. Let $\{N^{(i)}\}_{i = 1}^n \subseteq \mathbb{X}$ be a set of i.i.d.~spatial point processes, with intensity function $\lambda^*: \mathbb{X} \to \mathbb{R}_+$.  Suppose $\lambda^{*}$ satisfies that $\lambda^* \in \mathbb{L}_2(\mathbb{X})$ and $\|\lambda^*\|_\infty < \infty$.

Let $u_j \in \mathbb L_2(\mathbb{X}_j)$ for each $j \in [s]$. For a function $A: \mathbb{X} \to \mathbb{R}$, define the operator norm of $A$ as
\[
    \|A\|_{\op} = \sup_{\|u_j\|_{\mathbb L_2(\mathbb{X}_j)}\leq 1; \, j \in [s]} A[u_{1}, \dots,u_{s}].
\]
Let $\widehat{\lambda}$ be the empirical version of $\lambda^*$ based on $\{N^{(i)}\}_{i = 1}^n$. Let $\{\phi_{j,\mu_j}\}_{\mu_j = 1}^{\infty} \subset \mathbb L_2(\mathbb{X}_j)$ be a collection of orthonormal basis functions satisfying $\|\phi_{j, \mu_j}\|_{\infty} \leq C_{\phi} < \infty$. For $\widehat \lambda$ and $\lambda^*$, the associated coefficients tensors are
    \[\widehat b = \{\widehat b_{\mu_1,\dots,\mu_s}\}_{\mu_1,\dots,\mu_s = 1}^{m^{d_1}, \dots, m^{d_s}} = \{\widehat \lambda[\phi_{1,\mu_1}, \cdots, \phi_{s,\mu_s}]\}_{\mu_1,\dots,\mu_s = 1}^{m^{d_1}, \dots, m^{d_s}},\]
    \[b^* = \{b^*_{\mu_1,\dots,\mu_s}\}_{\mu_1,\dots,\mu_s= 1}^{m^{d_1}, \dots, m^{d_s}} = \{\lambda^*[\phi_{1,\mu_1}, \cdots, \phi_{s,\mu_s}]\}_{\mu_1,\dots,\mu_s = 1}^{m^{d_1}, \dots, m^{d_s}}.\]
Let $\mathcal{M}_j(b)$ be the mode-$j$ matricization of an $s$th-order tensor $b$.

Given elements $\{f_i\}_{i = 1}^n$ in a Hilbert space $\mathcal{H}$, define the span as
\[
\mathrm{Span}\{f_{i}: i \in [n]\} = \{b_1f_1 + \cdots + b_nf_n: \{b_i\}_{i = 1}^n \subset \mathbb{R}\}
\]
Let $\{\phi_i\}_{i = 1}^{\infty}$ be a set of orthonormal basis functions in $\mathcal{H}$.
The linear subspace $\mathcal{U} = \mathrm{Span}\{\phi_{i}: i \in [m]\}$ has dimension $m$. The projection operator onto $\mathcal{U}$ is defined for all $f \in \mathcal{H}$ by
\[
    \mathcal{P}_{\mathcal{U}}f = \sum_{i = 1}^m \left\{\int f(x)\phi_{i}(x)\d x\right\} \phi_{i}.
\]

\subsection{Technical tools for approximately low-rank matrices/tensors}\label{sec:approx_low_rank_tools}

\begin{theorem}[Weyl's inequality for singular values]
\label{thm:WeylSV}
Let $A,B\in\mathbb{R}^{m\times n}$ and denote their singular values (in nonincreasing order) by $\{\sigma_i(A)\}$ and $\{\sigma_i(B)\}$  respectively. In addition denote the singular values of $A+B$ as $\{\sigma_i(A+B)\}$. Then for all indices $i,j$ satisfying $i+j-1\le \min\{m,n\}$,
\[
\sigma_{i+j-1}(A+B) \le \sigma_i(A) + \sigma_j(B).
\]
\end{theorem}

\begin{theorem}[Wedin's sin$\Theta$ theorem, Theorem~2.9 of \cite{chen2021spectral}]\label{thm:wedin}
Let $M = M^{*} + E \in \mathbb R^{n_1 \times n_2}$ (without loss of generality assume that $n_1 \leq n_2$). The SVD of $M^*$ and $M$ are given respectively by
\[
M^* = \sum_{i=1}^{n_1} \sigma_i^* u_i^* {v_i^*}^\top \in \mathbb R^{n_1 \times n_2} \quad \text{and} \quad M = \sum_{i=1}^{n_1} \sigma_i u_i v_i^\top \in \mathbb R^{n_1 \times n_2},
\]
where $\sigma_1^* \geq \cdots \geq \sigma_{n_1}^*$ and $\sigma_1 \geq \cdots \geq \sigma_{n_1}$. For all $R \leq n_1$, let
\[
\Sigma^* = \mathrm{diag}([\sigma_1^*, \cdots, \sigma_R^*]) \in \mathbb{R}^{R \times R}, \quad U^* = [u_1^*, \cdots, u_R^*] \in \mathbb{R}^{n_1 \times R}, \quad V^* = [v_1^*, \cdots, v_R^*] \in \mathbb{R}^{R \times n_2},
\]
\[
\Sigma = \mathrm{diag}([\sigma_1,  \cdots, \sigma_R]) \in \mathbb{R}^{R \times R}, \quad U = [u_1, \cdots, u_R] \in \mathbb{R}^{n_1 \times R}, \quad V = [v_1, \cdots, v_R] \in \mathbb{R}^{R \times n_2}.
\]
If $\|\mathbb{E}\|_{\op} < \sigma_R^* - \sigma_{R+1}^*$, then we have
\[
\max\left\{\left\| \sin \Theta(U, U^*)\right\|_{\op}, \left\| \sin \Theta(V, V^*)\right\|_{\op} \right\} 
\leq \frac{ \|E \|_{\op}}{\sigma_R^* - \sigma_{R+1}^* - \|E\|_{\op}}.
\]
\end{theorem}

\begin{lemma}[Properties of the sin$\Theta$ distances, Lemma~1 of \cite{cai2018rate}]\label{lemma:sintheta}
\
\\
The following properties hold for the $\sin\Theta$ distances.
\begin{enumerate}
    \item (Equivalent Expressions) Suppose $V, \widehat{V} \in \mathbb{O}_{p, R}$. If $V_{\perp}$ is an orthogonal extension of $V$, namely $\begin{bmatrix} V & V_{\perp}\end{bmatrix} \in \mathbb{O}_{p}$, we have the following equivalent forms for $\| \sin \Theta(\widehat{V}, V) \|_{\op}$ and $\| \sin \Theta(\widehat{V}, V) \|_{\mathrm{F}}$,
    \[
    \| \sin \Theta(\widehat{V}, V) \|_{\op} = \sqrt{1 - \sigma_{\min}^2(\widehat{V}^T V)} = \| \widehat{V}^T V_{\perp} \|_{\op},
    \]
    \[
    \| \sin \Theta(\widehat{V}, V) \|_{\mathrm{F}} = \sqrt{r - \| V^T \widehat{V} \|_{\mathrm{F}}^2} = \| \widehat{V}^T V_{\perp} \|_{\mathrm{F}}.
    \]
    \item (Triangle Inequality) For all $V_1, V_2, V_3 \in \mathbb{O}_{p, R}$,
    \[
    \| \sin \Theta(V_2, V_3) \|_{\op} \leq \| \sin \Theta(V_1, V_2) \|_{\op} + \| \sin \Theta(V_1, V_3) \|_{\op},
    \]
    \[
    \| \sin \Theta(V_2, V_3) \|_{\mathrm{F}} \leq \| \sin \Theta(V_1, V_2) \|_{\mathrm{F}} + \| \sin \Theta(V_1, V_3) \|_{\mathrm{F}}.
    \]
    \item (Equivalence with Other Metrics)
\[
\| \sin \Theta(\widehat{V}, V) \|_{\op} \leq \sqrt{2} \| \sin \Theta(\widehat{V}, V) \|_{\op},
\]
\[
\| \sin \Theta(\widehat{V}, V) \|_{\mathrm{F}} \leq \sqrt{2} \| \sin \Theta(\widehat{V}, V) \|_{\mathrm{F}},
\]
\[
\| \sin \Theta(\widehat{V}, V) \|_{\op} \leq \| \widehat{V} \widehat{V}^\top - V V^\top \|_{\op} \leq 2 \| \sin \Theta(\widehat{V}, V) \|_{\op},
\]
\[
\| \widehat{V} \widehat{V}^\top - V V^\top \|_{\mathrm{F}} = \sqrt{2} \| \sin \Theta(\widehat{V}, V) \|_{\mathrm{F}}.
\]

\end{enumerate}
    
\end{lemma}

\begin{lemma}\label{lemma:perterbation_approx}
    Suppose $X, Z \in \mathbb{R}^{n \times m}$. For all $1 \leq R \leq \min\{n,m\}$, write the full SVD of $Y$ as
\[
Y = X + Z = \widehat{U} \widehat{\Sigma} \widehat{V}^\top = \begin{bmatrix} \widehat{U}_{(R)} \ \widehat{U}_{\perp} \end{bmatrix} \cdot \begin{bmatrix} \widehat{\Sigma}_{(R)} & \\ & \widehat{\Sigma}_{\perp} \end{bmatrix} \cdot \begin{bmatrix} \widehat{V}_{(R)}^\top \\ \widehat{V}_{\perp}^\top \end{bmatrix},
\]
where $\widehat{U}_{(R)} \in \mathbb{O}_{n, R}$, $\widehat{V}_{(R)} \in \mathbb{O}_{m, R}$ correspond to the leading $R$ left and right singular vectors; and $\widehat{U}_{\perp} \in \mathbb{O}_{n, n - R}$, $\widehat{V}_{\perp} \in \mathbb{O}_{m, m - R}$ correspond to their orthonormal complement. We have
    \begin{align*}
    \left\|\mathcal P_{\widehat{U}_{\perp}}X\right\|_{\mathrm{F}} \leq& 3\sqrt{\sum_{j = R+1}^{\min\{n,m\}}\sigma_j^2(X)} + 2\min\left\{\sqrt{R}\|Z\|_{\op}, \|Z\|_{\mathrm{F}}\right\}\\
    =& 3\left\|X_{(R)} - X\right\|_{\mathrm{F}} + 2\min\left\{\sqrt{R}\|Z\|_{\op}, \|Z\|_{\mathrm{F}}\right\}.
    \end{align*}
\end{lemma}
\begin{proof}
    Without loss of generality, assume $n \leq m$.
    For $A \in \mathbb{R}^{n\times m}$, let $\Sigma(A) \in \mathbb{R}^{n\times m}$ denote the non-negative diagonal matrices whose diagonal entries are the non-increasingly ordered singular values of $A$.

    For all $1 \leq R \leq n$, let $X_{(R)}$ denote the truncated SVD of $X$ with rank $R$, and we have 
    \[
    \left\|X_{(R)} - X\right\|_{\mathrm{F}} = \sqrt{\sum_{j = R+1}^{n}\sigma_{j}^2(X)}.
    \]
    We have
        \begin{align*}
        &\left\|\mathcal P_{\widehat{U}_{\perp}}X\right\|_{\mathrm{F}} \leq \left\|\mathcal P_{\widehat{U}_{\perp}}X_{(R)}\right\|_{\mathrm{F}} + \left\|\mathcal P_{\widehat{U}_{\perp}}(X-X_{(R)})\right\|_{\mathrm{F}} = \sqrt{\sum_{j = 1}^R\sigma_j^2(\mathcal P_{\widehat{U}_{\perp}}X_{(R)})} + \left\|\mathcal P_{\widehat{U}_{\perp}}(X-X_{(R)})\right\|_{\mathrm{F}}\\
        \leq& \sqrt{\sum_{j = 1}^R\sigma_j^2(\mathcal P_{\widehat{U}_{\perp}}X_{(R)})} + \left\|X-X_{(R)}\right\|_{\mathrm{F}} = \sqrt{\sum_{j = 1}^R\sigma_j^2(\mathcal P_{\widehat{U}_{\perp}}X_{(R)})} + \sqrt{\sum_{j = R+1}^{n}\sigma_{j}^2(X)}\\
        \leq& \left\|(\sigma_{1}(\mathcal P_{\widehat{U}_{\perp}}X_{(R)}) - \sigma_{1}(\mathcal P_{\widehat{U}_{\perp}}X), \dots, \sigma_{R}(\mathcal P_{\widehat{U}_{\perp}}X_{(R)}) - \sigma_{R}(\mathcal P_{\widehat{U}_{\perp}}X))^{\top}\right\|_2\\
        &+ \left\|(\sigma_{1}(\mathcal P_{\widehat{U}_{\perp}}X), \dots, \sigma_{R}(\mathcal P_{\widehat{U}_{\perp}}X))^{\top}\right\|_2
        + \sqrt{\sum_{j = R+1}^{n}\sigma_{j}^2(X)}\\
        \leq& \left\|\Sigma(\mathcal P_{\widehat{U}_{\perp}}X_{(R)}) - \Sigma(\mathcal P_{\widehat{U}_{\perp}}X)\right\|_{\mathrm{F}} + \left\|(\sigma_{1}(\mathcal P_{\widehat{U}_{\perp}}X), \dots, \sigma_{R}(\mathcal P_{\widehat{U}_{\perp}}X))^{\top}\right\|_2 + \sqrt{\sum_{j = R+1}^{n}\sigma_{j}^2(X)}\\
        \leq& \left\|\mathcal P_{\widehat{U}_{\perp}}(X_{(R)} - X)\right\|_{\mathrm{F}} + \sqrt{\sum_{j = 1}^R\sigma_{j}^2(\mathcal P_{\widehat{U}_{\perp}}X)} + \sqrt{\sum_{j = R+1}^{n}\sigma_{j}^2(X)}\\
        \leq& \sqrt{\sum_{j = 1}^R\sigma_{j}^2(\mathcal P_{\widehat{U}_{\perp}}X)} + 2\sqrt{\sum_{j = R+1}^{n}\sigma_{j}^2(X)},
    \end{align*}
    where the first equality follows from $\rank(X_{(R)}) = R$, and the fifth inequality follows from \Cref{thm:Mirsky}.
    To upper bound $\sqrt{\sum_{j = 1}^{R}\sigma_j^2(\mathcal P_{\widehat{U}_{\perp}}X)}$, we first consider $\sqrt{\sum_{j = 1}^{R}\sigma_j^2(\mathcal P_{\widehat{U}_{\perp}}Y)}$.
    Note that
    \[
        \mathcal P_{\widehat{U}_{\perp}}Y = \sum_{j = R+1}^{n}\sigma_j(Y)\widehat u_j \widehat v_j^{\top},
    \]
    where $\widehat u_j$ and $\widehat v_j$ are the left and right singular vector associated with the $j$th largest singular value $\sigma_j(Y)$. Let $\sigma_j(Y) = \sigma_j(X) = 0$ for $j > p_1$.
    It follows that
    \begin{align}\label{eq:proof_singularvalue_perp}
        &\sqrt{\sum_{j = 1}^{R}\sigma_j^2(\mathcal P_{\widehat{U}_{\perp}}Y)} = \sqrt{\sum_{j = R+1}^{2R}\sigma_j^2(Y)} = \left\|(\sigma_{R+1}(Y), \dots, \sigma_{2R}(Y))^{\top}\right\|_2\nonumber\\
        \leq& \left\|(\sigma_{R+1}(Y) - \sigma_{R+1}(X), \dots, \sigma_{2R}(Y) - \sigma_{2R}(X))^{\top}\right\|_2 + \left\|(\sigma_{R+1}(X), \dots, \sigma_{2R}(X))^{\top}\right\|_2\nonumber\\
        \leq& \min\left\{\sqrt{R}\|Z\|_{\op}, \|Z\|_{\mathrm{F}}\right\} + \sqrt{\sum_{j = R+1}^{n}\sigma_j^2(X)},
    \end{align}
    where the first inequality follows from the triangle inequality, and second inequality follows from \Cref{thm:WeylSV}, i.e.~$|\sigma_j(Y) - \sigma_j(X)| \leq \|Y - X\|_{\op}$ for all $1 \leq j \leq n$, as well as the fact that
    \[
    \left\|(\sigma_{R+1}(Y) - \sigma_{R+1}(X), \dots, \sigma_{2R}(Y) - \sigma_{2R}(X))^{\top}\right\|_2 \leq \left\|\Sigma(Y) - \Sigma(X)\right\|_{\mathrm{F}} \leq \left\|Z\right\|_{\mathrm{F}},
    \]
    where the last inequality follows from \Cref{thm:Mirsky}.
    It then follows from \eqref{eq:proof_singularvalue_perp},
    \begin{align*}
        &\sqrt{\sum_{j = 1}^{R}\sigma_j^2(\mathcal P_{\widehat{U}_{\perp}}X)} = \left\|(\sigma_{1}(\mathcal P_{\widehat{U}_{\perp}}(Y-Z)), \dots, \sigma_{R}(\mathcal P_{\widehat{U}_{\perp}}(Y-Z))^{\top}\right\|_2\\
        \leq& \left\|(\sigma_{1}(\mathcal P_{\widehat{U}_\perp}(Y-Z)) - \sigma_{1}(\mathcal P_{\widehat{U}_\perp}Y), \dots, \sigma_{R}(\mathcal P_{\widehat{U}_{\perp}}(Y-Z)) - \sigma_{R}(\mathcal P_{\widehat{U}_{\perp}}Y))^{\top}\right\|_2\\& + \left\|(\sigma_{1}(\mathcal P_{\widehat{U}_{\perp}}Y), \dots, \sigma_{R}(\mathcal P_{\widehat{U}_\perp}Y))^{\top}\right\|_2\\ \leq& \min\left\{\sqrt{R}\|\mathcal P_{\widehat{U}_\perp}Z\|_{\op}, \|\mathcal P_{\widehat{U}_\perp}Z\|_{\mathrm{F}} \right\} + \sqrt{\sum_{j = 1}^{R}\sigma_j^2(\mathcal P_{\widehat{U}_\perp}Y)}\\
         \leq& \min\left\{\sqrt{R}\|Z\|_{\op}, \|Z\|_{\mathrm{F}}\right\} + \sqrt{\sum_{j = 1}^{R}\sigma_j^2(P_{\widehat{U}_\perp}Y)}\\
          \leq& 2\min\left\{\sqrt{R}\|Z\|_{\op}, \|Z\|_{\mathrm{F}}\right\} + \sqrt{\sum_{j = R+1}^{n}\sigma_j^2(X)},
    \end{align*}
    where the first two inequalities follow from the same arguments as in \eqref{eq:proof_singularvalue_perp}.
    Consequently,
    \[
    \left\|\mathcal P_{\widehat{U}_\perp}X\right\|_{\mathrm{F}} \leq 3\sqrt{\sum_{j = R+1}^{n}\sigma_j^2(X)} + 2\min\left\{\sqrt{R}\|Z\|_{\op}, \|Z\|_{\mathrm{F}}\right\}.
    \]
\end{proof}

\begin{lemma}\label{lemma:singular values preserve}
Let $A \in \mathbb{R}^{p \times q}$ and let $U \in \mathbb{O}^{q \times r}$ (i.e.~$U^\top U = I_r$). Then
\[
\sigma_i(AU) \le \sigma_i(A) \quad \text{for all } i.
\]
\end{lemma}
\begin{proof}
Let $M = A^\top A \in \mathbb{R}^{q\times q}$. Then $M$ is symmetric positive semidefinite, and
\[
(AU)^\top (AU) = U^\top A^\top A U = U^\top M U.
\]
Therefore, for each admissible index $i$,
\[
\sigma_i^2(A) = \lambda_i(M),
\qquad
\sigma_i^2(AU) = \lambda_i(U^\top M U),
\]
where $\lambda_i(\cdot)$ denotes the $i$-th largest eigenvalue.

We now compare $\lambda_i(U^\top M U)$ and $\lambda_i(M)$ via the Courant--Fischer variational characterization (Min-max theorem):
For any symmetric matrix $B$, we have
\[
\lambda_i(B)
=
\max_{\substack{S \subset \mathbb{R}^{\dim(B)}\\ \dim S = i}}
\ \min_{\substack{x \in S\\ \|x\|_2=1}} x^\top B x.
\]
Applying this with $B = U^\top M U \in \mathbb{R}^{r\times r}$ yields
\[
\lambda_i(U^\top M U)
=
\max_{\substack{S \subset \mathbb{R}^{r}\\ \dim S = i}}
\ \min_{\substack{x \in S\\ \|x\|_2=1}} x^\top U^\top M U x
=
\max_{\substack{S \subset \mathbb{R}^{r}\\ \dim S = i}}
\ \min_{\substack{x \in S\\ \|x\|_2=1}} (Ux)^\top M (Ux).
\]
Since $U^\top U = I_r$, the map $x \mapsto Ux$ is an isometry, so $\|Ux\|_2=\|x\|_2$, and it is injective, so
$\dim(US)=\dim(S)=i$. Let $\mathcal{R}:=\mathrm{range}(U)\subset \mathbb{R}^q$. Then we can rewrite the previous
expression as
\[
\lambda_i(U^\top M U)
=
\max_{\substack{T \subset \mathcal{R}\\ \dim T = i}}
\ \min_{\substack{y \in T\\ \|y\|_2=1}} y^\top M y.
\]
On the other hand, Courant--Fischer also gives
\[
\lambda_i(M)
=
\max_{\substack{T \subset \mathbb{R}^{q}\\ \dim T = i}}
\ \min_{\substack{y \in T\\ \|y\|_2=1}} y^\top M y.
\]
Because the collection of $i$-dimensional subspaces of $\mathcal{R}$ is a subset of the collection of all
$i$-dimensional subspaces of $\mathbb{R}^q$, restricting the maximization can only decrease the value. Hence,
\[
\lambda_i(U^\top M U) \le \lambda_i(M).
\]
Taking square roots (all eigenvalues involved are nonnegative) yields
\[
\sigma_i(AU) = \sqrt{\lambda_i(U^\top M U)} \le \sqrt{\lambda_i(M)} = \sigma_i(A),
\]
which proves the claim.
\end{proof}

\begin{lemma} \label{lemma:Fobenius norm preserve}
    Let $A\in \mathbb R^{p\times q}$ and $U \in \mathbb O^{q\times r}$. Then 
    $$ \|AUU^\top  \|_{\mathrm{F}} = \|AU\|_{\mathrm{F}}. $$
\end{lemma}
\begin{proof}Observe that 
    \begin{align*}
        \|AUU^\top  \|_{\mathrm{F}} ^2 = \tr ( AUU^\top UU^\top A^\top ) = \tr ( AU U^\top A^\top ) = \|AU\|_{\mathrm{F}}^2.
    \end{align*}
\end{proof}
\begin{lemma}\label{lemma:lb_forbuious}
    For any   real matrices $A \in \mathbb R^{n\times m }$ and $B\in \mathbb R^{m \times m } $, it holds that 
    \[
    \|AB\|_{\mathrm{F}} \geq \sigma_{\min}(B) \|A\|_{\mathrm{F}} .
    \]
\end{lemma}
\begin{proof}[Proof of \Cref{lemma:lb_forbuious}]
Since $B $ is a square matrix, it follows that 
$$ \lambda_{\min }(BB^\top ) = \sigma_{\min}^2(B),$$
where $\lambda_{\min}(\cdot)$ denotes the minimum eigenvalue.
Note that 
$$ BB^\top  \succeq \lambda _{\min}(BB^\top ) I_m .$$
Therefore 
$$ABB^{\top}A^{\top} \succeq  A \{ \lambda _{\min}(BB^\top ) I_m\}  A ^{\top}, $$
and so 
$$\tr(ABB^{\top}A^{\top})     \geq \tr(A  \{ \lambda _{\min}(BB^\top ) I_m\}  A^{\top}) .$$
Then 
    \begin{align*}
        \|AB\|_{\mathrm{F}}^2 = \tr(ABB^{\top}A^{\top})     \geq \tr(A \{ \lambda _{\min}(BB^\top ) I_m\}  A^{\top})  = \lambda _{\min}(BB^\top ) \tr(A   A^{\top})    =   \sigma_{\min}^2(B) \|A\|^2_{\mathrm{F}}  .
    \end{align*}
\end{proof}

\begin{lemma}
\label{lemma:complement_d}
Let $X \in\mathbb R^{p_1\times\cdots\times p_d}$ and for each $j\in [d]$ let
$U_j \in\mathbb O_{p_j,r_j}$ with orthogonal complement $U_{j\,\perp} \in\mathbb O_{p_j,p_j-r_j}$.
For each $k \in [d]$, define $W_{-k} = \otimes_{j \neq k} U_j$ and let $W_{-k\,\perp}$ be an orthonormal basis of $\mathrm{span}(W_{-k})^\perp$.
Then, for each $k$, we have
\[
\big\|\mathcal M_k(X)\cdot W_{-k\,\perp}\big\|_{\mathrm F}
\le
\sum_{j\neq k} \big\|X \times_j U_{j\,\perp}\big\|_{\mathrm F}.
\]
\end{lemma}

\begin{proof}
Let $\mathcal P_j = U_j U_j^{\top}$ and $\mathcal Q_j = I-\mathcal P_j=U_{j\,\perp}U_{j\,\perp}^{\top}$ for $j\ge2$. We only consider the case with $k = 1$ in the following proof.
Denote
\[
W = U_2 \otimes\cdots\otimes U_d \in\mathbb O_{\prod_{j=2}^d p_j,\ \prod_{j=2}^d r_j},
\qquad
\Pi =I-\otimes_{j=2}^d \mathcal P_j,
\]
so that $\Pi$ is the orthogonal projector onto $\mathrm{span}(W)^\perp$ and
\[
\|\mathcal M_1(X) \cdot W_\perp\|_{\mathrm F}=\|\mathcal M_1(X) \cdot \Pi\|_{\mathrm F}.
\]
Using the identity $I-AB=(I-A)+A(I-B)$ iteratively with
$A=\mathcal P_2$ and $B=\otimes_{j=3}^d\mathcal P_j$, we obtain the exact decomposition
\begin{equation}
\label{eq:group_decomp}
I-\otimes_{j=2}^d\mathcal P_j
=
\sum_{j=2}^d
\left(\otimes_{\ell=2}^{j-1}\mathcal P_\ell\right)\otimes \mathcal Q_j \otimes I_{j+1:d},
\end{equation}
where $I_{j+1:d} =\otimes_{\ell=j+1}^d I_{p_\ell}$ (and by convention $\otimes_{\ell=2}^{1}\mathcal P_\ell$ is the scalar $1$).
Each summand in \eqref{eq:group_decomp} is an orthogonal projector onto the subspace
\[
\mathcal S_j
=
\left(\otimes_{\ell=2}^{j-1}\mathrm{span}(U_\ell)\right)\otimes \mathrm{span}(U_{j\,\perp})\otimes
\left(\otimes_{\ell=j+1}^d \mathbb R^{p_\ell}\right),
\]
and the subspaces $\{\mathcal S_j\}_{j=2}^d$ are mutually orthogonal. Hence \eqref{eq:group_decomp}
expresses $\mathrm{span}(W)^\perp$ as an orthogonal direct sum of $(d-1)$ groups.

By the triangle inequality and \eqref{eq:group_decomp},
\[
\|\mathcal M_1(X) \cdot \Pi\|_{\mathrm F}
\le
\sum_{j=2}^d
\left\|
\mathcal M_1(X)
\left(\otimes_{\ell=2}^{j-1}\mathcal P_\ell\right)\otimes \mathcal Q_j \otimes I_{j+1:d}
\right\|_{\mathrm F}.
\]
We can rewrite each summand as
\[
\left\|
\mathcal M_1(X)
\left(\otimes_{\ell=2}^{j-1}\mathcal P_\ell\right)\otimes \mathcal Q_j \otimes I_{j+1:d}
\right\|_{\mathrm F}
=
\left\|
X\times_2 \mathcal P_2 \times_3\cdots\times_{j-1}\mathcal P_{j-1}\times_j \mathcal Q_j
\right\|_{\mathrm F}.
\]
Since 
$\|T\times_\ell \mathcal P_\ell\|_{\mathrm F}\le \|T\|_{\mathrm F}$ for all $\ell$,
it follows that
\[
\left\|
X\times_2 \mathcal P_2 \times_3\cdots\times_{j-1}\mathcal P_{j-1}\times_j \mathcal Q_j
\right\|_{\mathrm F}
\le
\|X\times_j \mathcal Q_j\|_{\mathrm F}
=
\|X\times_j U_{j\,\perp}\|_{\mathrm F}.
\]
Combining the previous displays yields
\[
\|\mathcal M_1(X)\cdot W_\perp\|_{\mathrm F}
=
\|\mathcal M_1(X) \cdot \Pi\|_{\mathrm F}
\le
\sum_{j=2}^d \|X\times_j U_{j\,\perp}\|_{\mathrm F}.
\]
\end{proof}

\begin{theorem}[Matrix Bernstein's inequality, Corollary~3.3 in \cite{chen2021spectral}]\label{thm:MatrixBernstein} 
Let $\{X_{i}\}_{i = 1}^n$ be a set of independent real random matrices with dimension $d_1\times d_2$. Suppose that $\mathbb{E}(X_{i}) = 0$ and $\|X_i\|_{\op} \leq L$ almost surely, for all $i$. Define the variance statistics as 
\begin{align*}
    \nu = n\max\left\{\|\mathbb{E}(X_iX_i^{\top})\|_{\op}, \|\mathbb{E}(X_i^{\top}X_i)\|_{\op}\right\}.
\end{align*}
We have for all $t \geq 0$,
\begin{align*}
    \mathbb P\left(\left\|\sum_{i = 1}^n X_{i}\right\|_{\op} \geq t \right) \leq (d_1 + d_2)\exp\left(-\frac{t^2/2}{\nu + Lt/3}\right).
\end{align*}
In particular, for all $a \geq 2$, with probability at least $1 - 2(\max\{d_1,d_2\})^{1-a}$, we have 
\begin{align*}
    \left\|\sum_{i = 1}^n X_{i}\right\|_{\op} \leq \sqrt{2a\nu\log(d_1 + d_2)} + \frac{2a}{3}L\log(d_1 + d_2).
\end{align*}
\end{theorem}

\begin{theorem}[Modified Theorem 1 in \cite{banna2016bernstein}]\label{lemma:matrix_bernstein_dep}
Let $\{M_i\}_{i = 1}^n$ be a sequence of random matrices of size $d_1 \times d_2$. Assume that there exists a constant $c > 0$ such that for all $\ell \geq 1$, $\beta_M(\ell) \leq \exp(1-c\ell)$, and there exist a positive constant $L$ such that for all $i$,
\begin{align*}
    \mathbb{E}(M_i) = 0 \quad \text{and} \quad \|M_i\|_{\mathrm{op}} \leq L \quad \text{almost surely}.
\end{align*}
We have that there exists an absolute constant $C$ such that for all $t > 0$ and all integers $n \geq 2$,
\[
\mathbb{P}\left(\left\|\sum_{i=1}^nM_i\right\|_{\mathrm{op}} \geq t\right) \leq (d_1 + d_2)\exp\left(-\frac{Ct^2}{\nu + c^{-1}L^2 + tL\gamma(c,n)}\right),
\]
where
\[
\nu = n\sup_{\mathcal{K} \subseteq\{1, \dots, n\}}\frac{1}{|\mathcal{K}|}\max\left\{\left\|\mathbb{E}\left[\left(\sum_{i \in \mathcal{K}}M_i\right)\left(\sum_{i \in \mathcal{K}}M_i\right)^{\top}\right]\right\|_{\mathrm{op}}, \left\|\mathbb{E}\left[\left(\sum_{i \in \mathcal{K}}M_i\right)^{\top}\left(\sum_{i \in \mathcal{K}}M_i\right)\right]\right\|_{\mathrm{op}}\right\}
\]
and
\[
\gamma(c,n) = \frac{\log n}{\log 2}\max\left\{2, \frac{32\log n}{c\log 2}\right\}.
\]
\end{theorem}
Note that Theorem~1 of \cite{banna2016bernstein} only considers symmetric matrices. To obtain \Cref{lemma:matrix_bernstein_dep} for general cases, we the Hermitian dilation and the properties of block anti-diagonal matrices.  The proof is similar to that of \Cref{thm:MatrixBernstein_PPP} and thus is omitted.
\begin{definition}[Hermitian dilation, \cite{tropp2015introduction}]\label{def:dilation}
    Consider an asymmetric matrix $A \in \mathbb{R}^{d_1 \times d_2}$, its Hermitian dilation is defined as
\begin{equation*}
    \overline{A} = 
    \begin{bmatrix}
    0 & A\\
    A^{\top} & 0
    \end{bmatrix} \in \mathbb{R}^{(d_1+d_2) \times (d_1+d_2)}.
\end{equation*}
\end{definition}
\begin{remark}
    By construction $\overline{A}$ is a block anti-diagonal matrix, we have $\lambda_{\max}(\overline{A}) = \Vert \overline{A} \Vert_{\op} = \Vert A \Vert_{\op}$.
\end{remark}

\begin{lemma}[Lieb’s Theorem, Theorem~6 of \cite{lieb1973convex}]\label{lemma:Lieb}
    Fix an Hermitian matrix $H$ with dimension $d$. The function
    \[A \to \tr\exp(H + \log A),\]
    is a concave map on the convex cone of $d \times d$ positive-definite matrices.
\end{lemma}

\section{Technical tools for point processes}
The next lemma is Campbell's Theorem, a classical result for general spatial point processes \cite[see e.g.~Theorem~2.2 of][]{baddeley2007spatial}.
\begin{lemma}[Campbell's Theorem for spatial point processes]\label{lemma:campbell}
Let $N$ be a spatial point process in a compact space $\mathbb{X}$ with the intensity function $\lambda^*$.  For all measurable function $h: \mathbb{X} \to \mathbb{R}$, we have
    \[
    \mathbb{E}\left[\int_{\mathbb{X}} h(x)\d N(x)\right] = \mathbb{E}\left[\sum_{u \in N} h(u)\right] = \int_{\mathbb{X}}h(x)\lambda^*(x)\d x. 
    \]
\end{lemma}

\bigskip
Suppose we observe $N^{(1)},\dots,N^{(n)}$ i.i.d.\ PPPs with common intensity $\lambda$. Let $P_\lambda$ denote the marginal law of one PPP and $P_\lambda^{(n)} = P_\lambda^{\otimes n}$ denote the joint law.

\begin{theorem}[Kullback–Leibler divergence for Poisson point processes]\label{thm:KL}
Assume $\lambda,\mu$ are such that $P_\lambda\ll P_\mu$. Then the KL divergence for PPP laws $P_{\lambda}$ and $P_{\mu}$ is
\begin{align*}
\mathrm{KL}(P_\lambda\|P_\mu)
=
\mathbb E_\lambda\Big[\log\frac{dP_\lambda}{dP_\mu}(N)\Big] =
\int_{\mathbb X}\Big[
\lambda(x)\log\frac{\lambda(x)}{\mu(x)}-\lambda(x)+\mu(x)
\Big]\,d\upsilon(x).
\end{align*}
Moreover, for $n$ i.i.d.\ replicates,
\begin{equation*}
\mathrm{KL}(P_\lambda^{(n)}\|P_\mu^{(n)}) = n\,\mathrm{KL}(P_\lambda\|P_\mu).
\end{equation*}
\end{theorem}
\begin{proof}
    See the proof of Theorem 6 in \cite{leskela2024information}.
\end{proof}

\begin{theorem}[Matrix Bernstein's inequality for Poisson point processes]
\label{thm:MatrixBernstein_PPP} 
Let $N$ be an inhomogeneous Poisson point process, with intensity function $\lambda: \mathbb{X} \to \mathbb{R}_+$, in a compact subset $\mathbb{X} \subset \mathbb{R}^D$ for some $D \in \mathbb{Z}_+$, such that $\|\lambda\|_{\infty}< \infty$.  Let $F: \mathbb{X} \to \mathbb{R}^{d_1 \times d_2}$ be a matrix-valued, continuous and measurable function. Suppose that $\sup_{x \in \mathbb{X}}\|F(x)\|_{\op} \leq L < \infty$. Define the variance statistics as 
$$ \nu = \max\left\{\left\|\int_{\mathbb{X}}F(x)(F(x))^{\top} \lambda(x) \d x\right\|_{\op}, \left\|\int_{\mathbb{X}}(F(x))^{\top}F(x) \lambda(x) \d x\right\|_{\op} \right\}.$$
For all $t \geq 0$, we have
\begin{align*}
    \mathbb P\left(\left\|\sum_{X \in N}F(X) - \int_{\mathbb{X}}F(x)\lambda(x)\d x\right\|_{\op} \geq t \right) \leq (d_1 + d_2)\exp\left(-\frac{t^2/2}{\nu + Lt/3}\right).
\end{align*}
In particular, for all $a \geq 2$, with probability at least $1 - 2(\max\{d_1,d_2\})^{1-a}$, we have 
\begin{align*}
    \left\|\sum_{X \in N}F(X) - \int_{\mathbb{X}}F(x)\lambda(x)\d x\right\|_{\op} \leq \sqrt{2a\nu\log(d_1 + d_2)} + \frac{2a}{3}L\log(d_1 + d_2).
\end{align*}
\end{theorem}
\begin{proof}
We first consider the symmetric case; i.e.~$F$ is a symmetric matrix-valued function, and $d_1 = d_2 = d$. Results for the general case are derived based on that of the symmetric case and the Hermitian dilation (\Cref{def:dilation}).
\\
\textbf{The symmetric case.}
\
\\
\textbf{Step~1.}
We construct a sequence of piecewise constant matrix-valued function $F_n$, which converges uniformly to $F$ in $\|\cdot\|_{\op}$.
\
\\
We partition the compact space $\mathbb{X}$ into disjoint subsets $\{A_s^{(n)}\}_{s = 1}^n$, such that the diameter of each $A_s^{(n)}$ is at most $\delta_n$, where $\delta_n \to 0$ as $n \to \infty$. Define the piecewise constant function
\[
F_n(x) = F_s^{(n)} = F(x_s^{(n)}) \;\; \text{for} \;\; x\in A_s^{(n)},
\]
where $x_s^{(n)} \in A_s^{(n)}$ is the midpoint of $A_s^{(n)}$.
Since $F$ is continuous on a compact space $\mathbb{X}$, we have that $F$ is uniformly continuous on $\mathbb{X}$. Due to the uniform continuity, we have that $\forall \epsilon > 0$, there exists $\delta > 0$ such that for all $x, y \in \mathbb{X}$, $\|x-y\|_2 < \delta$ implies that $\|F(x) - F(y)\|_{\op} < \epsilon$.  Note that
\begin{align*}
    \sup_{x \in \mathbb{X}}\|F_n(x) - F(x)\|_{\op} =& \max_{s = 1}^n\sup_{x \in A_s^{(n)}}\|F_n(x) - F(x)\|_{\op}\\
    =& \max_{s = 1}^n\sup_{x \in A_s^{(n)}}\|F(x_s^{(n)}) - F(x)\|_{\op}\\
    <& \epsilon,
\end{align*}
where the inequality follows if $n$ is large enough such that $\delta_n \leq \delta$. Therefore, as $n \to \infty$
\begin{align*}
    \sup_{x \in \mathbb{X}}\|F_n(x) - F(x)\|_{\op} \to 0.
\end{align*}
\
\\
\textbf{Step~2.}
Define
\begin{align}\label{eq:concen_Poisson_proof_sigma_n}
\Sigma_n = \sum_{X \in N}F_n(X) - \int_{\mathbb{X}}F_n(x)\lambda(x)\d x &= \sum_{s = 1}^n\left\{F_s^{(n)} \cdot N(A_s^{(n)}) - \int_{A_s^{(n)}}F_n(x)\lambda(x)\d x \right\}\nonumber\\
&=: \sum_{s = 1}^nY_s.
\end{align}
Since $N(A_s^{(n)})$ are counts of disjoint regions in a Poisson point process, $\{Y_s\}_{s = 1}^n$ are independent random matrices. For all $t \in \mathbb{R}$, the Laplace transform of $\Sigma_n$ is 
\begin{align*}
    \mathbb{E}\tr \exp\left(t\Sigma_n\right) =& \mathbb{E}\tr \exp\left(t\sum_{s = 1}^n Y_s\right) = \mathbb{E}\left\{\mathbb{E}\left[\tr \exp\left(tY_1 + t\sum_{s = 2}^n Y_s\right)\bigg| \{Y_s\}_{s = 2}^n\right]\right\}\\
    \le& \mathbb{E}\left\{\tr \exp\left(t\sum_{s = 2}^n Y_s + \log \mathbb{E}\left[e^{Y_1}\big| \{Y_s\}_{s = 2}^n\right]\right)\right\}\\
    =& \mathbb{E}\left[\tr \exp\left(t\sum_{s = 2}^n Y_s + \log \mathbb{E}e^{Y_1}\right)\right]\\
    =& \mathbb{E}\left\{\mathbb{E}\left[\tr \exp\left(tY_2+ t\sum_{s = 3}^n Y_s + \log \mathbb{E}e^{Y_1}\right)\bigg| \{Y_s\}_{s = 3}^n \right]\right\}\\
    \le& \mathbb{E}\left[\tr \exp\left(t\sum_{s = 3}^n Y_s + \sum_{k = 1}^2\log \mathbb{E}e^{Y_k}\right)\right]\\
    &\cdots\\
    \leq& \tr \exp\left(\sum_{s = 1}^n \log \mathbb{E} e^{t\{F_s^{(n)} \cdot N(A_s^{(n)}) - \int_{A_s^{(n)}}F_n(x)\lambda(x)\d x\}}\right),
\end{align*}
where the first inequality follows from Jensen's inequality and the concavity due to Lieb’s Theorem (\Cref{lemma:Lieb}), the third equality follows from the independence of $\{Y_s\}_{s = 1}^n$, and the last inequality follows from the iterative use of the above arguments.
Note that $N(A_s^{(n)})$ is a Poisson random variable with parameter $m_s^{(n)} = \mathbb{E}[N(A_s^{(n)})]$. We have
\begin{align*}
    \mathbb{E} e^{tF_s^{(n)} \cdot N(A_s^{(n)})} =& \sum_{k = 0}^{\infty}e^{ktF_s^{(n)}} \cdot e^{-m_s^{(n)}}\cdot \frac{(m_s^{(n)})^k}{k!} = e^{-m_s^{(n)}}\sum_{k = 0}^{\infty}\frac{(m_s^{(n)} e^{tF_s^{(n)}})^k}{k!} = e^{-m_s^{(n)}}\exp\left(m_s^{(n)} e^{tF_s^{(n)}}\right)\\
    =& \exp(-m_s^{(n)})\cdot I_d \cdot\exp\left(m_s^{(n)} e^{tF_s^{(n)}}\right) = \exp(-m_s^{(n)}I_d)\exp\left(m_s^{(n)} e^{tF_s^{(n)}}\right)\\
    =& \exp\left(m_s^{(n)} e^{tF_s^{(n)}} -m_s^{(n)}I_d \right),
\end{align*}
where the first equality is by definition, the second, third and sixth equalities follow from the properties of matrix exponential, i.e.~$e^A\cdot e^B = e^{A+B}$ if $AB = BA$, for $A,B \in \mathbb{R}^{d\times d}$. The fifth equality follows from the properties of matrix functions \cite[see e.g.~Definition 2.1.2 of][]{tropp2015introduction}. 
We then have
\[
\mathbb{E} e^{t\{F_s^{(n)} \cdot N(A_s^{(n)}) - \int_{A_s^{(n)}}F_n(x)\lambda(x)\d x\}} = \exp\left(m_s^{(n)} e^{tF_s^{(n)}} -m_s^{(n)}I_d - t\int_{A_s^{(n)}}F_n(x)\lambda(x)\d x \right).
\]
By the properties of matrix logarithm, i.e.~$\log(e^A) = A$ for all Hermitian matrix $A$, we have
\begin{align*}
    \log\mathbb{E} e^{t\{F_s^{(n)} \cdot N(A_s^{(n)}) - \int_{A_s^{(n)}}F_n(x)\lambda(x)\d x\}} 
    = m_s^{(n)} \left( e^{tF_s^{(n)}} - I_d\right) - t\int_{A_s^{(n)}}F_n(x)\lambda(x)\d x.
\end{align*}
Thus,
\begin{align*}
    \mathbb{E}\tr \exp\left(t\Sigma_n\right)
    \leq& \tr \exp\left(\sum_{s = 1}^n \left\{m_s^{(n)} \left( e^{tF_s^{(n)}} - I_d\right) - t\int_{A_s^{(n)}}F_n(x)\lambda(x)\d x\right\}\right)\\
    =& \tr \exp\left(\sum_{s = 1}^n \int_{A_s^{(n)}}  \left( e^{tF_n(x)} - I_d - tF_n(x)\right) \lambda(x)\d x\right)\\
    =& \tr \exp\left( \int_{\mathbb{X}}  \left( e^{tF_n(x)} - I_d  - tF_n(x)\right) \lambda(x) \d x\right),
\end{align*}
where the first equality follows from the definitions of $F_n$ and $m_s^{(n)}$.
Define $$g(x) = \frac{e^{tx} - 1 - tx}{x^2}.$$
Note that for $t > 0$, $g(x)$ is positive and monotonically increasing.
Let $\sup_{x \in \mathbb{X}}\|F_n(x)\|_{\op} = L_n$.  We have
\begin{align*}
    e^{tF_n(x)} - I_d - tF_n(x) =& F_n(x) \cdot g(F_n(x)) \cdot F_n(x) \preceq g(L_n) \cdot (F_n(x))^2,
\end{align*}
where $A \preceq B$ represents that the matrix $B-A$ is positive semidefinite. 
For all $0 < t < 3/L_n$, we have
\begin{align*}
    g(L_n) = \frac{e^{tL_n} - 1 - tL_n}{L_n^2} = L_n^{-2}\sum_{k = 2}^{\infty}\frac{(tL_n)^k}{k!} \leq \frac{t^2}{2}\sum_{k = 2}^{\infty}\frac{(tL_n)^{k-2}}{3^{k-2}} = \frac{t^2/2}{1 - tL_n/3}.
\end{align*}
For all $0 < t < 3/L_n$, we have
\begin{align*}
    e^{tF_n(x)} - I_d - tF_n(x) \preceq \frac{t^2/2}{1 - tL_n/3} \cdot (F_n(x))^2.
\end{align*}
Define $\nu_n = \|\int_{\mathbb{X}}(F_n(x))^2 \lambda(x) \d x\|_{\op}$.    For all $0 < t < 3/L_n$, 
\begin{align*}
    \mathbb{E}\tr \exp\left(t\Sigma_n\right)
    \leq& \tr \exp\left( \int_{\mathbb{X}}  \left( e^{tF_n(x)} - I  - tF_n(x)\right) \lambda(x) \d x\right)\\
    \leq& \tr \exp\left( \frac{t^2/2}{1 - tL_n/3} \cdot\int_{\mathbb{X}}(F_n(x))^2 \lambda(x) \d x\right)\\
    \leq& d\exp\left( \frac{t^2/2}{1 - tL_n/3} \cdot \nu_n\right),
\end{align*}
where the last inequality follows from the properties of matrix functions \cite[see e.g.~Definition 2.1.2 of][]{tropp2015introduction}. 
By the matrix Chernoff inequality
\begin{align*}
    \mathbb{P}\left(\left\|\Sigma_n\right\|_{\op} \geq u\right) \leq& \inf_{t > 0}\frac{\mathbb{E}\tr \exp\left(t\Sigma_n\right)}{e^{tu}} \leq \inf_{0 < t < 3/L_n}\frac{\mathbb{E}\tr \exp\left(t\Sigma_n\right)}{e^{tu}}\\
    \leq& \inf_{0 < t < 3/L_n}d\exp\left( \frac{t^2/2}{1 - tL_n/3} \cdot \nu_n - tu\right)\\
    \leq& d\exp\left( -\frac{u^2/2}{\nu_n + uL_n/3}\right),
\end{align*}
where the last inequality follows by setting $t = u/(\nu_n + uL_n/3) = 3/(3\nu_n/u + L_n) < 3/L_n$.
Recall that $L_n = \sup_{x \in \mathbb{X}}\|F_n(x)\|_{\op}$ and $L = \sup_{x \in \mathbb{X}}\|F(x)\|_{\op}$.  By construction, for all $x \in \mathbb{X}$, we can find $x_s^{(n)}$ such that $F_n(x) = F(x_s^{(n)})$. Thus, we have $L_n \leq L$.  Moreover, since $$\sup_{x \in \mathbb{X}}\|F_n(x)\|_{\op} = L_n \leq L < \infty$$ and $$\sup_{x \in \mathbb{X}}\|(F_n(x))^2\|_{\op} \leq L^2,$$ by the dominated convergence theorem, we have $\lim_{n \to \infty}\nu_n = \nu$.  In other words, for all $\epsilon >0$, there exists $n_0$ such that for all $n \geq n_0$, $|\nu_n - \nu| < \epsilon$.  Consequently, due to the monotonicity, we have
\begin{align*}
    \mathbb{P}\left(\left\|\Sigma_n\right\|_{\op} \geq u\right) \leq  d\exp\left( -\frac{u^2/2}{\nu + \epsilon + uL/3}\right),
\end{align*}
for all $n \geq n_0$.
\
\\
\textbf{Step~3.} Recall $\Sigma_n$ defined in \eqref{eq:concen_Poisson_proof_sigma_n} and define
\begin{align*}
\Sigma = \sum_{X \in N}F(X) - \int_{\mathbb{X}}F(x)\lambda(x)\d x.
\end{align*}
Since $F_n$ converges to $F$ pointwisely and $L_n \leq L < \infty$, by the dominated convergence theorem, we have 
\[
\left\|\int_{\mathbb{X}}F_n(x)\lambda(x)\d x - \int_{\mathbb{X}}F(x)\lambda(x)\d x\right\|_{\op} \to 0.
\]
Moreover, we have that as $n \to \infty$
\[
\left\| \sum_{X \in N}F_n(X) - \sum_{X \in N}F(X)\right\|_{\op} \leq N(\mathbb{X})\cdot \sup_{x \in \mathbb{X}}\left\| F_n(x) - F(x)\right\|_{\op} \overset{a.s.}{\to} 0,
\]
since $N(\mathbb{X}) < \infty$ a.s.~for a compact space $\mathbb{X}$ and $\sup_{x \in \mathbb{X}}\|F_n(x) - F(x)\|_{\op} \to 0$.  Thus, we have
\[
\left\| \Sigma_n \right\|_{\op} \overset{a.s.}{\to} \left\| \Sigma\right\|_{\op},
\]
which implies weak convergence. By the Portmanteau lemma \cite[see e.g.~Lemma 2.2 in][]{van2000asymptotic}, we have
\begin{align*}
    \mathbb{P}\left(\left\|\Sigma\right\|_{\op} \geq u\right) \leq \liminf_{n \to \infty}\mathbb{P}\left(\left\|\Sigma_n\right\|_{\op} \geq u\right) \leq  d\exp\left( -\frac{u^2/2}{\nu + \epsilon + uL/3}\right).
\end{align*}
Since we can set $\epsilon$ arbitrarily small by choosing $n$ sufficiently large, we conclude the proof for the symmetric case.
\\
\textbf{The general case.}  We consider 
the general case, where $F: \mathbb{X} \to \mathbb{R}^{d_1 \times d_2}$ is an asymmetric matrix-valued function.  Define the Hermitian dilation as $\overline{F}: \mathbb{X} \to \mathbb{R}^{(d_1 + d_2) \times (d_1 + d_2)}$ (see \Cref{def:dilation}) and
\begin{align*}
\overline{\Sigma} = \sum_{X \in N}\overline{F}(X) - \int_{\mathbb{X}}\overline{F}(x)\lambda(x)\d x.
\end{align*}  
For all matrix $A$, its Hermitian dilation $\overline{A}$ is a block anti-diagonal matrix, and we have $\lambda_{\max}(\overline{A}) = \Vert \overline{A} \Vert_{\op} = \Vert A \Vert_{\op}$.
Note that by construction $\overline{F}$ is a block anti-diagonal matrix-valued function. Thus, $$\lambda_{\max}(\overline{\Sigma}) = \Vert \overline{\Sigma} \Vert_{\op} = \Vert \Sigma \Vert_{\op},$$
$$\sup_{x \in \mathbb{X}}\|\overline{F}(x)\|_{\op} = \sup_{x \in \mathbb{X}}\|F(x)\|_{\op} = L,$$
and by the arguments in Section~2.2.8 of \cite{tropp2015introduction},
$$\left\|\int_{\mathbb{X}}(\overline{F}(x))^2 \lambda(x) \d x\right\|_{\op} = \max\left\{\left\|\int_{\mathbb{X}}F(x)(F(x))^{\top} \lambda(x) \d x\right\|_{\op}, \left\|\int_{\mathbb{X}}(F(x))^{\top}F(x) \lambda(x) \d x\right\|_{\op} \right\} = \nu.$$
Finally, applying the results for the symmetric case,
\begin{align*}
    \mathbb{P}\left(\left\|\Sigma\right\|_{\op} \geq u\right) = \mathbb{P}\left(\left\|\overline{\Sigma}\right\|_{\op} \geq u\right) \leq  (d_1+d_2)\exp\left( -\frac{u^2/2}{\nu + uL/3}\right).
\end{align*}
\end{proof}

\begin{corollary}[Matrix Bernstein's inequality for Poisson point processes]
\label{coro:MatrixBernstein_PPP} 
Let $\{N^{(i)}\}_{i = 1}^n$ be a set of i.i.d.~inhomogeneous Poisson point processes, with intensity function $\lambda: \mathbb{X} \to \mathbb{R}_+$, in a compact subset $\mathbb{X} \subset \mathbb{R}^D$ for some $D \in \mathbb{Z}_+$. Let $F: \mathbb{X} \to \mathbb{R}^{d_1 \times d_2}$ be a matrix-valued, continuous and measurable function. Suppose that $\sup_{x \in \mathbb{X}}\|F(x)\|_{\op} \leq L < \infty$. Define the matrix variance statistics as 
$$ \nu = n \max\left\{\left\|\int_{\mathbb{X}}F(x)(F(x))^{\top} \lambda(x) \d x\right\|_{\op}, \left\|\int_{\mathbb{X}}(F(x))^{\top}F(x) \lambda(x) \d x\right\|_{\op} \right\}.$$
We have for all $t \geq 0$,
\begin{align*}
    \mathbb P\left(\left\|\sum_{i = 1}^n\sum_{X \in N^{(i)}}F(X) - n\int_{\mathbb{X}}F(x)\lambda(x)\d x\right\|_{\op} \geq t \right) \leq (d_1 + d_2)\exp\left(-\frac{t^2/2}{\nu + Lt/3}\right).
\end{align*}
In particular, for all $a \geq 2$, with probability at least $1 - 2(\max\{d_1,d_2\})^{1-a}$, we have 
\begin{align*}
    \left\|\sum_{i = 1}^n\sum_{X \in N^{(i)}}F(X) - n\int_{\mathbb{X}}F(x)\lambda(x)\d x\right\|_{\op} \leq \sqrt{2a\nu\log(d_1 + d_2)} + \frac{2a}{3}L\log(d_1 + d_2).
\end{align*}
\end{corollary}
\begin{proof}
    Let $N = \bigcup_{i = 1}^nN^{(i)}$, and then by the infinite divisibility of the Poisson point process, $N$ is a Poisson point process with intensity function $n\lambda$.  Applying \Cref{thm:MatrixBernstein_PPP} on $N$ concludes the proof.
\end{proof}

\section{Technical tools for compact operators on Hilbert spaces}\label{sec:HilbertSpace}

 \begin{lemma}[Lemma 14 of \cite{khoo2024nonparametric}]
    \label{LemmaDaren}  Let $\mathcal{W}$ and $\mathcal{W}^{\prime}$ be two separable Hilbert spaces. Suppose $A$ and $B$ are two compact operators from $\mathcal{W} \otimes \mathcal{W}^{\prime} \rightarrow \mathbb{R}$. For all $k \in \mathbb{N}_+$, we have
$$
\left|\sigma_k(A+B)-\sigma_k(A)\right| \leq\|B\|_{\mathrm{op}} .
$$
 \end{lemma}

\begin{lemma}[Mirsky's singular value inequality of \cite{mirsky1960symmetric}]\label{thm:Mirsky}
For all matrices $A, B \in \mathbb{R}^{m\times n}$, let $A = V_1\Sigma(A)W_1^\top$ and $B = V_2\Sigma(B)W_2^\top$ be the full SVDs of $A$ and $B$, respectively. Note that $\Sigma(A),\Sigma(B) \in \mathbb{R}^{m \times n}$ are non-negative (rectangular) diagonal matrices whose diagonal entries are the non-increasingly ordered singular values of $A$ and $B$, respectively. We have
\begin{equation}
    \|\Sigma(A) - \Sigma(B)\| \leq \|A - B\|
\end{equation}
for all unitarily invariant norm $\|\cdot\|$ on $\mathbb{R}^{m\times n}$. 
\end{lemma}

\begin{lemma}[Mirsky's inequality for compact operators on Hilbert spaces]\label{lem:Mirsky_Hilbert}
    Suppose \( A \) and \( B \) are two compact operators in \( \mathcal{W} \otimes \mathcal{W}^{\prime} \), where \( \mathcal{W} \) and \( \mathcal{W}^{\prime} \) are two separable Hilbert spaces. Let \( \{\sigma_k(A)\}_{k=1}^\infty \) be the singular values of \( A \) in decreasing order, and \( \{\sigma_k(B)\}_{k=1}^\infty \) be the singular values of \( B \) in decreasing order. We have
\[
\sum_{k=1}^\infty (\sigma_k(A) - \sigma_k(B))^2 \leq \|A - B\|_{\mathrm{F}}^2 = \sum_{k=1}^\infty \sigma_k^2(A - B).
\]
\end{lemma}
\begin{proof}
    Let \( \{\phi_i\}_{i=1}^\infty \) and \( \{\phi_i'\}_{i=1}^\infty \) be the orthogonal basis of \( \mathcal{W} \) and \( \mathcal{W}' \). Let
\[
\mathcal{W}_j = \operatorname{span}(\phi_i: i \in [j]) \quad \text{and} \quad \mathcal{W}_j' = \operatorname{span}(\phi_i^{\prime}: i \in [j]).
\]
Let
\[
A_j = A \cdot \mathcal{P}_{\mathcal{W}_j} \cdot \mathcal{P}_{\mathcal{W}_j'} \quad \text{and} \quad B_j = B \cdot \mathcal{P}_{\mathcal{W}_j} \cdot \mathcal{P}_{\mathcal{W}_j'},
\]
where $\mathcal{P}_{\mathcal{W}_j}$ denotes the orthogonal projection onto $\mathcal{W}_j$, and similarly for $\mathcal{P}_{\mathcal{W}^{\prime}_j}$.
Since both \( A \) and \( B \) are compact, let \( n \) be sufficiently large such that for all \( j \geq n \),
\[
\|A - A_j\|_{\mathrm{F}} \leq \epsilon \quad \text{and} \quad \|B - B_j\|_{\mathrm{F}} \leq \epsilon.
\]
It follows that
\[
\sqrt{\sum_{k=1}^\infty (\sigma_k(A) - \sigma_k(B))^2} = \sqrt{\sum_{k=1}^\infty (\sigma_k(A) - \sigma_k(A_j) + \sigma_k(A_j) - \sigma_k(B_j) + \sigma_k(B_j) - \sigma_k(B))^2}.
\]
By the triangle inequality, this is
\[
\leq \sqrt{\sum_{k=1}^\infty (\sigma_k(A) - \sigma_k(A_j))^2} + \sqrt{\sum_{k=1}^\infty (\sigma_k(A_j) - \sigma_k(B_j))^2} + \sqrt{\sum_{k=1}^\infty (\sigma_k(B_j) - \sigma_k(B))^2}.
\]
By the definitions of $A_j$ and $B_j$, we have
\[
\sqrt{\sum_{k=1}^\infty (\sigma_k(A) - \sigma_k(A_j))^2} = \|A - A_j\|_{\mathrm{F}} \leq \epsilon \quad \text{and} \quad \sqrt{\sum_{k=1}^\infty (\sigma_k(B_j) - \sigma_k(B))^2} = \|B - B_j\|_{\mathrm{F}} \leq \epsilon.
\]
In addition, both \( A_j \) and \( B_j \) can be viewed as finite-dimensional matrices of size \( j \times j \). So for \( k > j \),
\[
\sigma_k(A_j) = \sigma_k(B_j) = 0.
\]
By the finite-dimensional Mirsky's inequality (\Cref{thm:Mirsky}),
\begin{align*}
&\sqrt{\sum_{k=1}^\infty (\sigma_k(A_j) - \sigma_k(B_j))^2} = \sqrt{\sum_{k=1}^j (\sigma_k(A_j) - \sigma_k(B_j))^2}\\
\leq& \sqrt{\sum_{k=1}^j \sigma_k(A_j - B_j)^2} = \|A_j - B_j\|_{\mathrm{F}} \leq \|A - B\|_{\mathrm{F}} + \|A - A_j\|_{\mathrm{F}} + \|B - B_j\|_{\mathrm{F}} \leq 2\epsilon + \|A - B\|_{\mathrm{F}}.
\end{align*}
Therefore,
\[
\sqrt{\sum_{k=1}^\infty (\sigma_k(A) - \sigma_k(B))^2} \leq \|A - B\|_{\mathrm{F}} + 2\epsilon.
\]
Since $\epsilon$ is arbitrary, we can make $\epsilon$ arbitrarily small by choosing sufficiently large $j$ in our approximations. Taking $\epsilon \to 0$, concluds the proof.

\end{proof}

\end{document}